%% file: l2l2-foreach.tex
\documentclass[11pt]{article}

\pdfoutput=1
\usepackage{fullpage}              

\usepackage{times}

\usepackage{color} 

\usepackage{graphicx,graphics}
\usepackage{amsmath}
\usepackage{amsmath,amscd}
\usepackage{amsfonts}  
\usepackage{amssymb}    
\usepackage{amsthm}
\usepackage{wrapfig}
\usepackage{subfigure}
\usepackage{hyperref}

\usepackage{graphicx}

\usepackage{todonotes}

\usepackage{enumerate}

\newcommand{\comment}[1]{}   

\def\pr{\mathop{\bf Pr}\limits}

\renewcommand{\leq}{\leqslant}
\renewcommand{\geq}{\geqslant}
\renewcommand{\le}{\leqslant}
\renewcommand{\ge}{\geqslant}
\newcommand{\eps}{\varepsilon}
\renewcommand{\epsilon}{\varepsilon}
\newcommand{\mv}[1]{{\mathbf {#1}}}
\newcommand{\vnote}[1]{}

\newcommand{\F}{\mathbb{F}}
\newcommand{\R}{\mathbb{R}}

\newcommand{\rs}{{\rm RS}}

\newcommand{\myvec}[1]{\mathbf{#1}}

\newcommand{\vzero}{\myvec{0}}

\newcommand{\vu}{\myvec{u}}
\newcommand{\vv}{\myvec{v}}
\newcommand{\vn}{\myvec{n}}
\newcommand{\vw}{\myvec{w}}

\newcommand{\vx}{\myvec{x}}

\newcommand{\vy}{\myvec{y}}
\newcommand{\vz}{\myvec{z}}
\newcommand{\vr}{\myvec{r}}

\newcommand{\X}{\mathcal{X}}
\newcommand{\Y}{\mathcal{Y}}

\newtheorem{theorem}{Theorem}[section]
\newtheorem{defn}[theorem]{Definition}
\newtheorem{cor}[theorem]{Corollary}

\newtheorem{prop}[theorem]{Proposition}
\newtheorem{lemma}[theorem]{Lemma}

\newtheorem{remark}[theorem]{Remark}

\newcommand{\M}{\mathcal{M}}
\newcommand{\B}{\mathcal{B}}
\newcommand{\light}{\mathcal{L}}
\newcommand{\nbr}{\Gamma}
\newcommand{\edg}{\mathcal{E}}

\newcommand{\supp}{\mathrm{supp}}

\newcommand{\splt}{\mathrm{split}}
\newcommand{\lw}{\mathrm{LW}}

\newcommand{\badone}{{\sc Bad-1}}
\newcommand{\badtwo}{{\sc Bad-2}}
\newcommand{\badthree}{{\sc Bad-3}}
\newcommand{\badfour}{{\sc Bad-4}}

\newcommand{\ltlt}{\ell_2/\ell_2}
\newcommand{\lolo}{\ell_1/\ell_1}

\newcommand{\poly}{\mathrm{poly}}
\newcommand{\N}{\mathcal{N}}

\newcommand{\algo}{A}

\newcommand{\deter}[1]{\Psi(#1)}
\newcommand{\ndeter}[1]{\Pi(#1)}

\newcommand{\Mlb}{\zeta^{-6}\eta^{-2}\cdot k^{\frac{1-\alpha}{1-2\alpha}}\cdot (\log(N/k))^2}

\newcommand{\schemeo}{{\sc Scheme-}$1$}
\newcommand{\schemet}{{\sc Scheme-}$2$}

\newcommand{\E}{\mathbb E}
\newcommand{\Rdist}{\mathcal R}
\newcommand{\D}{\mathcal D}
\newcommand{\cost}{{\rm cost}}

\begin{document}
\title{\textbf{$\ltlt$-foreach sparse recovery with low risk}\thanks{An extended
abstract of this work will appear in the proceedings of ICALP 2013. 
ACG's work is supported in part by NSF CCF 1161233.
HQN's work
is partly supported by NSF grant CCF-1161196. 
AR is
supported by NSF CAREER grant CCF-0844796 and NSF grant CCF-1161196. 
MJS is supported in part by NSF CCF 0743372 and NSF CCF 1161233.}\\ \textsc{(Preliminary Version)}}

\author{\textsc{Anna C. Gilbert}\footnotemark[2] \and \textsc{Hung Q. Ngo}\footnotemark[3] \and \textsc{Ely Porat}\footnotemark[4] \and \textsc{Atri Rudra}\footnotemark[3] \and \textsc{Martin J. Strauss}\footnotemark[2]}

\date{\footnotemark[2]~~University of Michigan, 
{\tt \{annacg,martinjs\}@umich.edu}\\
\vspace*{3mm}
\footnotemark[3]~~University at Buffalo (SUNY),
{\tt \{hungngo,atri\}@buffalo.edu}\\
\vspace*{3mm}
\footnotemark[4]~~Bar-Ilan University,
{\tt porately@cs.biu.ac.il}
}

\maketitle

\setcounter{page}{0}
\thispagestyle{empty}

\begin{abstract}
In this paper, we consider the ``foreach'' sparse recovery problem with failure probability $p$.  The goal of which is to design a distribution over $m \times N$ matrices $\Phi$ and a decoding algorithm $\algo$ such that for every $\vx\in\R^N$, we have the following error guarantee with probability at least $1-p$
\[\|\vx-\algo(\Phi\vx)\|_2\le C\|\vx-\vx_k\|_2,\]
where $C$ is a constant (ideally arbitrarily close to $1$) and $\vx_k$ is the best $k$-sparse approximation of $\vx$.

Much of the sparse recovery or compressive sensing literature has focused on the case of either $p = 0$ or $p = \Omega(1)$.  We initiate the study of this problem for the entire range of failure probability.  Our two main results are as follows:
\begin{enumerate}
\item We prove a lower bound on $m$, the number measurements, of
$\Omega(k\log(n/k)+\log(1/p))$ for $2^{-\Theta(N)}\le p <1$. 
Cohen, Dahmen, and DeVore~\cite{CDD2007:NearOptimall2l2} prove that this bound is tight.
\item We prove nearly matching upper bounds for \textit{sub-linear} time decoding. Previous such results addressed only $p = \Omega(1)$.
\end{enumerate}

Our results and techniques lead to the following corollaries:
(i) the first ever sub-linear time decoding $\lolo$ ``forall" sparse recovery
system that requires a $\log^{\gamma}{N}$ extra factor (for some $\gamma<1$) over
the optimal $O(k\log(N/k))$ number of measurements, and
(ii) extensions of Gilbert et al.~\cite{GHRSW12:SimpleSignals} results for information-theoretically bounded adversaries.

Cohen, Dahmen, and DeVore~\cite{MR2449058} prove a $\Omega(N)$ lower bound
for the ``forall" case (i.e. $p=0$).
Our lower bound technique is inspired by their geometric approach,
and is thus very different from prior lower bound proofs using 
communication complexity. 
Our technique yields a stronger result which holds for the entire range 
of failure probability $p$, as well as provides a simpler, more 
intuitive proof of the original result by Cohen et al. 
For the upper bounds, we provide several algorithms that span the trade-offs between number of measurements and failure probability.  These algorithms include several innovations that may be of use in similar applications.  They include a new combination of the recursive constructive sparse recovery technique of Porat and Strauss~\cite{ely-martin} and the efficiently decodable list recoverable codes. These list recoverable codes focus on a different range of parameters than the ones considered in traditional coding theory. Our best parameters are obtained by considering a code defined by the algorithmic version of the Loomis-Whitney inequality due to Ngo, Porat, R\'{e} and Rudra~\cite{NPRR}.
\end{abstract}
\newpage

\section{Introduction}

In a large number of modern scientific and computational applications, we have
considerably more data than we can hope to process efficiently and more data
than is essential for distilling useful information.  Sparse signal recovery 
\cite{cs-survey} is one method for both reducing the amount of data we collect
or process initially and then, from the reduced collection of observations, 
recovering (an approximation to) the key pieces of information in the data.
Sparse recovery assumes the following mathematical model: a data point is a
vector $\vx \in \R^N$, using a matrix $\Phi$ of size $m \times N$, where $m \ll
N$, we collect ``measurements" of $\vx$ non-adaptively and linearly as $\Phi
\vx$; then, using a ``recovery algorithm'' $\algo$, we return a good approximation to $\vx$. The error guarantee must satisfy
\begin{equation}
\label{eq:sr-condn}
    \|\vx-\algo(\Phi\vx)\|_2\le C\|\vx-\vx_k\|_2,
\end{equation}
where $C$ is a constant (ideally arbitrarily close to $1$) and $\vx_k$ is the
best $k$-sparse approximation of $\vx$.  This is customarily called an {\em
$\ltlt$-error guarantee} in the literature. This paper considers the sparse
recovery problem with failure probability $p$, the goal of which is to design a
distribution over $m \times N$ matrices $\Phi$ and a decoding algorithm $\algo$
such that for every $\vx\in\R^N$, the error guarantee holds with probability at
least $1-p$. The reader is referred to \cite{cs-survey} and the references
therein for a survey of sparse matrix techniques for sparse recovery, and to \cite{SPMagazine2008} for a collection of articles (and the references therein) that emphasize the applications of sparse recovery in signal and image processing.

There are many parameters of interest in the design problem: (i) number of
measurements $m$; (ii) decoding time, i.e. runtime of algorithm $\algo$; (iii) approximation factor $C$ and (iv) failure probability $p$. We would like to minimize all the four parameters simultaneously. It turns out, however, that optimizing the failure probability $p$ can lead to wildly different recovery schemes.  Much of the sparse recovery or compressive sensing literature has focused on the
case of either $p = 0$ (which is called the ``{\em forall}'' model) or $p =
\Omega(1)$ (the ``{\em foreach}'' model).  Cohen, Dahmen, and
DeVore~\cite{MR2449058} showed a lower bound of $m = \Omega(N)$ for the number
of measurements when $p = 0$, rendering a sparse recovery system
useless as one must collect (asymptotically) as many measurements as the 
length of the original signal\footnote{For this reason, all of the forall sparse signal recovery results satisfy a different, weaker error guarantee. E.g. in the $\lolo$ forall sparse recovery we replace the condition (\ref{eq:sr-condn}) by $\|\vx-\algo(\Phi\vx)\|_1\le C\|\vx-\vx_k\|_1$.}.  Thus, algorithmically
there is not much to do in this regime.

The case of $p\ge \Omega(1)$ has resulted in much more algorithmic success.  Cand\'es and Tao showed in \cite{Candes:2006ej} that $O(k \log(N/k))$ random measurements with a polynomial time recovery algorithm are sufficient for compressible vectors and Cohen, et al. \cite{MR2449058} show that $O(k\log(N/k))$ measurements are sufficient for any vector (but the recovery algorithm given is not polynomial time).  In a subsequent paper, Cohen, et al. \cite{CDD2007:NearOptimall2l2} give a polynomial time algorithm with $O(k\log(N/k))$ measurements.  The next goal was to match the $O(k\log(N/k))$ measurements but with {\em sub-linear} time decoding. This was achieved by
Gilbert, Li, Porat, and Strauss~\cite{GLPS12:ForeachOptimal} who showed that
there is a distribution on $m \times N$ matrices with $m = O(k \log(N/k))$ and a decoding algorithm $\algo$ such that, for each $\vx \in\R^N$ the $\ltlt$-error guarantee is satisfied with probability $p = \Theta(1)$. The next natural goal was to nail down the correct dependence on $C=1+\eps$. Gilbert et al.'s result actually needs $O(\frac{1}{\eps}k\log(N/k))$ measurements. This was then shown to be tight by Price and Woodruff~\cite{PW11}.

At this point, we completely understand the problem for the case of $p=0$ or $p=\Omega(1)$. Somewhat surprisingly, there is {\em no} work that has explicitly considered the $\ltlt$ sparse recovery problem when
$0<p\le o(1)$. The main goal of this paper is to close this gap in our understanding.

Given the importance of the sparse recovery problem, we believe that it is
important to close the gap. Similar studies have been done extensively in a  closely related field: coding theory. While the model of worst-case errors pioneered by Hamming (which corresponds to the forall model) and the oblivious/stochastic error model pioneered by Shannon (which corresponds to the foreach model) are most well-known, there is a rich set of results in trying to understand the power of intermediate channels,
including the arbitrarily varying channel~\cite{avc-survey}. Another way to consider intermediate channels is to consider computationally bounded adversaries~\cite{L94}.
Gilbert et al.~\cite{GHRSW12:SimpleSignals} considered a computationally bounded
adversarial model for the sparse recovery problem in which signals are generated
neither obliviously (as in the foreach model) nor adversarially (in the forall model) in order to interpolate between the forall and foreach signal models. 
Our results in this paper imply new results for the $\lolo$ sparse recovery problem as well as the $\ltlt$ sparse recovery problem against bounded adversaries.

Our main contributions are as follows.
\begin{enumerate}
\item We prove that the number measurements has to be $\Omega(k\log(N/k)+\log(1/p))$ for $2^{-\Theta(N)}\le p <1$. 
\item We prove nearly matching upper bounds for \textit{sub-linear} time decoding. 
\item We present applications of our result to obtain 
 \begin{itemize}
   \item[(i)] the best known number
of measurements for $\lolo$ {\em forall} sparse recovery with sublinear
($\poly(k,\log{N})$) time decoding, and 
   \item[(ii)] nearly tight upper and lower bounds on the number of measurements needed to perform $\ltlt$-sparse recovery against information-theoretically bounded adversary.
 \end{itemize}
\end{enumerate}

As was motioned earlier, there are many parameters one could optimize. We will not pay very close attention
to the approximation factor $C$, other than to stipulate that $C\le O(1)$. In
most of our upper bounds, we can handle $C=1+\eps$ for an arbitrary constant
$\eps$, but optimizing the dependence on $\eps$ is beyond the scope of this
paper.

\paragraph{Lower Bound Result.} We prove a lower bound of
$\Omega(\log(1/p))$ on the number of measurements when the failure probability
satisfies $2^{-\Theta(N)}\le p <1$. (When $p\le 2^{-\Omega(N)}$, our results
imply a tight bound of $m = \Omega(N)$.) The $\Omega(\log(1/p))$ lower bound along with
the lower bound of $\Omega(k\log(N/k))$ from~\cite{PW11} implies the final form
of the lower bound claimed above. The obvious follow-up question is whether this
bound is tight. Indeed, an upper bound result Cohen, Dahmen, and
DeVore~\cite{CDD2007:NearOptimall2l2} proves that this bound is tight if we only
care about polynomial time decoding (see Theorem~\ref{thm:upper-poly}). Thus,
the interesting algorithmic question is how close we can get to this bound with
sub-linear time decoding.

\paragraph{Upper Bound Results.}
For the upper bounds, we provide several algorithms that span the trade-offs
between number of measurements and failure probability.  For completeness, we
include the running times and the space requirements of the algorithms and
measurement matrices in Table~\ref{tab:results}, which summarizes our main 
results and compares them with existing results. 

\begin{table}[h!]
	\centering
    \tiny
	\begin{tabular}{|c|c|c|c|c|c|}
		\hline
		Reference & $k$ & $m$  & $p$  & Decoding time& Space\\
		\hline
		\cite{CDD2007:NearOptimall2l2} & Any $k$ & $k\log(N/k)+\log(1/p)$ & Any $p$ & $\poly(N)$ &$\poly(N)$\\
		\hline
		\cite{GLPS12:ForeachOptimal}& Any $k$ & $k\log(N/k)$ & $p\ge \Omega(1)$ & $k\cdot \poly\log{N}$ & $k\cdot \poly\log{N}$\\
		\hline
		\cite{ely-martin}& $k\ge N^{\Omega(1)}$ & $k\log(N/k)$ & $p=(N/k)^{-k/\log^{c}{k}}$ & $k^{1+\alpha}\poly\log{N}$ & $Nk^{0.2}$\\
		                 & Any $k$ & $k\log(N/k)\log_k^c{N}$ & $p=k^{-k/\log^c{k}}$ & $k^{1+\alpha}\poly\log{N}$ & $Nk^{0.2}$\\
		\hline
		\cite{GLPS12:ForeachOptimal}& Any $k$ & $k\log(N/k)$ & $p\ge 2^{-k/\log^c{k}}$ & $k\cdot \poly\log{N}$ & $k\cdot \poly\log{N}$\\
		$+$ Lem~\ref{lem:weak-to-top}&&&&&\\
		\hline
		\hline
		This paper& $k\ge N^{\Omega(1)}$ & $k\log(N/k)$ & $p=(N/k)^{-k/\log^{c}{k}}$ & $k^{1+\alpha}\poly\log{N}$ & $k\cdot \poly\log{N}$\\
		                 & Any $k\ge \log(N/k)$ & $k\log(N/k)\log_k^{\alpha}{N}$ & $p=(N/k)^{-k/\log^c{k}}$ & $k^{2^{\alpha^{-1}}}\poly\log{N}$ & $k\cdot \poly\log{N}$\\
		\hline
	\end{tabular}
	\caption{Summary of algorithmic results. The results in~\cite{ely-martin} are for $\lolo$ forall sparse recovery but their results can be easily adapted to our setting with our proofs. $c$ is some constant $\ge 8$  and $\alpha>0$ is any arbitrary constant and we ignore the constant factors in front of all the expressions.}
	\label{tab:results}
\end{table}

We begin by first considering the most natural way to boost the failure probability of a given $\ltlt$ sparse recovery problem: 
we repeat the scheme $s$ times with independent randomness and pick the ``best" answer-- see
Appendix~\ref{app:repeat} for more details. This boosts the decoding error probability from the original $p$ to $p^{\Omega(s)}$-- though the reduction does blow up the
approximation factor by a multiplicative factor of $\sqrt{3}$.  

The above implies that if we can optimally solve the $\ltlt$ sparse recovery
problem for $p\ge (N/k)^{-k}$ (i.e. with $O(k\log(N/k))$ measurements), then we
can solve the problem optimally for smaller $p$. 
Hence, for the rest of the description we focus on the case $p\ge (N/k)^{-O(k)}$
(where the goal is to obtain $m=O(k\log(N/k))$). Note that in this case, the amplification does not help as even for $p=\Omega(1)$, previous results (e.g.~\cite{PW11}) imply that $m\ge \Omega(k\log(N/k))$. Thus, if the original decoding error probability is $p$ then to obtain the $(N/k)^{-k}$ decoding error probability implies that the number of measurements will be larger than the optimal value a factor of $k\log(N/k)/\log(1/p)$. As we will see shortly the best know upper bound can achieve $p=2^{-k}$, which implies that amplification will be larger than the optimal value of $\Omega(k\log(N/k)$ by a factor of $\log(N/k)$. In this work, we show how to achieve the same goal with an asymptotically smaller blow-up.

For $p \geq (N/k)^{-k}$, there are two related works.
The first is that of Porat and Strauss~\cite{ely-martin} who considered the
sparse recovery problem under the $\lolo$ forall guarantee. Despite the
different error guarantee, our construction is closely related to that
of~\cite{ely-martin} and our proofs imply the results for~\cite{ely-martin}
listed in Table~\ref{tab:results}. Note that the results for polynomially large
$k$ are pretty much the same except we have a better space complexity. For
general $k$, our result also has  better number of measurements and failure probability guarantee. The
second work is that of Gilbert et al.~\cite{GLPS12:ForeachOptimal}. Even though
the results in that paper are cited for $p\ge \Omega(1)$, it can be shown that
if one uses $O(k)$-wise independent random variables instead of the pair-wise
independent random variables as used in~\cite{GLPS12:ForeachOptimal}, one
can obtain a ``weak system" with failure probability $2^{-k}$. Then our ``weak
system to top level system conversion" (Lemma~\ref{lem:weak-to-top}) leads to
the result claimed in the second to last row in Table~\ref{tab:results}. 
Our results have a better failure probability at the cost of larger number 
of measurements.

It is natural to ask whether decreasing the failure probability 
(the base changed from $2$ to $(N/k)$)
is worth giving up the optimality in the number of
measurements (which is what~\cite{GLPS12:ForeachOptimal} obtains).
Note that achieving a failure probability of $(N/k)^{-k}$ is a 
very natural goal and our results are better than those 
in~\cite{GLPS12:ForeachOptimal} when we anchor on the failure probability 
goal first. Interestingly, it turns out that this difference 
is also crucial to our result for $\lolo$ sparse recovery which is discussed
next.

\paragraph{$\lolo$ forall sparse recovery.} As was mentioned earlier, our
construction is similar to that of Porat and Strauss. In fact, our techniques
also imply some results for the $\lolo$ forall sparse recovery problem. (We only
need the ``weak system" construction, after that results from~\cite{ely-martin}
can be applied.)  We highlight two results here. 

First, for $k\ge N^{\Omega(1)}$ we can get $O(k\log(N/k))$ measurements with
$k^{1+\alpha}\poly\log{N}$ decoding time and space usage, for any $\alpha>0$. 
This improves upon the $\lolo$ forall sparse recovery result from~\cite{ely-martin} with a better space usage and answers a question 
left open by~\cite{ely-martin}.

Our second result, which is for general $k$, achieves for {\em any} $\alpha>0$,
$O(k\log(N/k)\log_k^{\alpha}{N})$ number of measurements with decoding time
$k^{2^{1/\alpha}}\poly\log{N}$. This greatly improves upon the extra factor in
the number of measurements over the optimal $O(k\log(N/k)$) from something
around $\log^8_kN$ in~\cite{ely-martin} to $\log^{\alpha}_k{N}$ for arbitrary 
$\alpha>0$. Even though this ``only" shaves of log factors, no prior work 
on sublinear time decodable $\lolo$ sparse recovery scheme has been able to 
breach the ``log barrier."
There is a  folklore construction for a sub-linear time decodable matrix: 
using a ``bit-tester matrix" along with a lossless bipartite expander where 
each edge in its adjacency matrix is replaced with a random $\pm 1$ value.  
This scheme achieves $m = O(k\log(N/k) \log N)$. 
If we consider a ``weak system" {\em only} and small values of $k$, then the 
folklore construction achieves failure probability $p = (N/k)^{-k}$ which 
matches that of our construction. This along with a simple union bound 
(which we also use in our construction) and results in~\cite{ely-martin} 
implies an $\lolo$ sparse recovery system. 
However, this simple construction has an extra factor of $\log{N}$ in the number
of measurements, which our result reduces to $\log_k^{\alpha}{N}$ for any 
$\alpha>0$.\footnote{Interestingly one could use the weak system
from~\cite{GLPS12:ForeachOptimal} (with $O(k)$-wise independence) to obtain an
$\lolo$ sparse recovery system with an extra factor of $O(\log(N/k))$, which is
not that much better than that obtained by the simple bit-tester construction.}

\paragraph{Bounded Adversary Results.} We also obtain some results for $\ltlt$-sparse recovery against information-theoretically-bounded adversaries as considered by Gilbert et al.~\cite{GHRSW12:SimpleSignals}. (See Section~\ref{sec:prelims} for a formal definition of such bounded adversaries.) Gilbert et al. show that $O(k\log(N/k))$ measurements is sufficient for such adversaries with $O(\log{N})$ bits of information. Our results allow us to prove results for a general 
number of information bits bound of $s$. In particular, we observe that for such
adversaries $O(k\log(N/k)+s)$ measurements suffice. Further, if one desires
sublinear time decoding then our results in Table~\ref{tab:results} allows for a
similar conclusion but with extra $\poly\log{k}$ factors. We also observe that one needs $\tilde{\Omega}(\sqrt{s})$ many measurements against such an adversary (assuming the entries are polynomially large). 
In the final version of the paper, we will present a proof suggested to us by an anonymous reviewer that leads to the optimal $\Omega(s)$ lower bound.

\paragraph{Lower Bound Techniques.}
Our lower bound technique is inspired by the geometric approach of
Cohen et al.~\cite{MR2449058} for the $p = 0$ case.
Our bound holds for the entire range of failure probability $p$. 
Our technique also yields a simpler and more intuitive proof of Cohen et al.
result. Both results hold even for sparsity $k=1$.

The technical crux of the lower bound result in~\cite{MR2449058} for $p = 0$ 
is to show that any measurement matrix $\Phi$ with $O(N/C^2)$ rows has a null space vector 
$\vn$ that is ``non-flat,'' -- i.e. $\vn$ has one coordinate that has most of 
the mass of $\vn$.
On the other hand since $\Phi\vn=\vzero$, the decoding algorithm $\algo$ has to
output the same answer when $\vx=\vn$ and when $\vx=\vzero$. It is easy to see
that then $\algo$ does not satisfy (\ref{eq:sr-condn}) for at least one of these
two cases (the output for $\vzero$ has to be $\vzero$ while the output for
$\vn$ has to be non-flat and in particular not $\vzero$). 

To briefly introduce our technique, consider the case of $p=2^{-N}$ (where we 
want a lower bound of $m = \Omega(N)$). The straightforward extension of Cohen, et al.'s argument is to define a
distribution over, say, all the unit vectors in $\R^N$, argue that this gives a large measure of ``bad vectors,'' and then apply Yao's
minimax lemma to obtain our final result; i.e., that there are ``a lot'' of non-flat
vectors in the null space of a given matrix $\Phi$.  This argument fails because
the distribution on the bad vectors must be independent of the measurement
matrix $\Phi$ (and algorithm $\algo$) in order to apply Yao's lemma but null space vectors, of course, depend on $\Phi$. On the other
hand, if we define the ``hard" distribution to be the uniform distribution, then
the measure of null space vectors for any $m \ge 1$ is zero, and thus this
obvious generalization does not work. 

We overcome this obstacle with a simple idea. Our hard distribution is 
still the uniform distribution on the unit sphere $S^{N-1}$. We first show that
there is a region $R$ on this sphere with large measure ($\geq p$) such
that {\em all} vectors in $R$ have a {\em positive} ``spike" (large mass) at 
one particular coordinate $j^* \in [N]$. 
(The region $R$ is simply a small {\em spherical cap} above the unit vector $\mv
e_{j^*}$.)
In particular, to recover an input vector $\vv\in R$, the algorithm has to
assign a large positive mass to the $j^*$th coordinate of $A(\Phi\vv)$.
Next, by applying a certain invertible linear reflector to $R$, we can 
construct a region
$R'$ (which is also a region on the sphere, and is just a reflection of $R$) 
with the same measure satisfying the following: for each vector $\vv \in R$, the reflection $\mv v'$ of $\mv v$ ($\mv
v' \in R'$) has a {\em negative} spike at the same coordinate $j^*$;
furthermore, $\Phi\mv v = \Phi \mv v'$, which forces the algorithm $\algo$ into
a dichotomy. The algorithm can not recover both $\mv v$ and $\mv v'$ well at
once. Roughly speaking, the algorithm will be wrong with probability at least
half the total measure of $R$ and $R'$, which is $p$.
Finally, Yao's lemma completes the lower bound proof. There are some additional technical
obstacles that we need to overcome in this step---see Section~\ref{app:yao} 
for more details.

In the final version of the paper, we will present an alternate lower bound proof that was suggested to us by an anonymous reviewer.

\paragraph{Upper Bound Techniques.}
We believe that our main algorithmic contributions are the new techniques that
we introduce in this paper, which should be useful in (similar) applications (e.g. in the $\lolo$ sparse recovery problem as we have already pointed out). 

Our upper bounds follows the same outline used by Gilbert, Li, Porat and
Strauss~\cite{GLPS12:ForeachOptimal} and Porat and Strauss~\cite{ely-martin}. At
a high level, the construction follows three steps. The first step is to design
an ``identification scheme,'' which in sub-linear time computes a set
$S\subseteq [N]$ of size roughly $k$ that contains $\Omega(k)$ of the ``heavy
hitters.'' (Heavy hitters are the coordinates where if the output vector does not
put in enough mass then (\ref{eq:sr-condn}) will not be satisfied.) In the
second step, we develop a ``weak level system'' which essentially estimates the
values of coordinates in $S$. 
Finally, using a loop invariant iterative scheme, we
convert the weak system into a ``top level system,'' which is the overall system
that we want to design. (The way this iterative procedure works is that it makes
sure that after iteration $i$, one is missing only $O(k/2^i)$ heavy hitters-- so
after $\log{k}$ steps we would have recovered all of them.) The last two steps
are designed to run in time $|S|\cdot\poly\log{N}$, so if the first step runs in
sub-linear time, then the overall procedure is sub-linear\footnote{We would also like to point out that Gilbert et al.'s construction has a failure probability of $\Omega(1)$ in the very first iteration of the last step (weak to top level system conversion) and it seems unlikely that this can be made smaller without significantly changing their scheme.
}. Our main contribution is in the first step, so we will focus on the
identification part here. The second step (taking median of measurements like
Count-Sketch \cite{DBLP:conf/icalp/CharikarCF02}) is standard
\cite{IR08}. For a pictorial overview of our identification scheme, see Figure~\ref{fig:overview}. 

In order to highlight and to summarize our technical contributions, we present an overview of the scheme in~\cite{ely-martin} (when adapted to the $\ltlt$ sparse recovery problem). We focus only on the identification step. For near-linear time identification, one uses a lossless bipartite expander where each edge in the adjacency matrix is replaced by a random $\pm 1$ value. The intuition is that because of the expansion property most heavy hitters will not collide with another heavy hitter in most of the measurements it participates in. Further, the expansion property implies that the $\ell_2^2$ noise in most of the neighboring measurements will be low. (The random $\pm 1$ is a standard trick to convert this to an low $\ell_2$ noise.) Thus, if we define the value of an index to be median value of all the measurements, then we should get very good estimates for most of the heavy hitters (and in particular, we can identify them by outputting the top $O(k)$ median values). Since this step implies computing $N$ medians overall we have a near linear time computation. However, note that if we had access to a subset $S'\subseteq [N]$ that had most of the heavy hitters in it, we can get away with a run time nearly linear in $|S'|$ (by just computing the medians in $S'$).

This seems like a chicken and egg problem as the set $S'$ is what we were after to
begin with! Porat and Strauss use recursion to compute $S'$ in sub-linear time.
(The scheme was also subsequently used by Ngo, Porat and Rudra to design near
optimal sub-linear time decodable $\lolo$ forall sparse recovery schemes for non-negative signals~\cite{NPR}.) To
give the main intuition, consider the scheme that results in
$\tilde{O}(\sqrt{N})$ identification time. We think of the domain $[N]$ as
$L\times R$, where both $L$ and $R$ are isomorphic to $[\sqrt{N}]$. 
(Think of $L$ as the first $\frac{\log{N}}{2}$ bits in the $\log{N}$-bit
representation of any index in $[N]$ and $R$ to be the remaining bits.) 
If one can obtain lists $S_L\subset L$ and $S_R\subset R$ that contain the projections 
of the heavy hitters in $L$ and $R$, respectively, 
then $S_L\times S_R$ will contain all the heavy hitters, i.e.
$S'\subseteq S_L\times S_R$. (We can use the near linear time scheme to obtain
$S_L$ and $S_R$ in $\tilde{O}({\sqrt{N}})$ time in the base case. One also has
to make sure that when going from $[N]$ to a domain of size $\sqrt{N}$, not too
many heavy hitters collide. This can be done by, say, randomly permuting $[N]$
before applying the recursive scheme.) The simplest thing to do would be to set
$S'=S_L\times S_R$. However, since both $|S_L|$ and $|S_R|$ can be $\Omega(k)$,
this step itself will take $\Omega(k^2)$ time, which is too much if we are
shooting for a decoding time of $k^{1+\alpha}\poly\log{N}$ for $\alpha<1$. The
way Porat and Strauss solved this problem was to store the whole inversion map
as a table. This allowed $k\cdot\poly\log{N}$ decoding time but the scheme ended up
needing $\Omega(N)$ space overall. 

To get a running time of $k^{1+\alpha}\poly\log{N}$ one needs to apply the
recursive idea with more levels. One can think of the whole procedure as a recursion tree with $\N=O(\log_k{N})$ nodes. Unfortunately, this process introduces another technical hurdle. At each node, the expander based scheme loses some, say $\zeta$, fraction of the heavy hitters. To bound the overall fraction of lost heavy hitters, Porat and Strauss use the naive union bound of $\zeta\cdot\N$. However, we need the overall fraction of lost heavy hitters to be $O(1)$. This in turn introduces extra factors of $\log_k{N}$ in the number of measurements (resulting in the ultimate number of measurements of $k\log(N/k)\log^8_k{N}$ in~\cite{ely-martin}).

We are now ready to present the new ideas that improve upon Porat and Strauss'
solutions to solve the two issues raised above. 
Instead of dividing
$[N]$ into $[\sqrt{N}]\times [\sqrt{N}]$, we first apply a code
${\mathcal C}:[N]\rightarrow [\sqrt[b]{N}]^r$. 
(Note that the Porat Strauss construction
corresponds to the case when $r=b=2$ and ${\mathcal C}$ just ``splits'' the
$\log{N}$ bits into two equal parts.) Thus, in our recursive algorithm at the
root we will get $r$ subsets $S_1,\dots, S_r\subseteq [\sqrt[b]{N}]$ with the
guarantee that for (most) $i\in [r]$, $S_i$ contains ${\mathcal C}(j)_i$ for
most heavy hitters $j$. Thus, we need to recover the $j$'s for which the
condition in the last sentence is true. This is {\em exactly} the list recovery
problem that has been studied in the coding theory literature. (See
e.g.~\cite{atri-thesis}.) Thus, if we can design a code ${\mathcal C}$ that solves the list recovery problem very efficiently, we would solve the first problem above\footnote{We would like to point out that~\cite{NPR} also uses list recoverable codes but those codes are used in a  different context: they used it to replace expanders and further, the codes have the traditional parameters.}. For the second problem, note that since we are using a code ${\mathcal C}$, even if we only have ${\mathcal C}(j)_i\in S_i$  for say $r/2$ positions $i\in [r]$, we can recover all such indices $j$. In other words, unlike in the Porat Strauss construction where we can lose a heavy hitter even if we lose it in any of the $\N$ recursive call, in our case we only lose a heavy hitter if it is lost in {\em multiple} recursive calls. This fact allows us to do a better union bound than the naive one used in~\cite{ely-martin}.

The question then is whether there exists code ${\mathcal C}$ with the desired
properties. The most crucial part is that the code needs to have a decoding
algorithm whose running time is (near) linear in $\max_{i\in [r]}|S_i|$.
Further, we need such codes with $r=O(1)$, i.e. of constant block length independent of $\max_{i\in [r]}|S_i|$. Unfortunately, the known results on
list recovery, be it for Reed-Solomon codes~\cite{GS99} or folded Reed-Solomon
codes~\cite{GR08} do not work well in this regime---these results need $r\ge
\Omega(\max_{i\in [r]} |S_i|)$, which is way too expensive. For our setting, the
best we can do with Reed-Solomon list recovery is to do the naive thing of going
through all possibilities in $\times_{i\in [r]} S_i$. (These codes however can
correct for optimal number of errors and lead to our result in the last row of
Table~\ref{tab:results}.) Fortunately, a recent result of Ngo, Porat, R\'{e} and
Rudra~\cite{NPRR} gave an algorithmic proof of the Loomis-Whitney
inequality~\cite{LW}. The (combinatorial) Loomis-Whitney inequality has found
uses in theoretical computer science before~\cite{ITT04,LL05}. In this work, we
present the first application of the algorithmic Loomis Whitney inequality
of~\cite{NPRR} and show that it naturally defines a code $C$ with the required
(algorithmic) list recoverability. This code leads to the result in the second
to last row of Table~\ref{tab:results}. Interestingly, we get {\em optimal} weak level systems by this method. We lose in the final failure probability because of the weak level to top level system conversion.

We conclude the contribution overview by pointing out three technical aspects 
of our results. 
\begin{itemize}
\item As was mentioned earlier, we first randomly permute the columns of the matrix to make the recursion work. To complete our identification algorithm, we need to perform the inverse operation on the indices to be output. The naive way would be to use a table lookup, which will require $O(N\log{N})$ space, but would still be an improvement over~\cite{ely-martin}. However, we are able to exploit the specific nature of the recursive tree and the fact that our main results use the Reed-Solomon code and the code based on Loomis-Whitney inequality to have sub-linear space usage. 
\item In the weak level to top level system, both Gilbert et al., and Porat and Strauss decrease the parameters geometrically---however in our case, we need to use different decay functions to obtain our failure probability. 
\item Unlike the argument in~\cite{ely-martin} we explicitly use an expander while Porat and Strauss used a random graph. However, because of this,~\cite{ely-martin} need at least $N$-wise independence in their random variables to make their argument go through. Our use of expanders allows us to get away with using only $\tilde{O}(k)$-wise independence, which among others leads to our better space usage.
\end{itemize}

\section{Preliminaries}
\label{sec:prelims}

We fix notations, terminology, and concepts that will be used throughout the
paper. 

Let $[N]$ denote the set $\{1,\dots,N\}$.
Let $G:[N]\times[\ell]\rightarrow [M]$ be an $\ell$-regular bipartite graph, and
$\M_G$ be its adjacency matrix. We will often switch back and forth between the
graph $G$ and the matrix $\M_G$.

For any subset $S\subseteq [N]$, let $\nbr(S)\subseteq [M]$ denote the set of
neighbor vertices of $S$ in $G$. Further, let $\edg(S)$ denote the set of edges
incident on $S$. A bipartite graph $G:[N]\times[\ell]\rightarrow [M]$ is a
$(t,\eps)$-expander if for every subset $S\subseteq [N]$ of $|S|\le t$, we have
$|\Gamma(S)|\ge |S|\ell(1-\eps)$. Several expander properties used in our proofs
are listed in Appendix~\ref{app:exp-basics}. Appendix~\ref{app:prob-basics} has
some probability basics.

\paragraph{Sparse Recovery Basics.} For a vector $\vx = (x_i)_{i=1}^N \in\R^N$, the set of $k$
highest-magnitude coordinates of $\vx$ is denoted by $H_k(\vx)$. Such elements
are called \textit{heavy hitters}. Every element $i\in [N]\setminus H_k(\vx)$
such that $|x_i|\ge \sqrt{\frac{\zeta^2\eta}{k}}\cdot \|\vz\|_2$ will be called
a \textit{heavy tail} element. Here, $\zeta$ and $\eta$ are constants that will
be clear from context. All the remaining indices will be called
\textit{light tail} elements; let $\light$ denote the set of light tail
elements. A vector $\vw = (w_i)_{i=1}^N \in\R^N$ is called a \textit{flat tail}
if $w_i = 1/|S|$ for every non-zero $w_i$, where $S=\supp(\vw)$.

\begin{defn} A probabilistic $m\times N$ matrix $\M$ is called an $(k,C)$-approximate sparse recovery system or $(k,C)$-{\em top level system} with failure probability $p$ if there exists a {\em decoding algorithm} $\algo$ such that for every $\vx\in\R^N$, the following holds with probability at least $1-p$:
\[\|\vx-\algo(\M\vx)\|_2\le C\cdot \|\vx-\vx_{H_k(\vx)}\|_2.\]
The parameter $m$ is called the number of measurements of the system.
\end{defn}
We will also consider  $(k,C)$ $\lolo$ top level systems, which are the same as above except we have $p=0$ and the $\ell_2$ norms are replaced by $\ell_1$ norms.

\begin{defn} A probabilistic matrix $\M$ with $N$ columns is called a
$(k,\zeta,\eta)$-{\em weak identification} matrix with $(\ell,p)$-{\em
guarantee} if there is an algorithm that, given $\M\vx$ and a subset 
$S\subseteq [N]$, with probability at least $1-p$ outputs a subset 
$I\subseteq S$ such that 
(i) $|I|\le \ell$ and 
(ii) at most $\zeta k$ of the elements of $H_k(\vx)$ are not present in $I$. 
The time taken to compute $I$ will be called {\em identification time}.
\end{defn}

\begin{defn}
We will call a (random) $m\times N$ matrix $\M$ a $(k,\zeta,\eta)$ weak $\ell_2/\ell_2$ system if the following holds for any vector $\vx=\vy+\vz$ such that $|\supp(\vy)|\le k$. Given $\M\vx$ one can compute $\hat{\vx}$ such that there exist $\hat{\vy},\hat{\vz}$ that satisfy the following properties:
(1) $\vx=\hat{\vx}+\hat{\vy}+\hat{\vz}$;
(2) $|\supp(\hat{\vx})|\le O(k/\eta)$;\footnote{This part is different from the weak system in~\cite{ely-martin}, where  we have $|\supp(\hat{\vx})|\le O(k)$.}
(3) $|\supp(\hat{\vy})|\le \zeta k$;
(4) $\|\hat{\vz}\|_2\le \left(1+O\left(\eta\right)\right)\cdot \|\vz\|_2$
\end{defn}

We will also consider  $(k,\zeta,\eta)$ weak $\lolo$ systems, which is the same as above except the $\ell_2$ norms are replaced by $\ell_1$ norms (and the algorithm is deterministic).

\paragraph{Coding Basics.}

In this section, we define and instantiate some (families) of codes that we will be interested in. We begin with some basic coding definitions.
We will call a code $C:[N]\rightarrow [q]^r$ be a $(r,N)_q$-code.\footnote{We
depart from the standard convention and use the size of the code $N$ instead of
its dimension $\log_q{N}$: this makes expressions simpler later on.} Vectors in
the range of $C$ are called {\em codewords}. Sometimes we will think of 
$C$ as a subset of $[q]^r$, defined the natural way.
We will primarily be interested in \textit{list recoverable} codes. In particular,

\begin{defn} Let $N,q,r,\ell,L\ge 1$ be integers and $0\le \rho\le 1$ be a real
number. Then an $(r,N)_q$-code $C$ is called a $(\rho,\ell,L)$-list recoverable
code if the following holds. Given any collection of subsets $S_1,\dots,S_r\subseteq [q]$ such that $|S_i|\le \ell$ for every $i\in [r]$, there exists at most $L$ codewords $(c_1,\dots,c_r)\in C$ such that $c_i\in S_i$ for at least $(1-\rho)n$ indices $i\in [r]$. Further, we will call such a code recoverable in time $T(\ell,N,q)$ if all such codewords can be computed within this time upper bound.
\end{defn}

An $(r,N)_q$-code is said to be \textit{uniform} if for every $i\in [r]$ it is the case that $C(x)_i$ is uniformly distributed over $[q]$ for uniformly random $x\in [N]$.
In our construction we will require codes that are \textit{both} list recoverable and uniform. Neither of these concepts are new but our construction needs us to focus on parameter regimes that are generally not the object of study in coding theory. In particular, as in coding theory, we focus on the case where $N$ is increasing. Also we focus on the case when $q$ grows with $N$, which is also a well studied regime. However, we consider the case when $r$ is a \textit{fixed}. 

We now consider a code based on the Loomis-Whitney inequality \cite{LW}. Let
$d\ge 2$ be an integer and assume that $N$ is a power of $2$ and $\sqrt[d]{N}$
is an integer. Given $x\in [N]$ we think of it as $x = (x_1,\dots,x_d)\in
[\sqrt[d]{N}]^d$. Further, for any $i\in [d]$ define
$x_{-i}=(x_1,\dots,x_{i-1},x_{i+1},\dots,x_d)\in [\sqrt[d]{N}]^{d-1}$. Then
define $C_{\lw(d)}(x)=(x_{-1},\dots,x_{-d})$. The Loomis-Whitney inequality
implies that $C_{\lw(d)}$ is a $(0,\ell,\ell^{d/(d-1)})$-list recoverable code.
Ngo, Porat, R\'{e} and Rudra~\cite{NPRR} recently showed that the code
is list-recoverable in time $\tilde{O}(\ell^{d/(d-1)})$. 
The result implies the following:

\begin{lemma}
\label{lem:lw}
The code $C_{\lw(d)}$ is a uniform code that is $(0,\ell,\ell^{d/(d-1)})$-list recoverable code. Further, it is recoverable in $O(\ell^{d/(d-1)}\log{N})$ time.
\end{lemma}

For the sake of completeness, we prove the above via Theorem~\ref{thm:LW-d-1} (with a slightly different proof than the one from~\cite{NPRR}). Finally, we consider the well-known Reed-Solomon (RS)  codes.

\begin{lemma}
\label{lem:rs}
Let $\rho<1/2(1-b/r)$. Then
the code $C_{\rs}$ is a uniform code that is $(\rho,\ell,\ell^{r})$-list recoverable code. Further, it is recoverable in $O(\ell^{r}r^2\log^2{N})$ time.\footnote{The $r^2\log^2{N}$ factor follows from the fact that the Berlekamp Massey algorithm needs $O(r^2)$ operation over $\F_q$, each of which takes $O(\log^2{q})$ time.}
\end{lemma}

\paragraph{Bounded Adversary Model.}

We summarize the relevant definitions of (computationally-)bounded
adversaries from \cite{GHRSW12:SimpleSignals}. In this setting, Mallory
is the name of the process that generates inputs $x$ to the sparse
recovery problem. We recall two definitions for Mallory:

\begin{itemize}
\item
  \textbf{Oblivious:} Mallory cannot see the matrix $\Phi$ and generates
  the signal $x$ independent from $\Phi$. For sparse signal recovery,
  this model is equivalent to the ``foreach'' signal model.
\item
  \textbf{Information-Theoretic:} Mallory's output has bounded mutual
  information with the matrix. To cast this in a computational light, we
  say that an algorithm $M$ is
  \em ($s$-)information-theoretically-bounded \em if $M(x) = M_2(M_1(x))$, where the output of $M_1$ consists of at most $s$ bits.  This model is similar to that of the ``information bottleneck'' \cite{TishbyPereiraBialek}. 
\end{itemize}

Lemma 1 of \cite{GHRSW12:SimpleSignals} relates the
information-theoretically bounded adversary to a bound on the success
probability of an oblivious adversary. We re-state the lemma for
completeness:

\begin{lemma}
 \label{obs:bottleneck}
Pick $\ell = \ell(N)$, and fix $0 < \alpha < 1$.
Let $A$ be any randomized algorithm which takes input $x \in \{0,1\}^N$, $r \in
\{0,1\}^m$, and ``succeeds'' with probability $1-\beta$.
Then for any information theoretically bounded algorithm $M$ with space $\ell$,
$A(M(r), r)$ succeeds with probability at least 
$$\min \left\{ 1 - \alpha, \ 1 - \ell / \log(\alpha/\beta)\right\}$$
over the choices of $r$.
\end{lemma}

\section{Lower bounds}
\label{sec:lb}

\subsection{Lower bound for $\ltlt$-foreach sparse recovery with low risk}

Throughout this section, let $\Phi$ denote an $m \times N$ measurement matrix.
Without loss of generality, we will assume that all our measurement matrices
have full row rank: $\text{rank}(\Phi) = m$.
Let $\mathcal N$ be the (row) null-space of $\Phi$, and $\mathbf P$ be the 
matrix representing the orthoprojection onto $\mathcal N$.
For $j\in [N]$, let $\mathbf e_j$ denote the $j$th standard basis vector.
Note that 
$\mathbf P = \mathbf P^2 = \mathbf P^T\mathbf P$ because any orthogonal
projection matrix is symmetric and idempotent.
Hence, for any two vectors $\mathbf x, \mathbf y \in \mathbb R^N$, 
\[ \langle \mathbf{Px}, \mathbf y \rangle = \langle \mathbf x, \mathbf{Py} \rangle = \langle \mathbf x, \mathbf P^T\mathbf{Py}\rangle
 = \langle \mathbf{Px}, \mathbf {Py} \rangle.
\]

For the sake of completeness we present a simplified version of the proof of the $\Omega(N)$ lower bound from~\cite{MR2449058} for the $\ltlt$ forall sparse recovery in Appendix~\ref{app:forall}.

We first prove an auxiliary lemma which generalizes Proposition
\ref{prop:DeVore-non-flat}. 
The lemma implies that,
for any measurement matrix $\Phi$,
if $m/N$ is ``small'' then there will be ``a lot'' of unit vectors 
that are not ``flat,'' i.e. in these vectors
some coordinate $j^*$ has relatively large magnitude compared to 
all other coordinates. 

\begin{lemma}
\label{lmm:non-flat}
Let $\Phi$ be an arbitrary real matrix of dimension $m\times N$.
Then, there exists $j^* \in [N]$ such that, for any unit-length vector 
$\mathbf v \in \R^N$ we have
\begin{equation}
 \langle \mathbf{Pv}, \mathbf e_{j^*} \rangle \geq 
 \langle \mv v, \mv e_{j^*} \rangle - 
 \sqrt{1 - \langle \mv v,\mv e_{j^*}\rangle} 
 \cdot \sqrt{2m/N} - m/N.
\label{eqn:non-flat}
\end{equation}
\end{lemma}

\begin{proof}
Let $j^*$ be the coordinate
for which $\|\mv{Pe}_{j^*}\|_2^2 \geq 1-m/N$, guaranteed to exist by
Proposition \ref{prop:DeVore-non-flat}.
For any unit vector $\mathbf v$, we have
\begin{eqnarray*}
\langle \mathbf{Pv}, \mathbf e_{j^*}\rangle
&=& \langle \mv v, \mv P \mv e_{j^*}\rangle\\
&=& \langle \mv v, \mv e_{j^*} \rangle -
    \langle \mv v, \mv e_{j^*} - \mv P \mv e_{j^*}\rangle\\
&=& \langle \mv v, \mv e_{j^*} \rangle -
    \langle \mv v-\mv e_{j^*} , \mv e_{j^*} - \mv P \mv e_{j^*}\rangle
  - \langle \mv e_{j^*}, \mv e_{j^*} - \mv P \mv e_{j^*} \rangle\\
\text{(by Cauchy-Schwarz)} &\geq&
 \langle \mv v, \mv e_{j^*} \rangle -
 \|\mv v-\mv e_{j^*}\|_2 \cdot \| \mv e_{j^*} - \mv P \mv e_{j^*}\|_2
  - (1- \langle \mv e_{j^*}, \mv P \mv e_{j^*} \rangle)\\
 &=&
 \langle \mv v, \mv e_{j^*} \rangle -
 \sqrt{2 - 2 \langle \mv v,\mv e_{j^*}\rangle} \cdot
 \sqrt{1- \| \mv P \mv e_{j^*}\|_2^2}
  - (1- \| \mv P \mv e_{j^*}\|_2^2)\\
 &\geq&
 \langle \mv v, \mv e_{j^*} \rangle -
 \sqrt{1 - \langle \mv v,\mv e_{j^*}\rangle}
 \cdot \sqrt{2m/N} - m/N.
\end{eqnarray*}
\end{proof}
Next, we show that any unit vector $\mathbf v$ sufficiently close to 
$\mv e_{j^*}$
can be paired up in a 1-to-1 manner (through a reflection operator)
with a unit vector $\mathbf v'$ such that the decoding algorithm cannot work 
well on both $\mathbf v$ and $\mathbf v'$. 
Since the pairing is measure preserving, we can then infer that when $m/N$
is small the algorithm $A$ will fail with ``high'' probability. 

\begin{lemma}
Let $(\Phi, A)$ be a (deterministic) pair of measurement matrix and decoding
algorithm, 
where $\Phi$ is an $m \times N$ matrix.
Define $\delta = m/N$, and 
let $\gamma\geq 0, C\geq 1$ be arbitrary constants such that
\begin{equation}
1-2\gamma-2\delta-2\sqrt{2\gamma\delta} > \frac{C}{\sqrt{1+C^2}}.
\label{eqn:gammadelta}
\end{equation}
Let $j^* \in [N]$ be the index satisfying \eqref{eqn:non-flat} guaranteed
to exist by Lemma \ref{lmm:non-flat}.
Let $\mathbf v$ be any unit vector such that
$\langle \mathbf v, \mathbf e_{j^*}\rangle \geq 1- \gamma$,
and $\mathbf v' = (\mathbf I - 2\mathbf P)\mathbf v$ where $\mathbf I$ is 
the identity matrix. 
Then, the following two conditions cannot be true at the same time:
\begin{itemize}
 \item[(a)] $\|\mathbf v - A(\Phi \mathbf v)\|_2 \leq C \cdot \|\mathbf v-\mv
v_k\|_2$
 \item[(b)] $\|\mathbf v' - A(\Phi \mathbf v')\|_2 \leq C \cdot \|\mathbf v'-\mv
v'_k\|_2$
\end{itemize}
\label{lmm:symmetrize}
\end{lemma}
\begin{proof}
Since $\mathbf P\mathbf v$ is in the null space of $\Phi$, we have
$\Phi \mathbf v' = \Phi(\mathbf v - 2\mathbf P\mathbf v) = \Phi \mathbf v.$
To simplify notation,
define
\[ \mathbf z = A(\Phi \mathbf v') = A(\Phi \mathbf v). \]
Vector $\mv z$ is well-defined because $A$ is deterministic.
Assumes $(a)$ holds, we will show that $(b)$ does not hold.
Let $j^*$ be the coordinate from Lemma \ref{lmm:non-flat}.
From (a), we have
\begin{eqnarray*}
1-\gamma - \langle \mathbf z, \mathbf e_{j^*} \rangle
&\leq&
\langle \mathbf v, \mathbf e_{j^*} \rangle -
\langle \mathbf z, \mathbf e_{j^*} \rangle\\
&\leq&
|\langle \mathbf v-\mathbf z, \mathbf e_{j^*} \rangle|\\
&\leq& \|\mathbf v - \mathbf z\|_2\\
&\leq& C \cdot \| \mathbf v-\mv v_k \|_2 \\
&\leq& C \cdot \sqrt{1-\langle \mathbf v, \mathbf e_{j^*} \rangle^2}\\
&\leq& C \sqrt{1-(1-\gamma)^2}.
\end{eqnarray*}
Hence,
\[ \langle \mathbf z, \mathbf e_{j^*} \rangle
   \geq 1-\gamma-C\sqrt{1-(1-\gamma)^2}.
\]
In particular, due to \eqref{eqn:gammadelta},
\[ \langle \mathbf z, \mathbf e_{j^*} \rangle \geq 0. \]
Consequently,
\[ \|\mathbf v'-\mathbf z\|_2 \geq |\langle \mathbf v'-\mathbf z, \mathbf e_{j^*} \rangle|
=
|\langle \mathbf v', \mathbf e_{j^*} \rangle -
 \langle \mathbf z, \mathbf e_{j^*} \rangle|
\geq
 \langle \mathbf z, \mathbf e_{j^*} \rangle
 - \langle \mathbf v', \mathbf e_{j^*} \rangle
\geq
 - \langle \mathbf v', \mathbf e_{j^*} \rangle.
\]

We next claim that
\begin{equation}
 - \langle \mathbf v', \mathbf e_{j^*} \rangle >
 C \sqrt{1-\langle \mathbf v', \mathbf e_{j^*} \rangle^2}.
\label{eqn:suf}
\end{equation}
Before proving the claim, let us first see how it implies
that $(b)$ does not hold.
To this end, observe that
\[ C \cdot \|\mathbf v'-\mv v'_k\|_2 \leq C \sqrt{1-\langle \mathbf v', \mathbf
e_{j^*} \rangle^2} < - \langle \mathbf v', \mathbf e_{j^*} \rangle
 \leq \|\mathbf v'-\mathbf z\|_2.
\]

Finally, we prove claim \eqref{eqn:suf}, which is equivalent to
$\langle \mathbf v', \mathbf e_{j^*} \rangle < - \frac{C}{\sqrt{1+C^2}}$:
\begin{eqnarray}
\langle \mathbf v', \mathbf e_{j^*} \rangle &=&
\langle (\mathbf I-2\mathbf P)\mathbf v, \mathbf e_{j^*} \rangle \nonumber \\
&=& \langle \mathbf v, \mathbf e_{j^*} \rangle -
    2\langle \mathbf P \mathbf v, \mathbf e_{j^*} \rangle \nonumber \\
(\text{Lemma \ref{lmm:non-flat}})&\leq&
   1 - 2(1-\gamma-\delta - \sqrt{2\gamma\delta}) \nonumber \\
&=& -1 + 2\gamma + 2\delta+2\sqrt{2\gamma\delta} \label{eqn:large}\\
(\text{from \eqref{eqn:gammadelta}}) &<& - \frac{C}{\sqrt{1+C^2}}. \nonumber
\end{eqnarray}
\end{proof}

\begin{lemma}
Let $(\Phi, A)$ be a fixed pair of (deterministic) measurement matrix
and decoding algorithm, where $\Phi$ is an $m \times N$ matrix.
Define $\delta = m/N$, and let $\gamma > 0, C\geq 1$ be arbitrary constants 
satisfying \eqref{eqn:gammadelta}.
Suppose we chose input vectors $\mv x$ uniformly
on the unit sphere $S^{N-1}$, then 
\[ \pr_{\mv x}[ \|\mv x - A(\Phi\mv x)\|_2 > C \|\mv x-\mv x_k\|_2]
    \geq \sqrt{1/\gamma} \cdot e^{-\frac N 2 \ln(2/\gamma)}. 
\]
\label{lmm:randomx}
\end{lemma}
\begin{proof}
Let $j^* \in [N]$ be the index satisfying \eqref{eqn:non-flat} guaranteed
to exist by Lemma \ref{lmm:non-flat}.
The set of vectors $\mv v$ for which
$\langle \mv v, \mv e_{j^*}\rangle \geq 1-\gamma$ is called the
$(1-\gamma)$-cap about $\mv e_{j^*}$ on the sphere $S^{N-1}$.
It is known (see, e.g., \cite{ball} -- Lemma 2.3)
that the $(1-\gamma)$-cap has measure at least
\[ (1/2)\cdot(\sqrt{\gamma/2})^{N-1} =
   \sqrt{1/\gamma} \cdot e^{-\frac N 2 \ln(2/\gamma)}.
\]
The mapping $\mathbf v \to (\mathbf I-2\mathbf P)\mathbf v$ from Lemma
\ref{lmm:symmetrize} is a reflection through the rowspace of
$\Phi$, hence the image of the reflection is another cap on the sphere
with exactly the same measure.
From \eqref{eqn:large} and \eqref{eqn:gammadelta}, the two caps are disjoint.
Lemma \ref{lmm:symmetrize} then completes the proof, because if
algorithm $A$ works well on a vector in one cap, then it will not work well on
the vector's reflection.
\end{proof}

Armed with the lemma above we prove our final lower bound stated below (and proved in Appendix~\ref{app:yao}). To do this we need a continuous version of Yao's lemma and a simple padding trick to reduce the case of $p=2^{-\Theta(N)}$ to larger values of $p$.

\begin{theorem}
\label{thm:main-lb}
Let $C \geq 1$ and $p$ be such that
$\sqrt{12+16C^2} \cdot e^{-\frac{\ln(6+8C^2)}{2} \cdot N} \leq p < 1.$
Then, any $\ltlt$ foreach sparse recovery scheme using $m \times N$ measurement
matrices $\Phi$ with failure probability at most $p$ and approximation factor
$C$ must have
$ m \geq \frac{1}{(6+8C^2)\ln(6+8C^2)}
     \ln\left(\frac{\sqrt{12+16C^2}}{p}\right) =
    \Omega(\log(1/p))$
measurements.
\end{theorem}

\subsection{Lower Bound for Bounded Adversary Model}

In this section, we show the following result:
\begin{theorem}
\label{thm:lb-bounded}
Any $\ltlt$ sparse recovery scheme that uses at most $b$ bits in each entry of $\Phi$ needs at least $\Omega\left(\sqrt{\frac{s}{b}}\right)$ number of measurements to be successful against an $s$-information-theoretically-bounded adversary.
\end{theorem}

The result follows from the $\ltlt$ forall lower bound argument from Corollary~\ref{cor:l2l2-lb-N} and the following simple observation:
\begin{lemma}
\label{lem:null-extend}
Let $A$ be an $m\times n$ matrix with $m\le n$. Consider the column sub-matrix $A'$ which has the first $n'$ (for some $m\le n'\le n$). If $\vn'$ is in the null space of $A'$, then $\vn=(\vn',\vzero_{n-n'})$ (i.e. vector $\vn'$ followed by $n-n'$ zeroes) is in the null space of $A$.
\end{lemma}

\begin{proof}[Proof of Theorem~\ref{thm:lb-bounded}] For the sake of contradiction assume that there exists an $\ltlt$-sparse recovery matrix $\Phi$ against any $s$-information-theoretically-bounded adversary that achieves an approximation factor of $C$ and has $m< \frac{\sqrt{s}}{C^2\sqrt{b}}$ measurements. Next we present an $s$-information-theoretically-bounded adversary that can foil such a system.

Let $\Phi'$ be the column sub-matrix of $\Phi$ that has only the first $n'\stackrel{\mathrm{def}}{=}\sqrt{s/b}$ columns of $\Phi$. Then by Proposition~\ref{prop:DeVore-non-flat}, there is a unit vector $\vn'$ that is in the null space of $\Phi'$ and satisfies $\|\vn'\|_{\infty}\ge 1-\frac{1}{C^2}$. By Lemma~\ref{lem:null-extend}, $\vn=(\vn',\vzero)$ is in the null space of $\Phi$. Further, it is easy to verify that $\vn$ is a unit vector and $\|\vn\|_{\infty}\ge 1-\frac{1}{C^2}$. Using the proof of Corollary~\ref{cor:l2l2-lb-N}, one can then argue that any recovery algorithm will have to fail on either the input $\vzero$ or $\vn$.

To complete the proof, we need to argue that the adversary only needs to remember $s$-bits of information about $\Phi$ to compute $\vn'$. Indeed note that we only need at most $m\cdot \sqrt{s/b} \cdot b \le s/C^2$ bits to describe the matrix $\Phi'$, which is enough to compute $\vn$.
\end{proof}

\section{Sublinear Decoding}
\label{sec:ub}

We present known results with polynomial time decoding on $\ltlt$ sparse recovery problem in Appendix~\ref{app:known}.

Our strategy for designing sub-linear time decodable top level systems will be as follows: we will first design weak identification matrices that have sublinear identification time. Then we (in a black-box manner) convert such matrices to sub-linear time decodable top level systems. We now present an outline of how we implement our strategy. 

In Section~\ref{sec:exp-sparse-rec} we show how expanders can be used to construct various schemes that will be useful later. In Section~\ref{sec:conversion}, we show how to convert weak identification systems to top level systems. The rest of technical development (in Section~\ref{sec:intermediate} and~\ref{sec:recursive}) is in designing weak identification system with good parameters.

Our first main result on sub-linear time decodable top levels systems will be:

\begin{theorem}
\label{thm:exp-fail}
For any $k\ge N^{\Omega(1)}$ and $\eps,\alpha>0$, there exists a $(k,1+\eps)$-top level system with $O(\eps^{-11}k\log(N/k))$ measurements, failure probability $(N/k)^{-k/\log^{13+\alpha}{k}}$ and decoding time $\eps^{-4}\cdot k^{1+\alpha}\cdot\log^{O(1)}{N}$.\footnote{The $O(\cdot)$ notation here hides the dependence on $\alpha$.} This scheme uses $O_{\eps}(k\cdot \log^{O(1)}{N})$ bits of space.
\end{theorem}

In fact, our results also work for $k=N^{o(1)}$ but we then do not get the optimal number of measurements. However, an increase in the decoding time leads to our second main result, which has near-optimal number of measurements. 

\begin{theorem}
\label{thm:exp-gen-k}
For any $1\le k\le N$ and $\eps,\alpha>0$, there exists a $(k,1+\eps)$-top level system with \\
$O(\eps^{-11}k\log(N/k)\log_k^{\alpha}{N})$ measurements, failure probability $(N/k)^{-k/\log^{13+\alpha}{k}}$ and decoding time $(k/\eps)^{\Theta(2^{-\alpha})}\cdot\log^{O(1)}{N}$.\footnote{The $O(\cdot)$ notation here hides the dependence on $\alpha$.} This scheme uses $O_{\eps}(k\cdot \log^{O(1)}{N})$ bits of space.
\end{theorem}

\subsection{Proof of Main Results}

Later in this paper, we will prove the following results, which we will use to prove Theorems~\ref{thm:exp-fail} and~\ref{thm:exp-gen-k}. 

\begin{lemma}
\label{cor:lw}
Let $0<\alpha,\gamma,\eta<1$ be real numbers and $1\le k\le k_0\le N$ be integers. Then there exists an \\
$O\left(\gamma^{-6}\eta^{-4}\cdot k\cdot \log(N/k)\cdot (\log_{k_0\log(N/k_0)}{N})^{6+\log(1/\alpha)/\log(1+\alpha)}\right)\times N$ matrix that is $(k,\gamma,\eta)$-weak identification matrix with $\left(O(k/\eta), \left(k_0\log(N/k_0)\right)^{-\Omega(\alpha\gamma k/\log_{k_0\log(N/k_0)}{N})}\right)$-guarantee with identification time complexity of
 $O\left(\gamma^{-3}\eta^{-3-\alpha}\cdot k\cdot k_0^{\alpha}\cdot \log^{O(1)}{N}\right)$,
and space complexity of $O_{\zeta,\eta}(k\log^{O(1)}{N})$.
\end{lemma}

\begin{lemma}
\label{cor:rs}
Let $\alpha$ be a small enough real number. Then for large enough $n$, there exists a $(k,\zeta,\eta)$-weak identification matrix with $(O(k/\eta),(n/k)^{-2k})$-guarantee with
\[O\left(\zeta^{-7}\eta^{-4}\cdot k\cdot\log(n/k)\cdot(\log_k{n})^{\alpha}\right)\]
measurements and a decoding time of
\[\zeta^{-9}\eta^{-4}\cdot k^{O(2^{1/\alpha})}\cdot\poly(\log{n})+\zeta^{-3}\eta^{-2}\cdot (k/\eta)^{O(2^{1/\alpha})}\cdot\poly(\log{n}).\]
\end{lemma}

Lemmas~\ref{cor:lw} and ~\ref{cor:rs} along our weak system to top level system conversion (Lemma~\ref{lem:weak-to-top} below) then prove Theorem~\ref{thm:exp-fail} and~\ref{thm:exp-gen-k} respectively.

\paragraph{Proof of Theorem~\ref{thm:exp-fail}.} We first note that applying Lemma~\ref{cor:lw} with $\gamma=\eta/2$ and $k_0\ge N^{\Omega(1)}$ implies the existence of a $(O(k/\eta),\eta/2,\eta)$ weak identification matrix with $O(O(k/\eta),(N/K_0)^{-\alpha\eta k}$-guarantee with $O(\eta^{-10}k\log(N/k))$ measurements. By Lemma~\ref{lem:weak-identification-amplify} below, we can amplify the failure probability to $(N/k_0)^{-\Omega(sk)}$ by increasing the number of measurements to $O(\alpha^{-1}\eta^{-11}k\log(N/k))$. This implies that the identification algorithm will identify all but $k/2$ elements of $H_{k+k/\eta}(\vx)$. This means like Lemma~\ref{lem:identify-to-weak}, we can convert this into a weak system on which we can then applying the conversion technique of Lemma~\ref{lem:weak-to-top} to get the claimed bounds. (When we're applying the recursive procedure to $k/2^j$ we use this value of $k$ in the weak systems and the ``original" value of $k$ as $k_0$ in the weak system.) For the space requirement, we have to add up the space requirement for each weak system, which can be done within the claimed bound.
\hfill $\Box$

\paragraph{Proof of Theorem~\ref{thm:exp-gen-k}.} The proof is almost the same as that for Theorem~\ref{thm:exp-fail} except we use Lemma~\ref{cor:rs} instead of Lemma~\ref{cor:lw}. 
The rest of the proof is exactly the same. (The way Lemma~\ref{lem:weak-to-top} is stated it needs $O(k\log(N/k))$  measurements but it can be verified that the conversion also works with the extra $O(\log^{\alpha}_k{N})$ in the number of measurements.)
\hfill $\Box$

\subsection{Consequences for the Bounded Adversary Model}

Our first corollary is an upper bound for the information-theoretic
bounded adversary and follows directly from
Lemma~\ref{obs:bottleneck} (by setting $\beta=\alpha2^{-s/\alpha}$) and the result
of~\cite{MR2449058} (see Theorem~\ref{thm:upper-poly}).

\begin{cor}
Fix $0 < \alpha < 1$. There is a randomized sparse signal recovery algorithm that with $m = O(k \log(N/k) + s/\alpha)$ measurements will foil an $s$-information-theoretically bounded adversary; that is, the algorithm's output will meet the $\ltlt$ error guarantees with probability $1 - \alpha$.
\end{cor}

The algorithm in~\cite{MR2449058} does not have a
sublinear running time. If the goal is to defeat such an adversary and
to do so with a sublinear algorithm, we must adjust our measurements
accordingly, using Table 1 as our reference for the range of parameters
$p$ and $m$. We note that in \cite{GHRSW12:SimpleSignals}, there was a 
single result for $O(\log{N})$-information-theoretically bounded
adversaries ($O(k\log(N/k)$ measurements are sufficient) and this
corollary provides an upper bound for the entire range of
parameter $s$.

\subsection{Consequences for the $\lolo$ forall sparse recovery}

We observe that our techniques also work for the $\lolo$ foreach sparse recovery. Actually the fact that an $\ltlt$-foreach sparse recovery system is also an $\lolo$-foreach sparse recovery system follows easily from known results. So in particular, Lemma~\ref{cor:rs} implies that one can get a similar system in the $\lolo$ sense with the same parameters. Furthermore from Remark~\ref{rem:l1l1-intermediate}, to convert this into a  weak $(k,\zeta,\eta)$ identification system with $(O(k/\eta),0)$-guarantee (i.e. a deterministic system), we need to take union bound over $\binom{N}{k'}+N^{x}$ events, where $k'=O(\zeta^{-5}\eta^{-2} k)$ and $x=O(\log(N/k))$. Thus, for $k\ge \Omega(\log(N/k)$ with the existing machinery from~\cite{ely-martin} to convert a weak system into a top-level system, we get the following result:

\begin{theorem}\label{thm:l1l1-forall}
For any $1\le k\le N$ (such that $k\ge \Omega(\log(N/k))$ and $\eps,\alpha>0$, there exists a $(k,1+\eps)$-$\lolo$ top level system  with 
$O(\eps^{-18}k\log(N/k)\log_k^{\alpha}{N})$ measurements and decoding time $(k/\eps)^{\Theta(2^{-\alpha})}\cdot\log^{O(1)}{N}$.\footnote{The $O(\cdot)$ notation here hides th
e dependence on $\alpha$.} This scheme uses $O_{\eps}(k\cdot \log^{O(1)}{N})$ bits of space.
\end{theorem}
\begin{proof}[Proof Sketch]
The results in~\cite{ely-martin} imply that a weak $(k,\eta/2,\eta)$ identification system (for any $\eta>0$) with $(O(k/\eta),0)$-guarantee (i.e., it is {\em deterministic}) can be converted into a $(k,1+\eps)$-$\lolo$ top level system with only a constant blowup in the number of measurements (with the same dependence on $\eps$ as one has on $\eta$). 

We note that by Lemma~\ref{cor:rs} and Lemma~\ref{lem:weak-identification-amplify} (and the observation above that an $\ltlt$ guarantee implies an $\lolo$ guarantee), we obtain a weak $(k,\zeta,\eta)$ identification system with $(O(k/\eta),(N/k)^{-\Omega(\zeta^{-5}\eta^{-2} k)})$-guarantee with $O(\zeta^{-12}\eta^{-6}k\log(N/k)\log^{\alpha}_k{N})$ number of measurements. Thus, by Remark~\ref{rem:l1l1-intermediate}, one can convert such a system into a deterministic one, which by the discussion in the paragraph above completes the proof.
\end{proof}

\subsection{Basic Building Blocks}

In this section, we lay out some basic building blocks that will help us prove Lemmas~\ref{cor:lw} and~\ref{cor:rs}.

\subsubsection{Expander Based Sparse Recovery}
\label{sec:exp-sparse-rec}
In this section, we record three results that will be useful to prove our final result. The proofs modify those from~\cite{ely-martin} and use an expander instead of random graph. (The actual arguments are similar.) The proofs are deferred to Appendix~\ref{app:thm-l2l2-exp}.
We will prove the following result:
\begin{theorem}
\label{thm:l2l2-exp}
Let $s\ge 1$ be an integer and $0<\eta,\gamma<1$ be reals. Let $G:[N]\times [\ell]\rightarrow [M]$ be a $(4k,\eps)$-expander with $\eps=O(\gamma^3\eta)$ and $\ell\ge c\cdot \log(N/k)$ (for some large enough constant $c=\Theta(s)$ that can depend on $\eps$). Let $\M$ be the random matrix obtained by multiplying each (non-zero) entry in ${\M}_G$ by random (independent) $\pm 1$. Then except with probability $\binom{N}{\gamma k}^{-s}$, $\M$ is a $(k,2\gamma,\sqrt{\eta})$ weak $\ell_2/\ell_2$  system.
\end{theorem}

The estimation algorithm in Appendix~\ref{app:thm-l2l2-exp} with Corollary~\ref{cor:l2l2-exp-weak-identification} (where we substitute $\eta$ by $\eta^2$) implies the following:
\begin{theorem}
\label{thm:l2l2-exp-identification}
Let $s\ge 1$ be an integer and $0<\eta,\gamma<1$ be reals. Let $G:[N]\times [\ell]\rightarrow [M]$ be a $(4k,\eps)$-expander with $\eps=O(\gamma^3\eta^2)$ and $\ell\ge c\cdot \log(N/k)$ (for some large enough constant $c=\Theta(s)$ that can depend on $\eps$). Let $\M$ be the random matrix obtained by multiplying each (non-zero) entry in ${\M}_G$ by random (independent) $\pm 1$. Then except with probability $\binom{N}{\gamma k}^{-s}$,  the following is true.

There exists an algorithm that given as input $S\subseteq [N]$ in time $O(|S|\cdot \ell+k/\eta\cdot\log(|S|))$ outputs $O(k/\eta)$ items $I\subseteq S$ that contain all but $\gamma k$ items $i\in S$ such that $|x_i|>3\sqrt{\eta^2/k}\|\vz\|_2$.
\end{theorem}

Finally, the proof of of Theorem~\ref{thm:l2l2-exp} also implies the following:
\begin{theorem}
\label{thm:l2l2-exp-weak}
Let $s\ge 1$ be an integer and $0<\eta,\gamma<1$ be reals. Let $G:[N]\times [\ell]\rightarrow [M]$ be a $(4k,\eps)$-expander with $\eps=O(\gamma^3\eta^2)$ and $\ell\ge c\cdot \log(N/k)$ (for some large enough constant $c=\Theta(s)$ that can depend on $\eps$). Let $\M$ be the random matrix obtained by multiplying each (non-zero) entry in ${\M}_G$ by random (independent) $\pm 1$. Then except with probability $\binom{N}{\gamma k}^{-s}$,  the following is true.

$\M$ is $(k,\zeta,\eta)$ weak $\ltlt$ system. Further,
there exists an algorithm that given as input $S\subseteq [N]$ (that contains all but $\zeta k/2$ elements of $H_{k+k/\eta}(\vx)$) in time $O(|S|\cdot \ell+k/\eta\cdot\log(|S|))$ outputs $\hat{\vx}$ with the required properties.
\end{theorem}

\subsubsection{Weak Identification to Top-Level system conversion}
\label{sec:conversion}

We begin with the following observation that follows by repeating the given weak identification matrix $s$ times (proof in Appendix~\ref{app:weak-amplify}):
\begin{lemma}
\label{lem:weak-identification-amplify}
Let $\M$ be a $(k,\zeta,\eta)$ weak identification matrix with $(\ell,p)$ guarantee.
Then there exists a $(k,3\zeta,\eta)$ weak identification matrix $\M'$ with $(2\ell,p^{\Omega(s)})$ guarantee with $s$ times more measurements.
\end{lemma}

By combining Lemma~\ref{lem:weak-identification-amplify} and Theorem~\ref{thm:l2l2-exp-weak}, we get:
\begin{lemma}
\label{lem:identify-to-weak}
Let $\M$ be a $(k,\zeta,\eta)$ weak identification matrix with $(\ell,p)$ guarantee with $m$ measurements and let $G$ be an expander as in Theorem~\ref{thm:l2l2-exp-weak}. Then for any integer $s\ge 1$, there exists a $(k,3\zeta,\eta)$ weak $\ltlt$ system with $O(m\cdot s+ M)$ measurements with failure probability $p^\Omega(s)+\binom{N}{\gamma k}^{-s}$.
\end{lemma}

Next we present a simple yet crucial modification of the weak level to top level system conversion from~\cite{ely-martin}. 
\begin{lemma}
\label{lem:weak-to-top} Let the following be true for any $0<\zeta,\eta<1$ and integers $s,k\ge 1$.
There is  a $(k,\zeta,\eta)$ weak $\ltlt$ system $\M$ with $O(s\cdot \zeta^{-d}\eta^{-c}\cdot k\cdot \log(N/k))$ measurements\footnote{$c$ and $d$ are absolute constants.}, failure probability $q^{-\Omega(sk)}$ and decoding time $T(k,\zeta,\eta,N,s)$. Then for any $\eps,\alpha>0$ and $k\ge 1$, there is a $(k,\eps)$-top level system with failure probability $q^{-\Omega(k/\log^{(1+\alpha)c+2+\alpha}{k})}$,   $O(\eps^{-c}\cdot k\cdot \log(N/k))$ measurements and $O(\log{k}\cdot T(k,\zeta,\eta,1))$ decoding time.
\end{lemma}
\begin{proof}[Proof Sketch] We only sketch the differences from the corresponding argument in~\cite{ely-martin}. We will have $\log{k}$ stages. In stage $i$, we will pick a $\left(\frac{k}{2^i},\frac{1}{2},\eps\cdot\frac{1}{i^{1+\alpha}}\right)$ weak $\ltlt$ system with $s=2^i/i^{(1+\alpha)c+2+\alpha}$ and ``stack" them to obtain our final top level system. Using the ``loop invariant" technique of Gilbert et al., and arguments similar to~\cite{ely-martin}, one can design a decoding algorithm that has an approximation factor of
\[1+O(\eps)\cdot \sum_{i=1}^{\infty}\frac{1}{i^{1+\alpha}}=1+O(\eps),\]
as desired. (In the above we used the fact that $\sum_{i=1}^{\infty} i^{-1-\alpha}=O(1)$ for any constant $\alpha$.)

We briefly argue the claimed bounds on the number of measurements and the failure probability (the rest of the bounds are immediate). Note that at stage $i$, the number of measurements (ignoring the constant in the $O(\cdot)$ notation):
\[\frac{1}{i^{(1+\alpha)c+2+\alpha}}\cdot2^i\cdot 2^d\cdot \eps^{-c}\cdot \left(\frac{1}{i^{1+\alpha}}\right)^{-c}\cdot \frac{k}{2^i}\cdot \log(N 2^i/k)\le \frac{1}{i^{1+\alpha}}\cdot 2^d\cdot \eps^{-c}\cdot k\log(N/k).\]
The claimed bound on the final number of measurements follows from the fact that $\sum_{i=1}^{\infty}\frac{1}{i^{1+\alpha}}=O(1)$.
Next note that the failure probability at stage $i$ is $q^{-\Omega(2^ik/(2^i i^{(1+\alpha)c+2+\alpha})}=q^{-\Omega\left(\frac{k}{i^{(1+\alpha)c+2+\alpha}}\right)}$. The final failure probability is determined by the failure probability at $i=\log{k}$, which implies that the final failure probability is $q^{-\Omega(k/\log^{(1+\alpha)c+2+\alpha}{k})}$,
as desired
\end{proof}

\subsubsection{An Intermediate Result}
\label{sec:intermediate}

In this section, we will present an intermediate result that will be the building block of our recursive construction (in Section~\ref{sec:recursive}).
Let $N\ge M\ge 1$ be integers and let $h:[N]\rightarrow [M]$ be a random map that we will define shortly.
Let $G:[M]\times [\ell]\rightarrow [m]$ be a $(4k,\eps)$-expander. In this section, we will consider how good an identification matrix we can obtain from $\M_{G\circ h}$ (with random $\pm 1$) entries. 

Call a map $h:[N]\rightarrow [M]$ to be $(k',\alpha)$-random if the following is true.
For any subset $S\subseteq [N]$ of size $k'$, the probability that for any $i\not\in S$, $h(i)=h(j)$ for any $j\in S$ is upper bounded by $O\left(\frac{|S|}{M^{1-\alpha}}\right)$. Note that a $(k'+1)$-wise independent random map from $[N]\rightarrow [M]$ is $(k',0)$-random.

We will prove the following ``identification" analogue of Theorem~\ref{thm:l2l2-exp}
(the proof appears in Appendix~\ref{app:intermediate}):
\begin{lemma}
\label{lem:intermediate}
Let $0<\alpha,\eta,\gamma<1$ be reals. Let $f:[N]\rightarrow [M]$ be a $(k(\zeta^{-5}\eta^{-2}+1),\alpha)$-random map with $M\ge \Omega(\Mlb)$. Let $G:[M]\times [\ell]\rightarrow [m]$ be a $(4k,\eps)$-expander with $\eps=O(\gamma^3\eta^2)$ and $\ell\ge c\cdot \log(M/k)$ (for some large enough constant $c=\Theta(s)$ that can depend on $\eps$). Let $\M$ be the random matrix obtained by multiplying each (non-zero) entry in ${\M}_{G\circ f}$ by random (independent) $\pm 1$. Then except with probability $\left(\frac{M}{k}\right)^{\Omega(-\alpha \gamma k)}$, the following is true.

There exists an algorithm that given as input $S\subseteq [M]$ outputs in time $O(|S|\cdot \ell)$ $O(k/\eta)$ items $I\subseteq S$ that contain $f(i)$ for all but $\gamma k$ items $i\in S$ such that $|x_i|>3\sqrt{\eta^2/k}\|\vz\|_2$.
\end{lemma}

\subsection{Proof of Lemma~\ref{cor:lw}}
\label{sec:recursive}

See Figure~\ref{fig:overview} for an overview of the proof of Lemma~\ref{cor:lw}.

\begin{figure}[h]
\begin{center}
\resizebox{5.25in}{!}{\rotatebox{90}{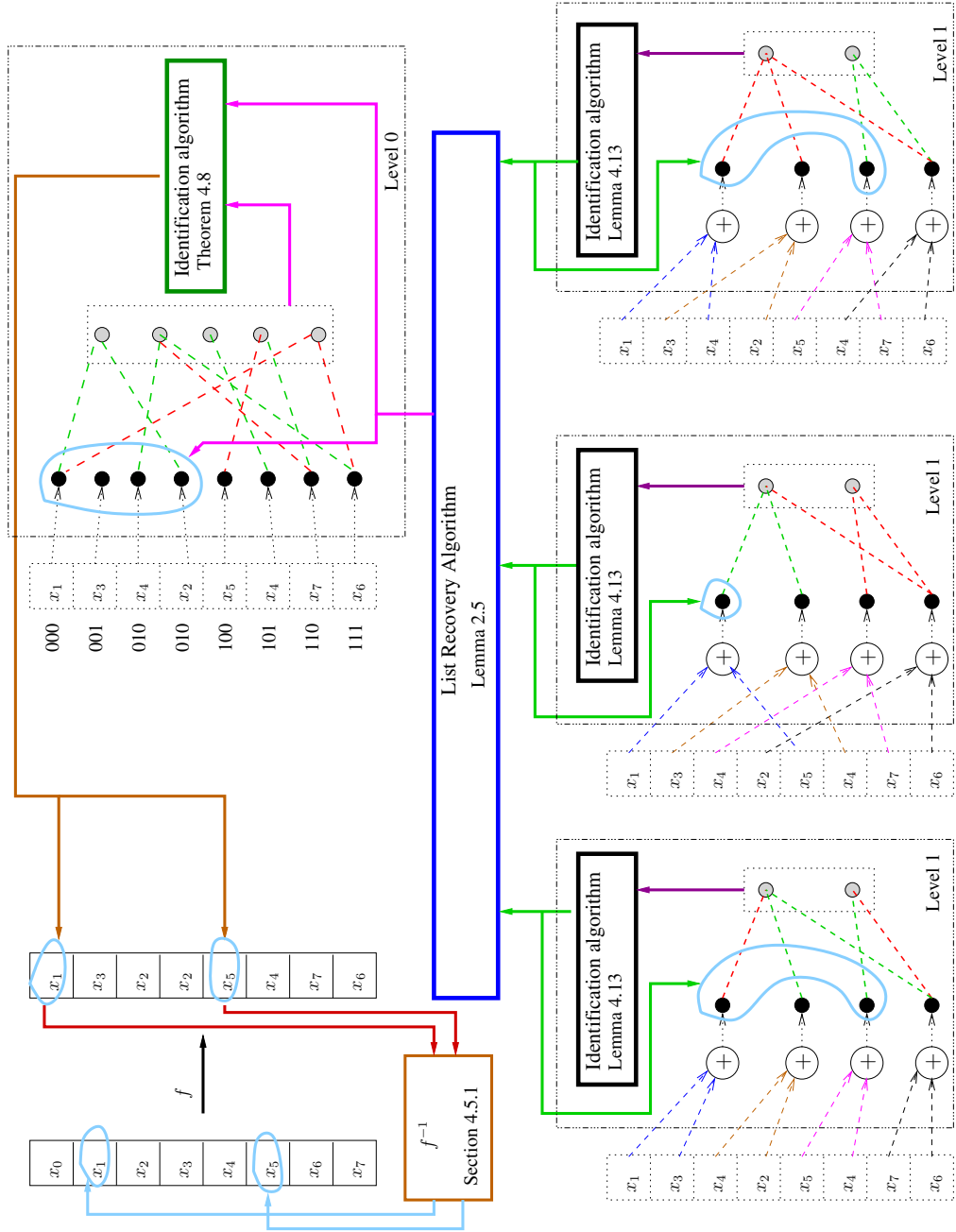}}
\caption{\footnotesize{Overview of the proof of Lemma~\ref{cor:lw} for the specific example of $N=8$ and one level of recursion. The encoding can be thought of as follows. The input vector $(x_0,\dots,x_7)$ is moved around by a random map $f$ (in this case $f$ is a permutation). Then the shuffled vector is acted upon by one expander on level $0$ and three expanders on level $1$. (For clarity the shuffled vector is copied next to the four expanders.) The green and red labels on the expander edges denote $+1$ and $-1$ weights respectively. At the level $0$ expander, the eight values are routed on the edges and are collected in the output (gray) nodes by adding the incoming values (weighted by the corresponding $\pm 1$ label). At the level $1$ expander a similar process happens except the shuffled vector is combined using the following code on the indices (which are represented in the vector next to the level $0$ expander for convenience) from Lemma~\ref{lem:lw}: $C(b_0,b_1,b_2)=((b_0,b_1),(b_1,b_2),(b_0,b_2))$. The left most level $1$ expander corresponds to the first codeword position, the middle one to the second position and the right one to the third position. Also the values corresponding to the indices that map to the same symbol in $\{0,1\}^2$ are added up together. For convenience, we have used the same colored arrows for the same symbol in $\{0,1\}^2$ for every level $1$ expander. E.g., in the middle expander, the top left vertex corresponds to $(b_1,b_2)=(0,0)$, which means index $000$ (which has value $x_1$) and index $100$ (which has value $x_5$) get added up. The decoding process reverses the encoding logic. We first identify the heavy hitters for level $1$ expanders using the identification algorithm from Lemma~\ref{lem:intermediate}. These lists are then combined for the $\{0,1\}^3$ domain by the list recovery algorithm for $C$ from Lemma~\ref{lem:lw}. The output is then used as the set $S$ in the identification algorithm from Theorem~\ref{thm:l2l2-exp-identification} to get the locations of the heavy items in the shuffled vector. Finally, we used the inversion procedure as outlined in Section~\ref{sec:invert} to obtain our final indices, which are the 2nd and the 6th positions. (For clarity, the set of identified indices at every step are surrounded by a light blue curve.)}}
\label{fig:overview}
\end{center}
\end{figure}

We note that Theorem~\ref{thm:l2l2-exp-identification} instantiated with an optimal (random) expander, implies the following:
\begin{cor}
\label{cor:weak-identification-family}
There exists a family of $(k,\zeta,\eta)$-weak identification matrices $\{\M_n\}_{n\ge \zeta^{-2}}$ with \\ $(O(k/\eta),(n/k)^{-\Omega(\gamma k)})$-guarantee, $O(\zeta^{-6}\eta^{-4} k\log(n/k))$ measurements and $O(\zeta^{-3}\eta^{-1}|S|\cdot\log{n})$ identification time. 
\end{cor}

Our main result is that we can recursively construct a weak identification matrix that has efficient identification from the (family) of weak identification matrices guaranteed by Corollary~\ref{cor:weak-identification-family} (that by itself does not have a sub-linear identification algorithm). The main idea is as follows: using a tree structure, we map $[N]$ to successively smaller sub-problems. At the ``leaves" the domain size is $\tilde{O}(k)$ and thus one can use any identification matrix (including the identity matrix). The main insight is on how to break up the domains at an internal node. For illustration purposes, consider the root. We use a uniform and (good) list recoverable code $(r,N)_{\sqrt[b]{N}}$-code $C$ (for some parameter $b>1$). (At the end we will use one of the codes from Section~\ref{sec:prelims} as $C$.) More precisely the domain $[N]$ is broken down to $r$ copies of the domain $[\sqrt[b]{N}]$ and  $i\in [N]$ gets mapped to $C(i)_j$ in the $j$th child/sub-problem. By induction, we will prove that each of the $r$ children, the weak identification problem can be solved efficiently. In other words, for each $j\in [r]$, we will have a candidate set $S_j\subseteq [\sqrt[b]{N}]$, each of size $O(k/\eta)$. As $C$ is list recoverable we can recover a list $S$ that contains $i$. However, this list could be bigger than the required $O(k/\eta)$ bound (though not much larger). To prune down this to a list of size $O(k/\eta)$ we use a weak identification matrix $\M$ to narrow down the list to $I$ in running time proportional to $|S|$.

Below we state our formal result.

\begin{theorem}
\label{thm:recursive-exp}
Let $0<\zeta,\alpha,\eta<1$ be real numbers.
Let $\{\M_n\}_{n\ge \zeta^{-2}}$ be a family of $m(n)\times n$ matrices from Corollary~\ref{cor:weak-identification-family} that are $(k,\zeta,\eta)$-weak identification matrices with $\left(\ell\stackrel{def}{=}O(k/\eta),p(n)\stackrel{def}{=}(n/k)^{-\Omega(\zeta k)}\right)$-guarantee where $m(n)=g(\zeta,\eta)\cdot k\cdot\log(n/k)$ for some function $g(\zeta,\eta)$. 
~Finally, for any real numbers $b>1$, $0\le \rho <1$ and integer $r\ge 1$, let $\{C_n\}_{n\ge 1}$ be a family of codes, where $C_n$ is a $(r,n)_{\sqrt[b]{n}}$ code that is $(\rho,\ell,L)$-list recoverable in time $T(\ell,n,b)$.

Then for large enough $N$ there exists a $m'(N)\times N$ matrix $\M^*$ that is $(k,\zeta',\eta)$-weak identification matrix with $(O(k/\eta),p'(N))$-guarantee and identification time complexity $D(k,N)$ as follows. In what follows let $A\ge \Omega(\Mlb)$ satisfy the lower bound in Lemma~\ref{lem:intermediate}.

There exists two integers $h\le O(\log_r\log_A{N})$ 
and
$\N=(\log_{A}{N})$ such that the following hold:

\begin{equation}
\label{eq:zeta-1}
\zeta'\le \left\{
\begin{array}{ll}
\zeta\cdot \N^{O\left(\frac{\log(1/\rho)}{\log{r}}\right)}&\text{ if } \rho>0\\
\zeta\cdot \N&\text{otherwise}
\end{array} \right. ,
\end{equation}
\begin{equation}
\label{eq:failure-1}
p'(N)\le \N\cdot \left(\frac{A}{k}\right)^{-\Omega(\zeta k)},
\end{equation}
\begin{equation}
\label{eq:measurements-1}
m'(N) \le O\left(g(\zeta,\eta)\cdot k\cdot \log(N/k)\cdot \N^{\frac{\log{\frac{r}{b}}}{\log{b}}}\right),
\end{equation}
and
\begin{equation}
\label{eq:decoding}
D(k,N)=O\left(\zeta^{-3}\eta^{-2}\cdot \N\cdot A\cdot \log{A}\right)+\sum_{j=0}^{h-1} \left(r^j T(O(k/\eta),\sqrt[b^j]{N},b)+ O\left(\zeta^{-3}\eta^{-2}\cdot L\cdot \log{\sqrt[b^j]{N}}\right)\right).
\end{equation}
\end{theorem}

We now instantiate Theorem~\ref{thm:recursive-exp} with Corollary~\ref{cor:weak-identification-family} and the list recoverable code from Lemma~\ref{lem:lw} to prove Lemma~\ref{cor:lw}.

\begin{proof}[Proof Sketch of Lemma~\ref{cor:lw}]
Let $d=1+1/\alpha$.
For the code from Lemma~\ref{lem:lw} we have $\rho=0$,  $r=d$, $b=d/(d-1)$, $\ell=O(k/\eta)$ and $L=O(k^{1+\alpha}/\eta^{1+\alpha})$. We apply Theorem~\ref{thm:recursive-exp} with the family from Corollary~\ref{cor:weak-identification-family} with $\zeta=\frac{\gamma}{\N}$ and $A=\Theta(\zeta^{-6}\eta^{-2}\cdot k\cdot k_0^{\alpha}(\log(N/k)^2)$. Note that this implies that $\N=O(\log_{k_0\log(N/k_0)}{N})$. The claimed bounds then follow from the bounds in Theorem~\ref{thm:recursive-exp}
and Lemma~\ref{lem:scheme2}. We also need to take into account the randomness needed to for the random $\pm$ signs. However, it is easy to check that in each of the $\N$ nodes, we need $\poly(\zeta,\eta)\cdot k$-wise independence (see Remark~\ref{rem:lim-indep-intermediate}). Storing these random bits can be done within the claimed space usage.
\end{proof}

In the rest of the section, we prove Theorem~\ref{thm:recursive-exp}.

\begin{proof}[Proof of Theorem~\ref{thm:recursive-exp}.] The proof is essentially setting up the recursive structure (which we will state with the help of a tree) and to define the overall identification algorithm in the natural recursive way. We then verify that all the claimed bounds on the parameters hold. The only non-trivial part is the bound on $\zeta'$ for the case when $\rho>0$, where we use a careful union bound to obtain the better bound.

Define
\[h=\lceil \log_r\log_{A}{N}\rceil,\]
and
\[\N=\frac{r^{h+1}-1}{r-1}=O(\log_{A}{N}).\]

To prove (\ref{eq:zeta-1}), we will prove
\begin{equation}
\label{eq:zeta}
\zeta'\le \zeta\cdot\sum_{j=0}^h \left(\frac{r}{\frac{\rho r+1}{2}}\right)^j=\left\{
\begin{array}{ll}
\zeta\cdot \N^{O\left(\frac{\log(1/\rho)}{\log{r}}\right)}&\text{ if } \rho>0\\
\zeta\cdot \N&\text{otherwise}
\end{array} \right. , 
\end{equation}

to prove (\ref{eq:failure-1}), we will prove
\begin{equation}
\label{eq:failure}
p'(N)\le \sum_{j=0}^h r^j p\left(\sqrt[b^j]{N}\right)\le \N\cdot p(A) \le \N\cdot \left(\frac{A}{k}\right)^{-\Omega(\zeta k)},
\end{equation}

and to prove (\ref{eq:measurements-1}), we will prove
\begin{equation}
\label{eq:measurements}
m'(N)=g(\zeta,\eta)\cdot \sum_{j=0}^h r^j\cdot k\cdot \log\left(\sqrt[b^j]{N}/k\right) \le O\left(g(\zeta,\eta)\cdot k\cdot \log(N/k)\cdot \N^{\frac{\log{\frac{r}{b}}}{\log{b}}}\right).
\end{equation}

\paragraph{The construction.} We will present the recursive construction via a tree. (See Figure~\ref{fig:overview} for an illustration.) In particular, for ease of exposition, we assume that $N$ is a power of $2$. Consider a complete $r$-ary tree with height $h$. We will label the tree as follows: associate any node at level $0\le j\le h$ with the domain $[\sqrt[b^j]{N}]$. Note that the leaves are associated with the domain $[A]$ and the tree has $\N$ nodes.

Let $v$ be an arbitrary node at level $j$. Order the $r$ outgoing edges in some (fixed) order and associate the $u$th edge with the $u$th coordinate of the codewords in $C_{\sqrt[b^j]{N}}$. In particular, this defines a function $\phi_w:[N]\rightarrow [\sqrt[b^{j+1}]{N}]$ for each node $w$ on the $(j+1)$th level. Let $u_1,\dots,u_j$ be the order of the edges used from the root to $v$ in the tree. Define $\phi_v(i)$ to be the symbol obtained by successively applying the corresponding code in each of the $j$ levels in the path from the root to $v$. In particular,
\[\phi_v(i)= C_{\sqrt[b^{j-1}]{N}}\left(\cdots  C_{\sqrt[b]{N}}\left(C_{N}(i)_{u_1}\right)_{u_2}\cdots\right)_{u_j}.\]

For every vertex $v$ at level $j$, we associate with it a $m(\sqrt[b^j]{N})\times N$ matrix $\M(v)$ as follows. For any $i\in [N]$, the $i$th column in $\M(v)$ is the same as the $\phi_v(i)$'th column in $\M_{\sqrt[b^j]{N}}$.

Let $\M'$ be the matrix that is the stacking of all the matrices $\M(v)$ for every node $v$ in the tree. (The order does not matter as long as one can quickly figure out the rows of $\M(v)$ from $\M'$.) As the final step, we will randomly rearrange the columns of $\M'$. In particular, let $f:[N]\rightarrow [N]$ be defined as follows: for every $i\in [N]\equiv \{0,1\}^{\log{N}}$, define $f(i)$ to be $i$ with each of the $\log{N}$ bits (independently) flipped with probability $1/2$. Note that this implies that for every $i\in [N]$, $f(i)$ is uniformly distributed in $[N]$.

Finally, define $\M^*$ to be the matrix, where the $i$th column (for any $i\in [N]$) is the $f(i)$'th column in $\M'$.

\paragraph{The identification algorithm.} The algorithm will be described recursively. First we assume that location $i$ is mapped to $f(i)$ and hence we can now think of $\M^*$ as $\M'$. We will start with a leaf $\ell$ (i.e. a vertex at level $h$). We claim that the matrix $\M(\ell)$ is of the form required in Lemma~\ref{lem:intermediate}. Indeed, since each of the codes used in the construction are uniform, $\phi_{\ell}$ maps elements in $[N]$ uniformly at random to $[A]$. Thus, given $\M(\ell)\cdot \vx$, we use the algorithm in Lemma~\ref{lem:intermediate} (with $S=S_{\ell}\stackrel{def}{=}[A]$) and ``send" $I_{\ell}$ to its parent node. Note that $|I_{\ell}|\le O(k/\eta)$.

Now consider a node $v$ at level $0\le j<h$. By induction for each of its $r$ children, we will receive subsets $I_1,\dots,I_r\subseteq [\sqrt[b^{j+1}]{N}]$. Note that these form a valid input for list recovery of $C_{\sqrt[b^j]{N}}$. We run the list recovery algorithm on $I_1,\dots,I_r$ to obtain a subset $S_v\subseteq [\sqrt[b^j]{N}]$. Again using the arguments similar to the paragraph above, we can assume we can run the algorithm from Lemma~\ref{lem:intermediate} to obtain a subset $I_v\subseteq [\sqrt[b^j]{N}]$ and ``send" it to its parent.

If $w$ is the root of the tree then the final output is $I_w$. There is one catch: the output $I_w$ is the collection of $f(i)$'s for the appropriate indices $i$. However, the way $f$ is defined it is not a one to one function and thus, we cannot just apply $f^{-1}$ on to the indices in $I_w$. We will return to this issue in Section~\ref{sec:invert}.

\paragraph{Correctness of the construction.} 
Consider an item $i\in [N]$ such that $|x_i|>3\sqrt{\eta/k}\|z\|_2$. Now consider a node $v$ at level $j$ in the tree and let $I_1,\dots,I_r\subseteq [\sqrt[b^{j+1}]{N}]$ be the sets passed up the $r$ children of $v$ during the identification process above. Further, assume that for at least $\rho r$ values $u\in [r]$, $C_{\sqrt[b^j]{N}}(\phi_v(i))_u\in I_u$. Then since $C_{\sqrt[b^j]{N}}$ is $(\rho,O(k/\eta),L)$-list recoverable, we will have $\phi_v(i)\in S_v$. Then when we run the algorithm from Lemma~\ref{lem:intermediate} on $S_v$. Now, two things can happen. Either, $\phi_v(i)\in I_v$ or not. In the latter case, we will lose $i$ but we will account for such items $i$ in the next part. In the former case, let $v$ be the $j'$th child of vertex $u$. Then note that $\phi_v(i)=C_{\sqrt[b^{j-1}]{N}}(\phi_u(i))_{j'}$ and thus, we can apply the argument inductively to $u$. In other words, if $i$ is never lost then we will have $i$ in the final output $I_w$, assuming we can prove the base case. For the base consider any leaf $\ell$. Since we picked $S_{\ell}=[A]$ (and hence, trivially $\phi_{\ell}(i)\in S_{\ell}$), if the algorithm from Lemma~\ref{lem:intermediate} does not lose $i$, then $\phi_{\ell}(i)\in I_{\ell}$ as desired.

\paragraph{Analyzing the number of lost heavy hitters.} If $\rho=0$, then it is easy to see that the identification algorithm above can lose at most $\N\cdot (\zeta k)$ of the ``heavy enough" items. Thus, in this case the bound of $\zeta'\le \N\cdot \zeta$ is obvious.

Next, we consider the case when $\rho>0$. Define $e=\rho r+1$. We will argue soon that every heavy item $i$ that is not in $I_w$ it has to be the case that there exists a level $j$ such that for at least $\left(e/2\right)^j$ nodes $v$ such that $\phi_v(i)$ is lost by the algorithm from Lemma~\ref{lem:intermediate} when it is run on $S_v$. In such a case we will assign item $i$ to level $j$. Note that in the previous naive argument, we count each $i$ at least $\left(e/2\right)^j$ times at level $j$. In total, level $j$ can lose at most $r^j\cdot\zeta\cdot k$ items, which implies that each level can lose at most $\left(\frac{r}{\frac{e}{2}}\right)^j\cdot \zeta\cdot k$ distinct items at level $j$. Summing over all the levels and noting that the total number of lost items is $\zeta'\cdot k$, we obtain the first inequality in (\ref{eq:zeta}). It is easy to verify that $\sum_{j=0}^h (r/(e/2))^j \le (2r/e)^{O(h)}$, from which the equality in (\ref{eq:zeta}) follows by substituting the values of $h$ and $e$.

Finally, we argue that every lost item $i\not\in I_w$ is associated with a level $j$, where at least $(e/2)^j$ invocations of the algorithm from Lemma~\ref{lem:intermediate} lose $i$. Towards this end, call a vertex $v$ \textit{bad internal node} for $i$ if the invocation of the algorithm from Lemma~\ref{lem:intermediate} does not lose $\phi_v(i)$ but $\phi_v(i)\not\in I_v$. Call the vertex $v$ \textit{bad leaf} for $i$ if the invocation of the algorithm from Lemma~\ref{lem:intermediate} itself loses $\phi_v(i)$. Note that is a node $v$ is a bad internal node for $i$ then at least $e$ of its children are either a bad internal node of a bad leaf for $i$. In other words, the set of bad internal nodes and bad leaves for $i$ form a tree $T_i$ with degree at least $e$ such the leaves of $T_i$ are exactly the bad leaves for $i$. Finally, note that we want to show for every lost $i$, there is a level $j$ with at least $(e/2)^j$ leaves in level $j$. This now is a simple question on trees that we prove by a simple greedy argument.

W.l.o.g. assume that every internal node in $T_i$ has degree exactly $e$. If $T_i$ has one node then we are done. If not, consider the $e$ children of the root. If at least $e/2$ of these are leaves then we are done. If not, the root has at least $e/2$ internal nodes. Thus, if we have not stopped for $j$ levels, we will have at least $(e/2)^j\cdot e$ children at level $j+1$. If at least half of them are leaves then we have at least $(e/2)^{j+1}$ leaves in level $j+1$ otherwise we have at least $(e/2)^{j+1}$ internal nodes. We can continue this argument in the worst-case till $j=h$ but all nodes at level $h$ have to be leaves, which means that the process will terminate with at least $e(e/2)^{h-1}$ leaves.

\paragraph{Analyzing the failure probability.} The first inequality in (\ref{eq:failure}) follows from the union bound (recall that there are $r^j$ nodes at level $j$). The second inequality follows from the fact that $p(n)$ is decreasing in $n$ and the final inequality follows from the definition of $p(\cdot)$.

\paragraph{Analyzing the number of measurements.} Note that $m'(M)=\sum_{j=0}^h r^j\cdot m(\sqrt[b^j]{N})$, which proves the equality in (\ref{eq:measurements}). The inequality follows by bounding the sum $\sum_{j=0}^h (r/b)^j$ as we did earlier.

\paragraph{Analyzing the identification time.} Finally, in the identification algorithm for each node $v$ (at level $j<h$), we run two algorithms: one is the list recovery algorithm (which takes times $T(O(k/\eta),\sqrt[b^j]{N},b)$) and the other is the algorithm from Lemma~\ref{lem:intermediate}, which by Corollary~\ref{cor:weak-identification-family} takes time $O(\zeta^{-3}\eta^{-2}|S_v|\cdot \log{\sqrt[b^j]{N}})$. The sum in (\ref{eq:decoding}) comes from noting that $|S|\le L$. Finally, at level $j=h$, we pick $S_v=[A]$ and hence  Lemma~\ref{lem:intermediate} and the fact that there are at most $\N$ nodes at level $h$ justifies the first term in (\ref{eq:decoding}).
\end{proof}

\subsubsection{Inverting the indices} 
\label{sec:invert}

We are done with the proof of Theorem~\ref{thm:recursive-exp} except for taking care of the issue that the random function $f:[N]\rightarrow [N]$ chosen earlier might not be one to one. 

We now mention two ways that we can modify the proof above to take care of this. 
The difference is in the final identification time we can achieve.

For simplicity we will assume that $N$ is a power of $2$. Thus it makes sense to talk about fields $\F_N$ and $\F_{N^2}$. We will also go back and forth between the field representation and $[N]$ and $[N^2]$ respectively. Let $d=k(\zeta^{-5}\eta^{-2}+1)$. 

With $\Omega(N)$ storage one can solve this problem for {\em any} code with the required properties. See Appendix~\ref{app:gen-invert} for the details.

Next, we present a
specific but cheap solution. For this part, we will make a strong assumption on our list recoverable codes $C$-- we will assume that we are looking at the Loomis-Whitney code (from Lemma~\ref{lem:lw}) or Reed-Solomon codes (from Lemma~\ref{lem:rs}). 

We pick $g:[N]\rightarrow [N^2]$ randomly by picking a random degree $d$ polynomial over $\F_{N}$ and define $g(i)$ to be the evaluation of this polynomial at $i$. It is well-known that such a random hash function is $(d+1)$-wise independent (see e.g.~\cite{coding-book}). Finally define $f(i)=(i,g(i))$.

We have to tackle three issues.

First, we now have to first compute the matrix $\M'$ with $N^2$ columns (then we pick the columns of $\M'$ indexed by $(i,g(i))$ for $i\in [N]$) instead of $N$ columns as we have used so far. This however, is not an issue since we have a family of matrices $\M_{n}$. This change results in some constants in the parameters of $\M^*$ changing but it does not affect the asymptotics.

Second, note that every $f(i)$ is unique and recovering $i$ from $f(i)$ can be done in $O(1)$ time. Thus, we will not need any additional space and will add an additive factor of $O(k/\eta)$ to the identification time.

Finally, we have to show that this new $f$ is fine with respect to applying Lemma~\ref{lem:intermediate}.
In particular, we have to show the following: consider any node $v$ in the recursion tree. Then we will show that for the (new) $\phi_v:[N]\rightarrow [M]$ is $(d,\alpha)$-random (for some small $\alpha$). By our assumption on $C$, note that this implies that the $\log{M}$ bits of $\phi_v(i)$ for any $i$ are from the $2\log{N}$ bits from $(i,g(i))$. Unfortunately, for small $M$, all these $\log{M}$ bits might all be from the deterministic part (i.e. these are some $\log{M}$ bits from the $\log{N}$ bit representation of $i$), which implies that there will be lot of collisions (in fact $N/M$ indices might collide).

However, the main insight is that when we divide the domain at a node $v$ into its $r$ children, we have {\em full} freedom in how we divide up the bits. (Recall that the code $C$ just repeats some subset of message bits as a codeword symbol.)  The bad case mentioned above is only true if we do not do this division carefully. Thus, the idea is to do this division carefully so that the ``$g(i)$" part of $f(i)$, which has full randomness, is passed along.

To fully implement this, we will need to change $g$ a bit. In particular, pick $t=\left\lceil \frac{1}{\alpha}\right\rceil$. 
Then pick $g$ to be random degree $d$ polynomial over $\F_{N^{t-1}}$. Now consider a code in the recursive tree structure. W.l.o.g., let us assume that $C$ is an $(r,N^t)_{\sqrt[b]{N^t}}$ code. Let $j\in [r]$, then we want to argue that $C((i,g(i)))_j$ is pretty much random. 

We first consider the case that $C$ is the Loomis-Whitney code. 
First think of the message as coming form $[N]\times [N^{t-1}]$ (recall that $f(i)=(i,g(i))$ and hence is from this domain). Consider a message $(i,\ell)$. We partition it up  into $d$ symbols from $[\sqrt[d]{N}]\times[\sqrt[d]{N^{t-1}}]$. I.e. each symbol has the same ``proportion" of randomness. In other words, a symbol $(i',\ell')$ in the codeword for $(i,\ell)$ will have $i'$ to be deterministic and the $\ell'$ part to be completely random.

Next we consider the case when $C$ is the $(r,N)_{\sqrt[b]{N}}$ Reed-Solomon codes.
First think of the message as coming form $[N]\times [N^{t-1}]$ (recall that $f(i)=(i,g(i))$ and hence is from this domain). Consider a message $(i,\ell)$. We partition it up  into $b$ symbols  $u_0,\dots,u_{b-1}\in [\sqrt[b]{N}]\times[(\sqrt[b]{N})^{t-1}]$. I.e. each symbol has the same ``proportion" of randomness. Let us associate $[\sqrt[b]{N}]$ with the field $\F_q$, where $q=\sqrt[b]{N}$. Recall that we need to choose $r$ evaluation points for the code $C$. Since we have the full freedom in picking these elements, pick distinct $\beta_1,\dots,\beta_r\in \F_q$. (Recall that $r$ is a constant, so this is always possible.) Now recall that the $C(i,\ell)_j$ will be the element
\[\sum_{s=0}^{b-1} \beta_j^s\cdot u_s.\]
Let us think of $u_s$ as an element of $\F_q^t$ in the standard notation. I.e. $u_s=\sum_{a=0}^{t-1} u_s^a\cdot\gamma^i$, where $\gamma$ is a root of an irreducible polynomial of degree $t$ over $\F_q$ (and the one that defines the standard representation of $\F_{q^t}$). Then the above relation can be written as
\[C(i,\ell)_j = \sum_{s=0}^{b-1} \beta_j^s\cdot \left(\sum_{a=0}^{t-1} u_s^a\cdot\gamma^a\right)=\sum_{a=0}^{t-1}\gamma^a\cdot \left(\sum_{s=0}^{b-1}\beta_j^s\cdot u_s^a\right).\]
Since $\ell=g(i)$ was completely random, the above implies that if one thinks of $C(i,\ell)_j$ as an element of $\F_q\times\F_{q^{t-1}}$, then the part in $\F_{q^{t-1}}$ will be completely random.

The above implies that when $C$ is the Loomis -Whitney or Reed-Solomon code, for any $\phi_v(i)$ has as its last $(1- \alpha)\log{M}$ of its bits to be uniformly random. We will now show that $\phi_v$ is $(d,\alpha)$-random, which will be sufficient to apply Lemma~\ref{lem:intermediate}. Towards that end, let us fix a subset $S\subseteq [N]$ of size $d$. Consider $i\in [N]\setminus S$. Consider the values  $\{\phi_v(j)\}_{j\in S}$. Note that we do not have any control over the first $\alpha\log{M}$ bits in these values-- so in the worst-case let us assume that they (as well as $\phi_v(i)$) have the same $\alpha\log{M}$ bit-vector as its prefix. However since, $g$ is $d$-wise independent, we have that the  $(1-\alpha)\log{M}$ bit suffixes of $\{\phi_v(j)\}_{j\in S\cup \{i\}}$ are all uniformly random and {\em independent}. Thus, even conditioned on the values of $\phi_v(j)$ for every $j\in S$, the probability that $\phi_v(i)=\phi_v(j)$ for some $j\in S$, is upper bounded by
\[\frac{d}{2^{(1-\alpha)\log{M}}}=\frac{d}{M^{1-\alpha}},\]
as desired.

We will call this scheme above \schemet. Note that we have show that \schemet\ works and
\begin{lemma}
\label{lem:scheme2}
\schemet\ adds $O(k/\eta)$ to the decoding time in Theorem~\ref{thm:recursive-exp} and needs $O(\zeta^{-5}\eta^{-2}\cdot k\cdot \log^2{N})$ bits of space.
\end{lemma}

The claim on the space comes from the amount of randomness needed to define $g$ (which we need to store).

\subsection{Proof of Lemma~\ref{cor:rs}}
\label{sec:weak-rs}

In this section, we do a better analysis of the earlier recursive construction (from the proof of Lemma~\ref{cor:lw}) assuming the list recoverable code also leads to a limited-wise independent source, which lead us to the proof of Lemma~\ref{cor:rs}. For our purposes, this code will be the RS code.

We begin with the following ``inductive step" of the recursive construction.

\begin{lemma}
\label{lem:lr-weak-one-step}
Let $\M$ be an $m\times \sqrt[b]{n}$ matrix from Corollary~\ref{cor:weak-identification-family} that is a $(k,\zeta,\eta)$-weak identification matrix with $(\ell,p)$ guarantee and identification time $T_2$. Let $C$ be an $(r,n)_{\sqrt[b]{n}}$ code that is $(\rho,\ell,L)$-list recoverable in time $T_2$. Assume that the following holds:
\begin{equation}
\label{eq:p-rho}
p<\left(\frac{\rho}{2e}\right)^2.
\end{equation}

Then there exists a $(rm)\times n$ matrix $\M^*$ that is $(k,\zeta',\eta)$-weak identification matrix with $(L,p')$ guarantee with identification time $r\cdot T_1+T_2$ with
\begin{equation}
\label{eq:new-zeta-one-step}
\zeta'\le \frac{2\zeta}{\rho}
\end{equation}
and
\begin{equation}
\label{eq:new-p-one-step}
p'\le p^{\rho r/4},
\end{equation}
where the codewords in $C$ form an $a$-wise independent source for some $a\ge \rho r/2$.
\end{lemma}
\begin{proof}
The basic idea is to first encode each index $i\in [n]$ with $C$ and then use a copy of $\M$ for each of the $r$ positions in $C(i)$. For identification, we first run the identification algorithm for $\M$ on each of the $r$ positions and then combine the result using the list recovery algorithm of $C$. The main insight in the analysis of the algorithm is that the $\zeta$ fraction of the heavy hitters dropped in each of the $r$ are $a$-wise independent. Thus, by Chernoff the probability that a heavy hitter is dropped more than $\rho r$ times is exponentially small. Next we present the details.

\paragraph{The construction.} For $j\in [r]$, let $\M_j$ denote the $m\times n$ matrix obtained as follows. For every $i\in [n]$, the $i$'th column in $\M_j$ is the $C(i)_j$'th column in $\M$. Let $\M'$ be the stacking of the matrices $\M_1,\dots,\M_r$. Let $f:[n]\rightarrow [n]$ be a completely random map. Then the final matrix $\M^*$ is obtained by putting as its $i$'th column the $f(i)$'th column of $\M'$.

It is easy to verify that $\M^*$ has $rm$ rows, as desired.

\paragraph{The identification algorithm.} We will first assume that the index $i$ has been mapped to $f(i)$. (Ultimately, we will need to do the inversion. We address this part separately in Section~\ref{sec:invert}.) Thus, we can assume that we are working with $\M'$. Now if we consider the $\M_j$ part of $\M'$ (for any $j\in [r]$) and run the identification algorithm for $\M_j$ (from Corollary~\ref{cor:weak-identification-family}), then we will get a list $S_j$ of the values $C(h)_j$ for all but $\zeta k$ heavy hitters $h\in [n]$. We then run the list recovery algorithm for $C$ on $S_1,\dots,S_r$ to obtain our final list of candidate heavy hitters. Note that this algorithm takes time $r\cdot T_1+T_2$ as desired. Next, we analyze the correctness of the algorithm.

\paragraph{Correctness of the algorithm.} Consider a heavy hitter $h\in [n]$. Note that it will not be output iff $C(h)_j\not\in S_j$ for $>\rho r$ positions $j\in [r]$. Let $\zeta'$ be the fraction of heavy hitter we miss due to the above condition. Next we bound $\zeta'$ and the failure probability $p'$.

Call a codeword position $j\in [r]$ to be {\em bad} if the identification algorithm for $\M$ loses more than $\zeta k$ heavy hitters when we run the algorithm for $\M_j$. By Lemma~\ref{lem:chernoff}, the probability of having $>\rho r/2$ bad positions is upper bounded by

\[\left( \frac{epr}{\rho r/2}\right)^{\rho r/2}= \left(\frac{2ep}{\rho}\right)^{\rho r/2} \le p^{\rho r/4},\]

where we use the fact any $a\ge \rho r/2$ positions in a random codeword in $C$ are independent and the inequality follows from (\ref{eq:p-rho}). We will show next that conditioned on there being at most $\rho r/2$ bad positions, at most $\zeta' k$ heavy hitters get lost, where $\zeta'$ satisfies (\ref{eq:new-zeta-one-step}). Note that this implies that the overall failure probability $p'$ is upper bounded by the above probability, which then proves (\ref{eq:new-p-one-step}).

To bound $\zeta'$, call a heavy hitter $h\in [n]$ {\em corrupted} if $C(h)_j\not\in S_j$ for at least $\rho r/2$ of the {\em good} positions $j$. Note by the fact that $C$ is $(\rho,\ell,L)$-list recoverable, if $h$ is not corrupted then, it is going to be output by our identification algorithm. Thus, we overestimate $\zeta'$ by upper bounding the number of corrupted heavy hitters. A simple counting argument show that number of such corrupted heavy hitters cannot be more than $2\zeta k/\rho$, which proves (\ref{eq:new-zeta-one-step}).
\end{proof}

Recall that an $(r=b/(1-2\rho)+1,n)_{\sqrt[b]{n}}$ RS code by Lemma~\ref{lem:rs-2}, is an $(\rho,\ell,\ell^r)$-list recoverable. Further, recall that such a code is $b$-wise independent. If $\rho<2/5$, then we will have Since $b>\rho r/2$, we have sufficient independent to use this code as $C$ in Lemma~\ref{lem:lr-weak-one-step}.

Using the RS code above with $\rho=1/4$ in Lemma~\ref{lem:lr-weak-one-step} one can prove Lemma~\ref{cor:rs}. 

\begin{proof}[Proof Sketch of Lemma~\ref{cor:rs}]
The proof is very similar to that of Theorem~\ref{thm:recursive-exp} and here we just sketch how the current proof differs from the earlier one. As in Theorem~\ref{thm:recursive-exp}, we start off with
$\{\M_n\}_{n\ge \zeta^{-2}}$--- a family of $m(n)\times n$ matrices from Corollary~\ref{cor:weak-identification-family} that are $(k,\zeta,\eta)$-weak identification matrices with $\left(\ell\stackrel{def}{=}O(k/\eta),p(n)\stackrel{def}{=}(n/k)^{-\Omega(\zeta k)}\right)$-guarantee where $m(n)=g(\zeta,\eta)\cdot k\cdot\log(n/k)$ for some function $g(\zeta,\eta)$. 

However, unlike Theorem~\ref{thm:recursive-exp}, we will pick a specific family of code. In particular, let $\{C_n\}_{n\ge 1}$ be the $(r\stackrel{def}{=}2b+1,n)_{\sqrt[b]{n}}$ Reed-Solomon code: we will later pick $b$ to be $O(2^{1/\eps})$. Note that Lemma~\ref{lem:rs} implies that this code will be $(\rho\stackrel{def}{=}1/4,\ell,\ell^{2b+1})$-list recoverable in time $O(\ell^{2b+1} b^2\log^2{n})$.

Other than the specific family of codes, the rest of the construction is exactly the same as in the proof of Theorem~\ref{thm:recursive-exp} to obtain
for large enough $N$ a $m'(N)\times N$ matrix $\M^*$ that is $(k,\zeta',\eta)$-weak identification matrix with $(O(k/\eta),p'(N))$-guarantee and identification time complexity $D(k,N)$ as follows. (The analysis is pretty much the same as in the proof of Theorem~\ref{thm:recursive-exp} except in the analysis of $\zeta'$ and $p'$ in the ``recursive" step instead of using the arguments in proof of Theorem~\ref{thm:recursive-exp}, we use Lemma~\ref{lem:lr-weak-one-step}.)

In what follows let $A\ge \Omega(\Mlb)$ satisfy the lower bound in Lemma~\ref{lem:intermediate}.
There exists two integers $h\le O(\log_r\log_A{N})$
and
$\N=(\log_{A}{N})$ such that the following hold:

\begin{equation}
\label{eq:zeta-RS}
\zeta'\le \zeta\cdot \left(\frac{2}{\rho}\right)^h
\end{equation}
\begin{equation}
\label{eq:failure-RS}
p'(N)\le  \left(\frac{A}{k}\right)^{-\Omega(\zeta\cdot k\cdot (\rho/4)^h\cdot r^h )},
\end{equation}
\begin{equation}
\label{eq:measurements-RS}
m'(N) \le O\left(g(\zeta,\eta)\cdot k\cdot \log(N/k)\cdot \N^{\frac{\log{\frac{r}{b}}}{\log{b}}}\right),
\end{equation}
and
\begin{equation}
\label{eq:decoding-RS}
D(k,N)=O\left(\zeta^{-3}\eta^{-2}\cdot \N\cdot A\cdot \log{A}\right)+\sum_{j=0}^{h-1} \left((r/b)^j O((k/\eta)^r\log^2{N})+ O\left(\zeta^{-3}\eta^{-2}\cdot L\cdot \log{\sqrt[b^j]{N}}\right)\right).
\end{equation}

Noting that $r^h=\log_A{N}$ and that $(O(1)/\rho)^h=\N^{O(\log(1/\rho)/\log{r})}$ along with the choice of $A=\zeta^{-6}\eta^{-2}\cdot k^{r}\cdot \log^2(N/k)$ and some simple manipulation implies the following parameters:
\begin{equation}
\label{eq:zeta-RS-2}
\zeta'\le \zeta\cdot \N^{O\left(\frac{\log(1/\rho)}{\log{r}}\right)}
\end{equation}
\begin{equation}
\label{eq:failure-RS-2}
p'(N)\le \left(\frac{N}{k}\right)^{-\Omega(\zeta k/\N^{O\left(\frac{\log(1/\rho)}{\log{r}}\right)})},
\end{equation}
\begin{equation}
\label{eq:measurements-RS-2}
m'(N) \le O\left(g(\zeta,\eta)\cdot k\cdot \log(N/k)\cdot \N^{\frac{\log{\frac{r}{b}}}{\log{b}}}\right),
\end{equation}
and
\begin{equation}
\label{eq:decoding-RS-2}
D(k,N)=O\left(\zeta^{-3}\eta^{-2}\cdot \N\cdot A\cdot \log{A}\right)+ \zeta^{-3}\eta^{-2}(k/\eta)^r\cdot\poly(\log{N}).
\end{equation}

To obtain the claimed parameter 
$p'$ one has to apply Lemma~\ref{lem:weak-identification-amplify} with $s=\N^{O\left(\frac{\log(1/\rho)}{\log{r}}\right)}/\zeta$. Finally, to obtain the claimed value of
$\zeta$ one has to re-scale $\zeta$ by a factor of $1/\N^{O\left(\frac{\log(1/\rho)}{\log{r}}\right)}$. (To verify the dependence of these parameters on $\eps$, recall that we picked $\rho=1/4$, $b=O(2^{1/\eps})$ and $r=2b+1$.)
\end{proof}

\section*{Acknowledgments}

We thanks anonymous reviewers for their helpful comments. In particular, we thank an anonymous reviewer who pointed out a proof of Lemma~\ref{lem:repeat} that achieves better parameters than our previous proof. Thanks to Yi Li and Mary Wootters for helpful discussions.

\appendix

\section{Threshold for failure probability}
\label{app:repeat}

We will prove the following result:

\begin{lemma}
\label{lem:repeat}
If there exists a $(k,C)$-top level system with $m$ measurements and failure probability $p$ then for every integer $g\ge 1$, there exists a $(k,\sqrt{3}\cdot C)$-top level system with $m\cdot s$ measurements and $p^{\Omega(s)}$ failure probability.
\end{lemma}

The above follows easily from the following result:
\begin{lemma}[\cite{MP12}, Lemma 5.3]
\label{lem:median-eric}
Let $\vy_1,\dots,\vy_s\in\R^n$ be independent random vectors such that $\pr[\|\vy_i\|_2<D]\ge 1-p$ for every $i\in [s]$. Let $\vy\in\R^n$ be obtained by taking the component-wise median of $\{\vy_1,\dots,\vy_s$. Then
\[\pr[\|\vy\|<\sqrt{3}\cdot D]\ge 1-(4ep)^{s/4}.\]
\end{lemma}

In the rest of the section we prove Lemma~\ref{lem:repeat} using Lemma~\ref{lem:median-eric} in the obvious way. Let $\M,\algo$ be the given top level system. The new top level system $\M',\algo'$ is defined as follows. $\M'$ is just $s$ independent copies of $\M$ (call the copies $\M_1,\dots,\M_s$). The new decoding algorithm is defined as follows: let $\vz_1,\dots,\vz_s$ be the output of $\algo$ on the $s$ copies of $\M$, i.e. $\vz_i=\algo(\M_i\vx)$.
Then $\algo'(\M'\vx)$ is defined to be the component wise median of $\vz_1,\dots,\vz_s$. Lemma~\ref{lem:repeat} follows by applying Lemma~\ref{lem:median-eric} with $\vy_i=\vx-\vz_i$ and $D=C\cdot \|\vx-\vx_k\|_2$, where $\vx$ is the original signal.

\section{Expander Basics}
\label{app:exp-basics}

We will need the following three simple lemmas.

\begin{lemma}
\label{lem:exp-intersect}
Let $G:[N]\times [\ell]\rightarrow [M]$ be a $(2t,\eps)$-expander. Let $S,T\subseteq [N]$ such that $|T|=t$, $|S|\ge t$ and $S\cap T=\emptyset$, then
\[|\nbr(S)\cap\nbr(T)|\le |\edg(\nbr(S)\cap\nbr(T))| \le 4|S|\ell\eps.\]
\end{lemma}
\begin{proof}
The first inequality is obvious so we prove the second inequality. Divide $S$ into $b=\left\lceil \frac{|S|}{t}\right\rceil$ many disjoint subsets $S_1,\dots,S_b$ where all $S_i$'s but possibly $S_b$ have size exactly $t$. (For what follows,  by adding extra elements to $S$ if necessary, we will assume that $|S_b|=t$.) Note that $b\le 2|S|/t$ as $|S|\ge t$. 

Fix an $i\in [b]$ and note that as $G$ is a $(2t,\eps)$ expander, we have
\[|\nbr(T\cup S_i)|\ge 2t\ell(1-\eps).\]
As $|\nbr(S_i)|\le t\ell$, the above implies that
\[|\nbr(T)\setminus \nbr(S_i)|\ge t\ell(1-2\eps).\]
The above in turn implies that
\[|\edg(\nbr(T)\setminus \nbr(S_i)) |\ge |\nbr(T)\setminus \nbr(S_i)|\ge t\ell(1-2\eps).\]
Since $|\edg(T)|=t\ell$, the above implies that
\[|\edg(\nbr(T)\cap \nbr(S_i))|\le t\ell-t\ell(1-2\eps) =2\eps t\ell.\]
Finally, as $S$ is the disjoint union of $S_1,\dots S_b$, we have
\[|\edg(\nbr(T)\cap \nbr(S))|\le 2b\eps t\ell \le 4\eps |S|\ell,\]
as desired, where in the last inequality we have used the fact that $bt\le 2|S|$.
\end{proof}

\begin{lemma}
\label{lem:exp-rt-nbrs}
Let $a\ge 2$ be an integer and
let $G:[N]\times [\ell]\rightarrow [M]$ be a $(t,\eps< 1/(2a))$-expander. Let $R\subseteq [M]$ be any subset of size $|R|\le \gamma t\ell$ for some $\gamma>0$. Then there are $<2a\gamma t$ elements $i\in [N]$ such that $|\Gamma(i)\cap R|\ge \ell/a$.
\end{lemma}
\begin{proof}
For the sake of contradiction, assume that there exists a subset $L\subseteq [N]$ with $|L|=2a\gamma t$ such that for every $i\in L$, $|\Gamma(i)\cap R|>\ell/a$. Consider the neighbors $\Gamma(L)\setminus R$. Note that by assumption each $i\in L$ has at most $(1-1/a)\cdot \ell$ such neighbors. Thus, we have
\[ |\Gamma(L)| \le 2a(1-1/a)\gamma t\ell+|R| \le 2a(1-1/a)\gamma t\ell +\gamma t \ell = \left(1-\frac{1}{2a}\right)\cdot |L|\cdot \ell < (1-\eps)\cdot |L|\cdot \ell,\]
which contradicts the assumption that $G$ is a $(t,\eps<1/(2a))$-expander.
\end{proof}

The following is a by-product of the well-known ``unique neighborhood" argument (it also follows from the proof of Lemma~\ref{lem:exp-intersect}):

\begin{lemma}
\label{lem:exp-unique}
Let $G:[N]\times [\ell]\rightarrow [M]$ be a $(t,\eps)$-expander. Then for any 
subset $S\subseteq [N]$ with $|S|\le t$, we have that at most $2\eps |S|\ell$ 
vertices $j\in [M]$ each of which has more then one neighbor in $S$.
(A vertex in $[M]$ adjacent to exactly one vertex in $S$ is called a
``unique neighbor" of $S$.)
\end{lemma}

\section{Probability Basics}
\label{app:prob-basics}

Given any two vectors $\vx,\vy\in\R^n$, we will denote its dot-product by $\langle\vx,\vy\rangle=\sum_{i=1}^n x_i\cdot y_i$.

We will use the following well-known result:

\begin{lemma}
\label{lem:ams}
Let $\vx=(x_1,\dots,x_n)\in\R^n$ and $\vr=(r_1,\dots,r_n)\in\{-1,1\}^n$ be a pair-wise independent random vector. Then
\[\Pr_{\vr}\left[|\langle \vx,\vr\rangle| > a\cdot \|\vx\|_2\right] \le \frac{1}{a^2}.\]
\end{lemma}

We will also use the following form of the Chernoff bound:

\begin{lemma}
\label{lem:chernoff}
Let $X_1,\dots,X_n$ be random iid binary random variables where $p=\Pr[X_i=1]$. Then
\[\Pr\left[\sum_{i=1}^n X_i >a\right]\le \left(\frac{epn}{a}\right)^a,\]
for any $a>pn$. The above bound also holds if $X_1,\dots,X_n$ are $a$-wise independent. Further, if the random variables are $k<a$-wise independent, then the upper bound is $(epn/a)^k$.\footnote{The statement of the Chernoff bound is pretty standard. The claim on the independence, follows e.g. from Lemma 3 and Theorem 2 in~\cite{chernoff-lim}.}
\end{lemma}

\section{Coding Basics}
\label{app:code-basics}

In this section, we define and instantiate some (families) of codes that we will be interested in. We begin with some basic coding definitions.

We will call a code $C:[N]\rightarrow [q]^r$ be a $(r,N)_q$-code.\footnote{We depart from the standard convention and use the size of the code $N$ instead of its dimension $\log_q{N}$: this makes expressions simpler later on.}. Vectors in the range of $C$ are called its codewords. Sometimes we will think of $C\subseteq [q]^r$, defined the the natural way.

We will primarily be interested in \textit{list recoverable} codes. In particular,

\begin{defn} Let $N,q,r,\ell,L\ge 1$ be integers and $0\le \rho\le 1$ be a real number. Then an $(r,N)_q$ code $C$ is called a $(\rho,\ell,L)$-list recoverable if the following holds. Given any collection of subsets $S_1,\dots,S_r\subseteq [q]$ such that $|S_i|\le \ell$ for every $i\in [r]$, there exists at most $L$ codewords $(c_1,\dots,c_r)\in C$ such that $c_i\in S_i$ for at least $(1-\rho)n$ indices $i\in [r]$. Further, we will call such a code recoverable in time $T(\ell,N,q)$ if all such codewords can be computed within this time upper bound.
\end{defn}

Further, we will call an $(r,N)_q$ code to be \textit{uniform}, if for every $i\in [r]$ it is the case that $C(x)_i$ is uniformly distributed over $[q]$ for uniformly random $x\in [N]$.
In our construction we will require codes that are \textit{both} list recoverable and uniform. Neither of these concepts are new but our construction needs us to focus on parameter regime that is generally not the object of study in coding theory. In particular, as in coding theory, we focus on the case where $N$ is increasing. Also we focus on the case when $q$ grows with $N$, which is also a well studied regime. However, we consider the case when $r$ is a \textit{fixed}. In particular, our ``ideal" code should have the following properties:

\begin{enumerate}
\item $q$ should be as small as possible.
\item $r$ should be as small as possible.
\item $L$ should be as close to $\ell$ as possible.
\item $\rho$ should be as large as possible.
\item $T(\ell,N,q)$ should has poly-logarithmic dependence on both $N$ and $q$ (and at the same time have as close to linear dependence on $\ell$ as possible).
\end{enumerate}

Next, we present three codes that we will use in our constructions (that achieve some of the properties above but not all of the above).

\paragraph{Split code.} In this case we assume that $N$ is both a power of $2$ and a perfect square. Given a $x\in [N]$, think of it as a bit vector of length $\log{N}$. Let $x_1$ ($x_2$ resp.) be the first $\log{N}/2$ (last $\log{N}/2$ resp.) bits of $x$. Then define $C_{\splt}(x)=(x_1,x_2)$. This is of course a trivial code and we record its properties below for future use. The list recovery algorithm for this code is very simple: given $S_1$ and $S_2$ as the input, output $S_1\times S_2$.

\begin{lemma}
\label{lem:split}
The code $C_{\splt}$ is a uniform code that is $(0,\ell,\ell^2)$-list recoverable code. Further, it is recoverable in $O(\ell^2\log{N})$ time.
\end{lemma}

\paragraph{Loomis-Whitney code.} We now consider a code based on the well-known Loomis-Whitney inequality. Let $d\ge 2$ be an integer and assume that $N$ is a power of $2$ and $\sqrt[d]{N}$ is an integer. Given $x\in [N]$ think of it as $(x_1,\dots,x_d)\in [\sqrt[d]{N}]^d$. Further, for any $i\in [d]$ define $x_{-i}=(x_1,\dots,x_{i-1},x_{i+1},\dots,x_d)\in [\sqrt[d]{N}]^{d-1}$. Then define $C_{\lw(d)}(x)=(x_{-1},\dots,x_{-d})$. Note that for $d=2$, we get $C_{\splt}$. The Loomis-Whitney inequality shows that $C_{\lw(d)}$ is a $(0,\ell,\ell^{d/(d-1)})$-list recoverable code. We also show how to algorithmically achieve this bound in time $\tilde{O}(\ell^{d/(d-1)})$. This implies the following:

\begin{lemma}
\label{lem:lw-1}
The code $C_{\lw(d)}$ is a uniform code that is $(0,\ell,\ell^{d/(d-1)})$-list recoverable code. Further, it is recoverable in $O(\ell^{d/(d-1)}\log{N})$ time.
\end{lemma}

The above follows from Theorem~\ref{thm:LW-d-1}.

\paragraph{Reed-Solomon code.} Finally, we consider the well-known Reed-Solomon (RS)  codes. In particular, let $b\ge 1$ be an integer and let $q=\sqrt[b]{N}$ be a prime power and consider $C_{\rs}:[N]\rightarrow [\sqrt[b]{N}]^r$. There are known results on list recovery of Reed-Solomon codes but they need $b/r=O(1/\ell)$, which is too weak for our purposes. Instead we first note that we can always first output $S_1\times S_2\times\cdots\times S_r$ and then for each vector check whether it is within Hamming distance of $\rho n$ of some RS codeword. E.g., one can use the well-known Berlekamp Massey algorithm that does unique decoding for $\rho<1/2(1-b/r)$.\footnote{One could potentially use list decoding to recover from even more errors but that does not seem to buy much for our application.} Further, it is well-known that RS codes are linear codes and it is well-known that linear codes are also uniform. This implies the following:

\begin{lemma}
\label{lem:rs-2}
Let $\rho<1/2(1-b/r)$. Then
the code $C_{\rs}$ is a uniform code that is $(\rho,\ell,\ell^{r})$-list recoverable code. Further, it is recoverable in $O(\ell^{r}r^2\log^2{N})$ time.\footnote{The $r^2\log^2{N}$ factor follows from the fact that the Berlekamp Massey algorithm needs $O(r^2)$ operation over $\F_q$, each of which takes $O(\log^2{q})$ time.}
\end{lemma}

\section{Constructive Proof of Loomis Whitney Inequality}
\label{app:LW}

\subsection{Notations}

We begin with some notations. Let $\Sigma$ denote an arbitrary discrete
``alphabet." We will consider subsets $S\subseteq \Sigma^d$. 
For any subset $T\subseteq [d]$, we will use $S_T$ to denote the vectors in $S$ projected down to $T$. Further, for any integer $0\le i\le d$, we will use $\binom{[d]}{i}$ to denote the set of all subsets of $d$ of size exactly $i$.

Next, we define the ``join" operators. Given any two subsets 
$T_1,T_2\subseteq [d]$, and two sets of vectors $V_1\subseteq\Sigma^{T_1}$ 
and $V_2\subseteq \Sigma^{T_2}$ and any subset 
$G\subseteq \Sigma^{T_1\cap T_2}$, we define the {\em join of $V_1$ and $V_2$
over $G$}, denoted by $V_1\Join_G V_2$, to be the set of vectors 
$(u_1,a,u_2)$ in 
$\Sigma^{T_1\cup T_2}\equiv \Sigma^{T_1\setminus T_2}
\times\Sigma^{T_1\cap T_2}\times\Sigma^{T_2\setminus T_1}$, 
where $a\in G$, $(u_1,a)\in V_1$ and $(a,u_2)\in V_2$. 
The {\em join of $V_1$ and $V_2$} (without $G$ as an anchor), denoted by
$V_1 \Join V_2$, is the set of vectors $(u_1,a,u_2)$ in 
$\Sigma^{T_1\cup T_2}\equiv \Sigma^{T_1\setminus T_2}
\times\Sigma^{T_1\cap T_2}\times\Sigma^{T_2\setminus T_1}$, 
$(u_1,a)\in V_1$ and $(a,u_2)\in V_2$.
In particular, when $T_1 \cap T_2 = \emptyset$, $V_1 \Join V_2$ is simply
$V_1 \times V_2$ whose coordinates are indexed by $T_1\cup T_2$.

\subsection{Projections of size $d-1$}

\newcommand{\barS}{\overline{S}}

In this subsection, we will consider the case when the projections are over $[d]\setminus \{i\}$ for every $i\in [d]$. 
We will prove the following result:

\begin{theorem} 
\label{thm:LW-d-1}
Let $d\ge 1$ be an integer. 
For each $i \in [d]$, let $\bar S_i \subset \Sigma^{[d]\setminus\{i\}}$ be given
finite sets where $|\bar S_i| = k_i$ for some integer $k_i$.
Let $S\subseteq \Sigma^{[d]}$ be the set of vectors such that 
$S_{[d]\setminus\{i\}} \subseteq \bar S_i$
for every $i\in [d]$.  Then,
\begin{equation}
\label{eq:main-bound-proj-d-1}
|S| \le (d-1)\cdot\sqrt[d-1]{\prod_{i\in [d]} k_i}.
\end{equation}
Furthermore, the ``join" $S$ can be computed from the inputs 
$\{\barS_i\}_{i\in [d]}$ in time 
$\tilde{O}\left(d\cdot \sqrt[d-1]{\prod_{i\in [d]}k_i}+\sum_{i\in [d]} k_i
\right)$.
\end{theorem}

\subsubsection{Facts about certain labeled trees}

\newcommand{\T}{\mathcal{T}}
\newcommand{\lchld}[1]{{#1}_{\mathrm{L}}}
\newcommand{\rchld}[1]{{#1}_{\mathrm{R}}}
\newcommand{\lf}[1]{\mathcal{L}(#1)}

Let $\T$ be a binary tree. For every internal node $v\in \T$, let its 
left child be denoted by $\lchld{v}$ and its right child be denoted by 
$\rchld{v}$. Further, the subtree rooted at any node $v\in \T$ will be 
denoted by $\T(v)$. 
Finally, let $\lf{\T}$ denote the set of leaves in $\T$.

To prove Theorem~\ref{thm:LW-d-1}, we will consider the following 
labeled trees. Given $d\ge 2$, consider any binary tree $\T$ 
with $d$ leaves where each node $v \in \T$ is labeled with a 
subset $C(v) \subset [d]$. 
Without loss of generality, we use the numbers $\{1,2,\cdots,d\}$ to index 
the leaves $\lf{\T}$, i.e. each $\ell \in [d]$ is identified with a unique 
leaf node of $\T$.  The labeling is done as follows.
\begin{itemize}
\item Each leaf $\ell \in\lf{\T}=[d]$ is labeled with the set
$C(\ell) = [d]\setminus \{\ell\}$. 
\item Each internal node $v$ has label $C(v)=C(\lchld{v})\cap C(\rchld{v})$. 
\end{itemize}

We record the following simple properties of such labeled trees:
\begin{lemma}
\label{lem:label-tree-props}
Let $d\ge 2$ be an integer and let $\T$ be a binary tree with $d$ leaves 
labeled as above. Then the following are true:
\begin{itemize}
\item[(i)] For any internal node $v$, $C(\lchld{v})\cup C(\rchld{v})=[d]$; 
and
\item[(ii)] For the root $r$ of $\T$, $C(r)= \emptyset$.
\end{itemize}
\end{lemma}
\begin{proof} By induction, it is easy to see  that for any node $v$, 
\begin{equation}
\label{eq:set-label}
C(v)=\bigcap_{\ell\in\lf{\T(v)}} C(\ell).
\end{equation}
The above immediately implies $(ii)$ as 
$\bigcap_{\ell \in [d]} ([d]\setminus \{\ell \})= \emptyset$. 
For any leaf $\ell$ recall that 
$C(\ell)=[d]\setminus \ell$, which along with \eqref{eq:set-label} 
imply that $C(v)=[d]\setminus \lf{\T(v)}$.
This along with the fact that 
$\lf{\T(\lchld{v})}\cap \lf{\T(\rchld{v})}=\emptyset$ imply $(i)$.
\end{proof}

\subsubsection{Proof of Theorem~\ref{thm:LW-d-1}}

For notational convenience define
\[P=\sqrt[d-1]{\prod_{i\in [d]} k_i}.\]

We prove both parts by presenting an algorithm to compute $S$ from its 
potential projections $\bar S_i$, $i\in [d]$. 
Let $\T$ be an arbitrary labeled binary node with $d$ leaves as described in 
the last subsection. 
Every node $v$ will be associated with its labeled subset $C(v)$ as well as 
two auxiliary sets $\deter{v}\subseteq \Sigma^{[d]}$ and 
$\ndeter{v}\subseteq \Sigma^{C(v)}$. 
The set $\deter{v}$ is supposed to contain the candidate vectors some of which
will be members of the final output $S$. 
The set $\ndeter{v}$ will be a superset of the projection
$(S\setminus \deter{v})_{C(v)}$.

For each leaf $\ell \in \lf{\T}$, define $\deter{\ell} = \emptyset$
and $\ndeter{\ell} = \bar S_{\ell}$.
We next describe how to algorithmically compute the sets 
$\deter{v}$ and $\ndeter{v}$ for
internal nodes $v$ recursively from the leaves up to the root.

Let $v$ be any internal node in $\T$ whose left and right children's auxiliary
sets have already been computed. 
Without loss of generality, assume that 
$\ndeter{\lchld{v}}_{C(v)} = \ndeter{\rchld{v}}_{C(v)}$. 
(If not, we can simply remove the elements from $\ndeter{\lchld{v}}$ and 
$\ndeter{\rchld{v}}$ whose projections on to $C(v)$ lie in 
the symmetric difference 
$\ndeter{\lchld{v}}_{C(v)}\bigtriangleup\ndeter{\rchld{v}}_{C(v)}$.) 

When $v$ is {\em not} the root, partition $\ndeter{\lchld{v}}_{C(v)}$ into 
two sets $G$ and $B$ such that for every vector $u$ in the ``good" set $G$, 
the number of vectors in $\ndeter{\lchld{v}}$ whose projection onto $C(v)$ is 
$u$ is upper bounded by 
$\left\lceil \frac{P}{|\ndeter{\rchld{v}}|}\right\rceil - 1$.
The ``bad" set $B$ is $\ndeter{\lchld{v}}_{C(v)}\setminus G$. 
Finally, compute
\begin{eqnarray*}
\deter{v} &=& (\ndeter{\rchld{v}} \Join_G \ndeter{\lchld{v}}) 
\cup \deter{\lchld{v}} \cup \deter{\rchld{v}}\\
\ndeter{v}&=& B.
\end{eqnarray*}
When $v$ {\em is} the root, we compute
$\deter{v} = (\ndeter{\rchld{v}} \Join \ndeter{\lchld{v}})
\cup \deter{\lchld{v}} \cup \deter{\rchld{v}}$
and $\ndeter{v} =\emptyset$. 
By induction on each step of the algorithm, we will show that the following
three properties hold for every node $v \in \T$:
\begin{enumerate}
\item $(S\setminus \deter{v})_{C(v)} \subseteq \ndeter{v}$;
\item $|\deter{v}|\le (|\lf{\T(v)}|-1) \cdot P$; and
\item $|\ndeter{v}|\le \min\left( \min_{\ell\in \lf{\T(v)}} k_{\ell}, \frac{\prod_{\ell \in \lf{\T(v)}} k_{\ell}}{P^{|\lf{\T(v)}|-1}}\right)$.
\end{enumerate}

Assuming the above are true, we first complete the proof of the theorem. 
Let $r$ denote the root of the tree $\T$. By property 1, 
$(S\setminus \deter{r_L})_{C(r_L)} \subseteq \ndeter{r_L}$ and
$(S\setminus \deter{r_R})_{C(r_R)} \subseteq \ndeter{r_R}$.
Also recall that by Lemma~\ref{lem:label-tree-props} 
$C(r_L) \cup C(r_R) = [d]$ and $C(r_L) \cap C(r_R) = \emptyset$.
Hence, 
\[ S \setminus (\deter{r_L} \cup \deter{r_R}) \subseteq 
   \ndeter{r_L} \times \ndeter{r_R}  = 
   \ndeter{r_L} \Join \ndeter{r_R}.
\]
This implies $S \subseteq \deter{r}$. Thus, from $\deter{r}$ we can compute $S$
by keeping only vectors in $\deter{r}$ 
whose projection on any subset $L\in \binom{[d]}{d-1}$ is contained in $S_L$.
In particular, $|S| \leq \deter{r} \leq (d-1)P$, proving 
\eqref{eq:main-bound-proj-d-1}.

For the run time complexity of the above algorithm, we claim that for every 
node $v$, we need time $\tilde{O}(|\deter{v}|+|\ndeter{v}|)$. 
To see this note that for each node $v$, we need to do the following: 
\begin{itemize}
 \item[(i)] Make sure $\ndeter{\lchld{v}}_{C(v)}=\ndeter{\rchld{v}}_{C(v)}$, 
 \item[(ii)] Compute $G$ from $\ndeter{\lchld{v}}$, 
 \item[(iii)] Compute $\deter{v} = \ndeter{\rchld{v}} \Join_G 
\ndeter{\lchld{v}} \cup \deter{\lchld{v}} \cup \deter{\rchld{v}}$ and 
$\ndeter{v} =B$. 
\end{itemize}
It can be verified that all of these steps can be computed in time near-linear 
in the size of the largest set involved (after sorting the sets all the 
required computation can be done with a linear scan of the input lists), 
which along with property 3 leads to a (loose) upper bound of  
$\tilde{O}(P+ \min_{\ell \in \lf{\T(v)}} k_{\ell})$ on 
the run time for node $v$. 
Summing the run time over all the nodes in the tree gives the claimed run 
time.

To complete the proof we argue that properties 1-3 hold. For the base case, 
consider $\ell\in\lf{\T}$. Recall that in this case 
$\deter{\ell}=\emptyset$ and 
$\ndeter{\ell} = \bar S_{\ell}$. 
It can be verified that for this case, properties 1-3 hold. 

Now assume that properties 1-3 hold for all children of an internal node $v$.
We first verify properties 2-3 for $v$. From the definition of $G$, 
\[ |\ndeter{\rchld{v}} \Join_G \ndeter{\lchld{v}}|\le 
   \left(\left\lceil \frac{P}{|\ndeter{\rchld{v}}|}\right\rceil-1\right)
   \cdot |\ndeter{\rchld{v}}|\le P. 
\] 
From the inductive upper bounds on $\deter{\lchld{v}}$ and $\deter{\rchld{v}}$, 
property 2 holds at $v$. 
By definition of $G$ and an averaging argument, note that 
\[ |B|=|\ndeter{v}|\le|\ndeter{\lchld{v}}|\cdot 
   \frac{1}{\lceil P/|\ndeter{\rchld{v}}|\rceil}\le 
   \frac{|\ndeter{\lchld{v}}| \cdot |\ndeter{\rchld{v}}|}{P}. \]
From the induction hypotheses on $v_L$ and $v_R$, we have
$|\ndeter{\lchld{v}}|\le \frac{\prod_{\ell\in\lf{\T(\lchld{v})}} k_{\ell}}{P^{|\lf{\T(\lchld{v})}|-1}}$ and $|\ndeter{\rchld{v}}|\le \frac{\prod_{\ell\in\lf{\T(\rchld{v})}} k_{\ell}}{P^{|\lf{\T(\rchld{v})}|-1}}$, which implies that $|\ndeter{{v}}|\le \frac{\prod_{\ell\in\lf{\T({v})}} k_{\ell}}{P^{|\lf{\T({v})}|-1}}$. Further, it is easy to see that $|\ndeter{{v}}|\le \min(|\ndeter{\lchld{v}}|, |\ndeter{\rchld{v}}|)$, 
which by induction implies that 
$|\ndeter{v}|\le \min_{\ell\in \lf{\T(v)}} k_{\ell}$. Property 3 is thus
verified.

Finally, we verify property 1. By induction, we have 
$(S\setminus\deter{\lchld{v}})_{C(\lchld{v})} \subseteq
\ndeter{\lchld{v}}$ and
$(S\setminus \deter{\rchld{v}})_{C(\rchld{v})} \subseteq
\ndeter{\rchld{v}}$.
This along with the fact that $C(\lchld{v}) \cap C(\rchld{v}) = C(v)$ 
implies that $(S\setminus \deter{\lchld{v}} \cup \deter{\rchld{v}})_{C(v)} \subseteq \ndeter{\lchld{v}}_{C(v)}\cap \ndeter{\rchld{v}}_{C(v)} =B\uplus G$. 
Further, every vector in 
$(S\setminus\deter{\lchld{v}}\cup\deter{\rchld{v}})$ whose projection onto 
$C(v)$ is in $G$ also belongs to 
$\ndeter{\rchld{v}} \Join_G \ndeter{\lchld{v}}$. 
This implies that $(S\setminus \deter{v})_{C(v)}= B= \ndeter{v}$, 
as desired.

\subsection{Error-Tolerant Constructive Loomis-Whitney Inequality}

In this section, we prove the following  version of the Loomis-Whitney 
inequality that can handle certain ``errors," i.e. we are interested in the 
set $S$ such that a few projections need not lie in the given input 
projection sets. 
In particular, we will prove the following:

\begin{theorem}
\label{thm:LW-d-1-error}
Let $d\geq 2$  and $0\le e\le d-2$ be integers. 
Let $\bar S_i\subseteq \Sigma^{[d]\setminus\{i\}}$ be  sets of vectors such 
that for every $i\in [d]$,
$|\bar S_i| = k_i$ for some positive integer $k_i$. 
Let $S\subseteq \Sigma^{d}$ be the largest set such that for every 
vector $(v_1,\dots,v_d)\in S$, there are at least $n-e$ values of $i \in [d]$ 
for which $(v_1, \dots, v_{i-1},v_{i+1},\dots,v_d) \in \bar S_i$. 
Then,
\begin{equation}
\label{eq:main-bound-proj-d-1-error}
|S| \le \sum_{i=0}^e \sum_{B\in \binom{[d]}{i}} (d-i-1) \cdot \left(\prod_{j\in [d]\setminus B} k_j\right)^{\frac{1}{d-i-1}}.
\end{equation}
Further $S$ can be computed from the projections 
$\{\bar S_i\}_{i\in [d]}$ in time 
\[\tilde{O}\left( \sum_{i=0}^e \sum_{B \in \binom{[d]}{i}} 
 (d-i-1)\left(\prod_{j\in [d]\setminus B} k_j\right)^{\frac{1}{d-i-1}}+
 \sum_{\ell=1}^{e}\binom{d-1}{\ell-1}\sum_{j\in [d]} k_j\right).
\]
\end{theorem}
\begin{proof}
For each potential ``error set" $B \subset [d]$ with the number of errors
$|B| \leq e$, we apply the algorithm from Theorem \ref{thm:LW-d-1} to join
all $\bar S_\ell$, for $\ell \in [d] \setminus B$. 
The algorithm is identical to that of Theorem \ref{thm:LW-d-1}, except for
two facts. First, the tree $\T$ now only has $d-|B|$ leaves, each identified
by a member $\ell \in [d]\setminus B$.
Each leaf $\ell$ has label $C(\ell) = [d]\setminus \{\ell\}$ as before.
Second, the product $P$ is now defined to be 
\[ P = \left(\prod_{j\in [d]\setminus B} k_j\right)^{\frac{1}{d-|B|-1}}. \]
Note also that the label $C(r)$ of the root $r$ of $\T$ is no longer the
emptyset; however, this fact does not change the analysis one bit.
\end{proof}

\section{Omitted Material from Section~\ref{sec:lb}}

Note that
$\mathbf P = \mathbf P^2 = \mathbf P^T\mathbf P$ because any orthogonal
projection matrix is symmetric and idempotent.
Hence, for any two vectors $\mathbf x, \mathbf y \in \mathbb R^N$,
\[ \langle \mathbf{Px}, \mathbf y \rangle = \langle \mathbf x, \mathbf{Py} \rangle = \langle \mathbf x, \mathbf P^T\mathbf{Py}\rangle
 = \langle \mathbf{Px}, \mathbf {Py} \rangle.
\]
In particular,
$\langle \mathbf P\mathbf e_j, \mathbf e_j\rangle = \|\mathbf{Pe}_j\|_2^2$
for any $j\in [N]$.
The following was proved in \cite{MR2449058}. We provide here a very short
proof.

\begin{prop}
Let $\Phi$ be an arbitrary real matrix of dimension $m\times N$.
Then, there exists $j^* \in [N]$ such that
\[ \|\mv P\mv e_{j^*}\|_2^2 =
   \langle \mathbf{Pe}_{j^*}, \mathbf e_{j^*} \rangle \geq  1 - m/N. \]
\label{prop:DeVore-non-flat}
\end{prop}
\begin{proof}
Since $\langle \mathbf P\mathbf e_j, \mathbf e_j\rangle$ is precisely the
$j$th diagonal entry of the matrix $\mathbf P$, we have
$\mbox{trace}(\mathbf P) =
   \sum_{j=1}^N\langle \mathbf P\mathbf e_j, \mathbf e_j\rangle$.
But the trace of an orthoprojector is the dimension of the target
space which is $N-m$ in this case. Hence,
$N-m = \sum_{j=1}^N\langle \mathbf P\mathbf e_j, \mathbf e_j\rangle$,
which completes the proof.
\end{proof}

\subsection{The forall case}
\label{app:forall}

\begin{cor}[Cohen-Dahmen-DeVore \cite{MR2449058}]
\label{cor:l2l2-lb-N}
Let $\Phi$ be an $m \times N$ $\ltlt$ ``forall'' sparse recovery
measurement matrix with $k\geq 1$, i.e. there exists a decoding
algorithm $A$ and a constant $C\geq 1$ such that
for any input signal $\mathbf x \in \R^N$,
\begin{equation}
 \|\mv x - A(\Phi\mv x)\|_2 \leq C \cdot \|\mv x - \mv x_k\|_2
\label{eqn:forall-guarantee}
\end{equation}
then it must be the case that $m\geq N/C^2$.
\end{cor}
\begin{proof}
Let $\mv y = A(\Phi \mv 0) = A(\mv 0)$. Then,
By applying \eqref{eqn:forall-guarantee} with $\mv x = \mv 0$, it is easy to see
that $A(\mv 0) = A(\mv \Phi \mv 0) = \mv 0$.
Next, from Proposition \ref{prop:DeVore-non-flat} there exists $j^*\in [N]$
such that
$\langle \mv P \mv e_{j^*}, \mv e_{j^*} \rangle \geq 1 - m/N$.
Let $\mv x = \frac{\mv P\mv e_{j^*}}{\|\mv P\mv e_{j^*}\|_2}$ then
$\|\mv x\|_2 = 1$ and $\langle \mv x, \mv e_{j^*}\rangle^2 \ge
\|\mv P\mv e_{j^*}\|_2^2$.
Moreover, $A(\Phi \mv x) = A(\mv 0) = \mv 0$
because $\mv P\mv e_{j^*}$ is in the null space of $\Phi$.
Consequently, from \eqref{eqn:forall-guarantee} we obtain
\[
 1 = \|\mv x\|_2^2
\leq C^2 \cdot \|\mv x - \mv x_k\|^2_2
\leq C^2 \cdot \left(\sum_{j\neq j^*} x_j^2\right)
= C^2 \cdot \left( 1 - \langle \mv x, \mv e_{j^*} \rangle^2\right)
= C^2 \cdot \left( 1 - \|\mv P\mv e_{j^*}\|_2^2\right)
\leq C^2 \cdot m/N,
\]
which is the desired result.
\end{proof}

\subsection{Proof of Theorem~\ref{thm:main-lb}}
\label{app:yao}

To complete our lowerbound proof, we need a ``continuous'' version of Yao's
minimax principle.

	\begin{defn}
		A sparse recovery system $S = (\Phi,A)$ is defined to be a pair consisting of a measurement matrix $\Phi \in \R^{m \times N}$ and a mapping $A:\R^m \to \R^N$ .  We assume that the mapping algorithm $A$ has a finite description length, is deterministic, and is the best mapping for the matrix $\Phi$. 
	\end{defn}
We consider sparse recovery systems $(\Phi,A)$ in a compact set $Y$ (that is, our matrices $\Phi$ are in a compact set in $\R^{m \times N}$).  Let ${\mathcal R}$ be a probability measure on sparse recovery systems in $Y$, by which we mean ${\Rdist}$ specifies a distribution on measurement matrices $\Phi$ (and the mapping $A$ is the best possible deterministic, finite description length mapping for that distribution on $\Phi$).  Let ${\mathcal Y}$ be a compact set of convex combinations of probability measures $\Rdist$ on sparse recovery systems in $Y$.

Let us assume that our input vectors $\vx \in X$ and that $X$ is a compact subset of $\R^N$.  Let $\D$ be a probability measure on input vectors in $X$ and let ${\mathcal X}$ be a convex set of all such measures.

\begin{defn}
	We say that a sparse recovery system $S = (\Phi, A)$ decodes $\vx$ correctly if
	\[
		\|\vx - A(\Phi \vx)\|_2 \leq C \|\vx - \vx_k\|_2.
	\]
	We define the cost of a sparse recovery system $S$ on input $\vx$ as
	\[
		\cost(\vx, S) = \begin{cases}
			1 & \text{if $A$ does not decode $\vx$ correctly} \\
			0 & \text{otherwise}.
		\end{cases}
	\]
\end{defn}
Thus, the failure probability of a randomized sparse recovery system 
$S$ on input $\vx$ is
\[
	\cost(\vx,\Rdist) = \E_{S \sim \Rdist} \Big( \cost(\vx,S) \Big).
\]

\begin{defn}
 Let $\D \in \X$ and $\Rdist \in \Y$ and define
 \[
 	f(\D,\Rdist) = \cost(\D,\Rdist) = \E_{\vx \sim \D} \Bigg( \E_{S \sim \Rdist}
\Big( \cost(\vx, S) \Big) \Bigg) =  \E_{S \sim \Rdist} \Bigg( \E_{\vx \sim \D}
\Big( \cost(\vx, S) \Big) \Bigg).
 \] 
\end{defn}
Observe that $f$ is bi-linear and always finite (hence, it's a proper convex and concave function).  Furthermore, we can change the order of expectation by Fubini's theorem. 

\begin{lemma}[Continuous Yao's Lemma]
	With the compact, convex sets $\X$ and $\Y$ and the function $f: \X \times \Y \to \R$ defined above, 
	\[
		\max_{\D \in \X} \min_{S \in Y} \cost(\D, S) =  \min_{\Rdist \in \Y}
\max_{\vx \in X} \cost(\vx, \Rdist).
	\]
\end{lemma}
\begin{proof}
First, we observe that the hypotheses of the theorem match those of Sion's Minimax Theorem~\cite{Sion:58}; hence, we immediately have
\[
	\max_{\D \in \X} \min_{\Rdist \in \Y} \cost(\D,\Rdist) = \min_{\Rdist \in \Y} \max_{\D \in \X} 
	\cost(\D,\Rdist).
\]
(in particular, we have used the fact that $\X$ and $\Y$ are compact to ensure that the suprema and infima are attained.)

To finish the proof, we argue that for all distributions on recovery systems $\Rdist' \in \Y$,
\[
	\max_{\D} \cost(\D,\Rdist') = \max_{\D} \E_{\vx \sim \D} \E_{S \sim \Rdist'}
\Big( \cost(\vx, S) \Big) = 
	\max_{\vx \in X} \E_{S \sim \Rdist'} \Big( \cost(\vx, S) \Big)
\]
and that for all distributions on inputs $\D' \in \X$,
\[
	\min_{\Rdist} \cost(\D', \Rdist) = \min_{\Rdist} \E_{S \sim \Rdist} \E_{\vx \sim \D'} 
	    \Big( \cost(\vx, S) \Big) = \min_{S \in Y} \E_{\vx \sim \D'} 
	    \Big( \cost(\vx, S) \Big)
\]
from the convexity of $\X$ and $\Y$ and the compactness of $X$ and $Y$.
\end{proof}

So, if we find some distribution $\D'$ on inputs for which the best 
sparse recovery system has failure probability at least $p$ (i.e. high cost), 
then we have established a lower bound on the failure probability for 
randomized recovery systems:
\begin{align*}
	p & \leq \min_{\Rdist \in \Y} \cost(\D', \Rdist) \\
	  & \leq \max_{\D \in \X} \min_{\Rdist \in \Y}  \cost(\D,\Rdist) \\
	  & = \max_{\D \in \X} \min_{S \in Y} \cost(\D, S) \\
	  & = \min_{\Rdist \in \Y} \max_{\D \in \X} 	\cost(\D,\Rdist) \\
	  &= \min_{\Rdist \in \Y} \max_{\vx \in X} \cost(\vx, \Rdist).
\end{align*}
In Lemma \ref{lmm:randomx}, we exhibited a hard distribution on input vectors 
$\vx$ for which the best sparse recovery system has failure probability at
least $\sqrt{1/\gamma} \cdot e^{-\frac N 2 \ln(2/\gamma)}$,
given that $\gamma$ and $\delta=m/N$ satisfy \eqref{eqn:gammadelta}.
It is not hard to see that
$\gamma=\delta=\frac{1}{12+16C^2}$ satisfy \eqref{eqn:gammadelta}.
Hence, any foreach sparse recovery system with failure probability
at most $p = \sqrt{12+16C^2} \cdot e^{-\frac{\ln(6+8C^2)}{2} \cdot N}$
must have at least $m \geq \delta N = \frac{N}{12+16C^2}$ measurements.
In particular, we have shown that for failure probability 
$2^{-\Theta(N)}$, the number of measurements is $\Omega(N)$. 

We next give the simple reduction to handle the case of
larger failure probability 
$p>\sqrt{12+16C^2} \cdot e^{-\frac{\ln(6+8C^2)}{2}
\cdot N}$. Define 
$N'$
such that $p = \sqrt{12+16C^2} \cdot e^{-\frac{\ln(6+8C^2)}{2} \cdot N'}$,
i.e.
\[ N' = \frac{2}{\ln(6+8C^2)}\ln\left(\frac{\sqrt{12+16C^2}}{p}\right) = \Theta(\log(1/p)). \]
In this case for the
hard distribution, we zero out the last $N-N'$  entries in the input vectors 
and
then apply the hard distribution on the first $N'$ coordinates. This with the
previous result implies for failure probability at most $p$, 
we need $m = \delta N' = \Omega(\log(1/p))$ measurements, as desired.
We just proved Theorem~\ref{thm:main-lb}.

\section{Omitted Material from Section~\ref{sec:ub}}

\subsection{Known Results}
\label{app:known}

The following result from Cohen, Dahmen, DeVore~\cite{CDD2007:NearOptimall2l2} establishes a tight upper bound on the number of measurements with a polynomial time decoding algorithm for the foreach sparse recovery problem.  The algorithm $\algo$ is Orthogonal Matching Pursuit (OMP) which runs in time $O(MNk)$.
\begin{theorem}
\label{thm:upper-poly}
	There is a distribution on $M \times N$ matrices $\Phi$ such that for each $\vx \in \R^N$, the output of OMP after $2k$ iterations satisfies
	\[
		\|\vx - \algo(\Phi \vx) \|_2 \leq C \|\vx - \vx_k\|_2
	\]
	with probability larger than $1 - p$ provided that 
	$M \geq C (k \log(N/k) + \log(1/p))$.
\end{theorem}
Cohen, et al. provide three examples of such distributions on matrices: (i) iid Gaussian random entries $\Phi_{i,j}$ with variance $1/M$, (ii) iid Bernoulli random entries $\Phi_{i,j}$ with values $\pm 1/\sqrt{M}$, and (iii) columns of $\Phi$ drawn from a uniform distribution on $S^{M-1}$.

\subsection{Proof of Lemma~\ref{lem:weak-identification-amplify}}
\label{app:weak-amplify}
\begin{proof}
We will use a simple repetition trick. $\M'$ is just $s$ copies of $\M$, where each copy gets fresh random bits. The decoding algorithm is as follows: given the outputs $I_1,\dots,I_s$ from the $s$ copies of $\M$, output an $i\in \cup_{j=1}^s I_j$ if it appears in $>s/2$ $I_j$'s. Next we argue that the claimed bounds hold.

First note that since $|\cup_{j=1}^s I_j|\le s\ell$ and each $i$ that is output appear in $>s/2$ intermediate outputs $I_j$, we can only output $<2\ell$ such indices. Next note that by Lemma~\ref{lem:chernoff}, except with probability $p^{\Omega(s)}$, at most $s/4$ outputs $I_j$ miss more than $\zeta k$ elements from $H_k(\vx)$. Call the remaining (at least $3s/4$) intermediate outputs to be {\em good}. Note that we will not lose an $i$ if it appears in $>s/2$ good $I_j$'s. Then a simple counting argument implies that we can have at most $3\zeta k$ that are missing form at least $s/4$ good $I_j$'s.
\end{proof}

\section{Proof of Theorem~\ref{thm:l2l2-exp}}
\label{app:thm-l2l2-exp}

Let $\vu=(u_1,\dots,u_M)=\M \vx$. For any $i\in [N]$, define it's estimate $\bar{x}_i$ to the median of the values $\{ \M_{i,j}\cdot u_j\}_{j\in \Gamma(i)}$. We will show that 

\begin{lemma}
\label{lem:estimate} The following holds with probability $1-\binom{N}{\gamma k}^{-s}$.
Except for $\gamma k$ positions $i\in [N]$, every other index will have a good estimate: i.e., 
\[|x_i-\bar{x}_i|\le \sqrt{\eta/k}\cdot\|\vz\|_2.\]
\end{lemma}  

In this section, we prove Theorem~\ref{thm:l2l2-exp} using Lemma~\ref{lem:estimate}. The proof is almost identical to the similar one in~\cite{ely-martin} (except $\hat{\vx}$ has a larger support size).

Let $\bar{\vx}=(\bar{x}_1,\dots,\bar{x}_N)$ and define $\hat{\vx}$ to be $\bar{\vx}$ with all but the top $k'=k+k/\sqrt{\eta}$ entries (by their absolute values) zeroed out. Further, let $T$ be the set of items in $[N]$ that do not have a good estimate (i.e. $|x_i-\bar{x}_i|>\sqrt{\eta/k}\cdot\|\vz\|_2$). Note that by Lemma~\ref{lem:estimate}, $|T|\le \gamma k$. 

To complete the proof of Theorem~\ref{thm:l2l2-exp}, we prove the following:

\begin{lemma} There exists a vector $\hat{\vy}$ with $|\supp(\hat{\vy})|\le 2\gamma k$ such that for $\hat{\vz}=\vx-\hat{\vx}-\hat{\vy}$, we have
\begin{equation}
\label{eq:new-noise}
\|\hat{\vz}\|_2^2 \le \left(1+22\sqrt{\eta}\right)\cdot\|\vz\|_2^2.
\end{equation}
\end{lemma}
\begin{proof}
We'll prove the above by case analysis and adding up the contribution of different indices to $\|\hat{\vz}\|_2^2$ (and also define the elements in $\hat{\vy}$ in the process):
\begin{enumerate}
\item ($i\in\supp(\hat{\vx})$ with a good estimate) In this case, each such item contributes $\frac{\eta}{k}\cdot\|\vz\|_2^2$ to $\|\hat{\vz}\|_2^2$. Since $|\supp(\hat{\vx})|=k(1+1/\sqrt{\eta})$, items in this case contribute at most $2\sqrt{\eta}\cdot\|\vz\|_2^2$. (Set $\hat{y}_i=0$.)
\item ($i\in\supp(\hat{\vx})$ with a bad estimate) These items do not contribute anything. (Set $\hat{y}_i=x_i-\hat{x}_i$.)
\item ($i\in \supp(\vz)\setminus\supp(\hat{\vx})$) These contribute at most $\|\vz\|_2^2$. (Set $\hat{y}_i=0$.)
\item ($i\in H_{k'}(\vx)\setminus \supp(\hat{\vx})$ with a good estimate is displaced by an item $i'\in \supp(\hat{\vx})\setminus H_{k'}(\vx)$ with a good estimate) In this case set $\hat{y}_i=\hat{y}_{i'}=0$. Thus, in this case, we have $\hat{x}_{i}=0$ and $\hat{z}_i=x_i,\hat{z}_{i'}=x_{i'}-\hat{x}_{i'}$. Now note that by definitions of $i$ and $i'$, we have 
\[|x_{i'}|\ge |\hat{x}_{i'}|-\sqrt{\eta/k}\|\vz\|_2 \ge |\hat{x}_i|-\sqrt{\eta/k}\|\vz\|_2 \ge |x_i|-2\sqrt{\eta/k}\|\vz\|_2.\]
 Next, note that any item $i''$ such that $|x_{i''}|> \sqrt{\eta/k}\|\vz\|_2$ would satisfy $i''\in H_{k'}(\vx)$. This in turn implies that $|x_{i'}|\le \sqrt{\eta/k}\|\vz\|_2$. This along with the above inequalities implies that $|x_i|\le 3\sqrt{\eta/k}\|\vz\|_2$. This in turn implies that
\[\hat{z}_i^2+\hat{z}_{i'}^2 \le \frac{9\eta}{k}\cdot\|\vz\|_2^2+\frac{\eta}{k}\cdot\|\vz\|_2^2 = \frac{10\eta}{k}\cdot\|\vz\|_2^2.\]
Since there are at most $k+k/\sqrt{\eta}$ such pairs $(i,i')$, the total contribution from this case is at most $20\sqrt{\eta}\|\vz\|_2^2$.
\item ($i\in H_{k'}(\vx)\setminus \supp(\hat{\vx})$ with a bad estimate or is displaced by an item with a bad estimate) These do not contribute anything. (Set $\hat{y}_i=x_i$.) 
\end{enumerate}

It can be verified that the following three things holds: (i) $\vx=\hat{\vx}+\hat{\vy}+\hat{\vz}$. (ii) Since every item with bad estimate can contribute at most two non-zero items to $\hat{\vy}$ (once in item $2.$ and once in item $5.$), we have that $|\supp(\hat{\vy})|\le 2\gamma k$. (iii) Finally, by items $1,3$ and $4$, (\ref{eq:new-noise}) is satisfied.
\end{proof}

\paragraph{Some Remarks.} We first note that it is possible to compute $\hat{\vx}$ efficiently in time near linear in $N$ (as it involves computing $N$ median values and then outputting the top $k+k/\sqrt{\eta}$ values).

The argument in item 4 in the proof above also implies the following:
\begin{cor}
\label{cor:l2l2-exp-weak-identification}
$\supp(\hat{\vx})$ contains all but $\gamma k$ items $i\in [N]$ that satisfy $|x_i|> 3\sqrt{\eta/k}\cdot\|\vz\|_2$.
\end{cor}

In particular, consider the following algorithm.
\begin{enumerate}
\item[0.] The input is the vector $\M\vx$ and a subset $S\subseteq [N]$.
\item For each $j\in S$, compute the estimate $\bar{x}_j$ as the median of the values $\{|(\M\vx)_b|\}_{b\in\Gamma(j)}$.
\item Output in $I$ the items $j\in S$ with the top $k+k/\eta$ estimates $\bar{x}_j$.
\end{enumerate}

The above along with Corollary~\ref{cor:l2l2-exp-weak-identification} (where we substitute $\eta$ by $\eta^2$) implies Theorem~\ref{thm:l2l2-exp-identification}.

Finally, the proof of of Theorem~\ref{thm:l2l2-exp} also implies Theorem~\ref{thm:l2l2-exp-weak}.
\begin{proof}[Proof Sketch of Theorem~\ref{thm:l2l2-exp-weak}]
The proof is pretty much mimics the proof of Theorem~\ref{thm:l2l2-exp} except the following two small changes. First, we adjust the constants so that Lemma~\ref{lem:estimate} implies that the median estimate is correct for all but $\zeta k/2$ elements $i\i [N]$, we have that $\bar{x}_i$ is a good estimate of $x_i$. Finally, in the argument of proof of Theorem~\ref{thm:l2l2-exp}, we also have to take into account the at most $\zeta k/2$ elements of $H_{k+k/\eta}(\vx)$ that are missing from $S$. (Note that the algorithm to compute $\hat{\vx}$ is the one we used to prove Theorem~\ref{thm:l2l2-exp-identification}-- just output the top $k+k/\eta$ median estimates for the elements in $S$.)
\end{proof}

\subsection{Proof of Lemma~\ref{lem:estimate}}

We begin with some definitions and notation that will be useful in the proof. Let $\zeta>0$ be a real that we will fix later (to be $\Theta(\gamma)$) and let $\eps=\zeta^3\eta$.

We will call each element $j\in [M]$ a \textit{bucket}. $\Gamma(i)$ for some $i\in [N]$ will be the set of $i$'s buckets. Finally, we will call a bucket $j\in \Gamma(i)$ \textit{bad} for index $i\in [N]$ if at least one of the following conditions hold:

\begin{itemize}
\item (\badone) $j\in \Gamma(h)$ for some heavy hitter $h\neq i$.
\item (\badtwo) $j\in \Gamma(h)$ for some heavy tail element $h\neq i$.
\item (\badthree) The $\ell_2^2$ contribution of all light tail elements (other than $i$) to $j$ is $>\frac{\zeta\eta}{k}\cdot\|\vz\|_2^2$.
\item (\badfour) Define $\hat{\vx}_j=\sum_{b\in \light\setminus \{i\}\text{ and } j\in\Gamma(b)} (\M_{b,j}\cdot x_b)$. Then 
\[\|\hat{\vx}_j\|_2 > \sqrt{\frac{\eta}{k}}\cdot \|\vz\|_2.\]
\end{itemize}
If a bucket satisfies ({\sc Bad-}$b$) for some $b\in [4]$, then we will also call it {\sc Bad-}$b$-bucket for item $i$. If a bucket is not {\sc Bad-}$b$ for any $b\in [4]$, then we will call it \textit{good} for $i$. If we do not specify for which item a bucket is bad (or the more specific versions of bad as above) then, it'll be assumed to be bad for some $i\in [N]$. For any $i\in [N]$ and $b\in [4]$, let $\B_i^b\subseteq [M]$ denote the set of {\sc Bad}-$b$-buckets for item $i$ and for notational convenience define $\B_i=\B_i^1\cup\B_i^2\cup\B_i^3\cup\B_i^4$. 

Note that if for any $i\in [N]$, we have $|\Gamma(i)\cap \B_i|<\frac{\ell}{2}$, then $|x_i-\bar{x}_i|\le \sqrt{\eta/k}\cdot\|\vz\|_2$ (this is because then in the majority of the buckets, $x_i$ is the only potential heavy element and the tail noise is low),  as desired.\footnote{Actually we only need to consider buckets that are not \badone, not \badtwo\ or not \badfour\ but not being \badthree\ makes the analysis somewhat modular.} To prove that at most $\gamma k$ items have bad estimates we will prove the following:

\begin{lemma}
\label{lem:hh-decoy}
For any subset $S\subseteq [N]\setminus H_k(\vx)$ with $|S|=k$, the following are true
\begin{align}
\label{eq:hh-decoy-1}
|\cup_{i\in H_k(\vx)\cup S} \B_i^1| \le \frac{\gamma}{12} k\ell& \text{ with probability } 1.\\
\label{eq:hh-decoy-2}
|\cup_{i\in H_k(\vx)\cup S} \B_i^2| \le \frac{\gamma}{12} k\ell& \text{ with probability } 1.\\
\label{eq:hh-decoy-34}
|\cup_{i\in H_k(\vx)\cup S} \B_i^3\cup \B_i^4| \le \frac{\gamma}{12} k\ell& \text{ with probability at least } 1-\exp(-\Omega(\gamma k\ell)).
\end{align}
\end{lemma}

Let $\B=\cup_{i\in H_k(\vx)\cup S} \B_i$. Then note that Lemma~\ref{lem:hh-decoy} proves that except with probability $\exp(-\Omega(\gamma k\ell))$, 
$|\B|\le \frac{\gamma}{16}\cdot (4k)\ell$. This implies that by Lemma~\ref{lem:exp-rt-nbrs} (with $a=2$), there are at most 
$4(\gamma/16)(4k)=\gamma k$
elements $i\in H_k\cup S$, such that they have at least $\ell/2$ bad buckets for item $i$. This shows that for every $S$, the probability that there exists a subset $T\subset H_k(\vx)\cup S$ with $|T|=\gamma k$ such that every item in $T$ has a bad estimate is upper bounded by $\exp(-\Omega(\gamma k\ell))$. Taking the union bound over all the $\binom{N}{\gamma k}$ choices for $T$, the probability that there exists some set of $\gamma k$ items with a bad estimate is at most:
\[\binom{N}{\gamma k}\cdot e^{-\Omega(\gamma k \ell)} \le \binom{N}{\gamma k}\cdot\left(\binom{N}{\gamma k}\right)^{-s+1}= \binom{N}{\gamma k}^{-s},\]
where the inequality follows from the assumption that $\ell \ge c\cdot \log(N/k)$ for some large enough constant $c\ge \Omega(s\log(1/\gamma))$.
 Thus, except with probability $\binom{N}{\gamma k}^{-s}$,
other than at most $\gamma k$ items $i\in [N]$, every other item has a good estimate, as desired.

In the rest of the section, we will prove Lemma~\ref{lem:hh-decoy}.

\subsection{Proof of Lemma~\ref{lem:hh-decoy}}

\paragraph{Proof of (\ref{eq:hh-decoy-1}).} Fix any $i\in H_k(\vx)\cup S$ and consider a bucket $j\in \Gamma(i)$. If $j$ is \badone\ for $i$ then it means that $j$ is not a ``unique neighbor" of $H_k(\vx)\cup S$. By Lemma~\ref{lem:exp-unique}, we then get that
\[|\cup_{i\in H_k(\vx)\cup S} \B_i^1| \le 2\eps(2k)\ell = 4\zeta^3\eta k\ell \le \frac{\gamma}{16} k\ell < \frac{\gamma}{12} k\ell,\]
where the second inequality follows if we choose $\zeta\le \gamma/4$ (as $\gamma^3\le \gamma$ and $\eta\le 1$).

\paragraph{Proof of (\ref{eq:hh-decoy-2}).} We will make separate arguments for heavy tail elements that are in $S$ and those that are not. To be more precise let $H_{in}\subset S$ be the heavy tail elements in $S$ and $H_{out}\subseteq [N]\setminus H_{in}$ be the heavy tail items outside of $S$. We'll make slightly different arguments depending on whether $|H_{out}|< 2k$ or not:
\begin{itemize}
\item {\sf Case 1} ($|H_{out}|<2k$) In this case we will prove the following stronger inequality:
\[|\cup_{i\in H_k(\vx)\cup S\cup H_{out}} \B_i^2| \le \frac{\gamma}{12} k\ell.\]
Fix an $i\in H_k(\vx)\cup S\cup H_{out}$ and consider a bucket $j\in \Gamma(i)$ that is \badtwo\ for $i$. Then as in the previous case, we get a non-unique neighbor in $\Gamma(H_k(\vx)\cup S\cup H_{out})$. Note that $|H_k(\vx)\cup S\cup H_{out}|<4k$ and thus by Lemma~\ref{lem:exp-unique}, we have
\[|\cup_{i\in H_k(\vx)\cup S\cup H_{out}} \B_i^2| < 2\eps(4k)\ell = 8\zeta^3\eta k\ell \le \frac{\gamma}{27} k\ell < \frac{\gamma}{12} k\ell,\]
where the second inequality follows if we choose $\zeta\le \gamma/6$.

\item {\sf Case 2} ($|H_{out}|\ge 2k$) In this case we will handle collisions with heavy tail elements from $H_{out}$ separately. Fix an element $i\in H_k(\vx)\cup S$ and consider a bucket $j\in \Gamma(i)$. If $j$ is also in $\Gamma(h)$ for some heavy tail element $h\in H_{in}$, then call such a bucket light-\badtwo\-bucket for $i$. Otherwise if $i$ collides with some heavy tail element $h\in H_{out}$ in bucket $j$, then we call $j$ to be a heavy-\badtwo\-bucket for $i$. By the same argument as in the proof of (\ref{eq:hh-decoy-1}), we can upper bound the number of light-\badtwo\-buckets for any $i$ by
\begin{equation}
\label{eq:light-badtwo}
2\eps(2k)\ell = 4\zeta^3\eta k\ell \le \frac{4\gamma}{125} k\ell < \frac{\gamma}{24} k\ell,
\end{equation}
where the first inequality follows if we choose $\zeta\le \gamma/5$.

To bound the number of heavy-\badtwo\-buckets we note that this would be upper bounded by $|\Gamma(H_k(\vx)\cup S)\cap\Gamma(H_{out})|$. By Lemma~\ref{lem:exp-intersect}, this is upper bounded by
\begin{equation}
\label{eq:heavy-badtwo}
4\eps|H_{out}|\ell \le 4\eps\cdot\frac{k}{\zeta^2\eta}\cdot\ell =4\zeta k\ell \le \frac{\gamma}{24} k\ell,
\end{equation}
where the first inequality follows by noting that the total number of heavy tail items (which upper bounds $|H_{out}|$) is at most $k/(\zeta^2\eta)$ and the last inequality follows if $\zeta\le \gamma/96$.

Upper bounds in (\ref{eq:light-badtwo}) and (\ref{eq:heavy-badtwo}) proves (\ref{eq:hh-decoy-2}).
\end{itemize}

\paragraph{Proof of (\ref{eq:hh-decoy-34}).} We will prove (\ref{eq:hh-decoy-34}) by proving the following two inequalities. We will show that the following always holds:
\begin{equation}
\label{eq:hh-decoy-3}
|\cup_{i\in H_k(\vx)\cup S} \B_i^3| \le \frac{\gamma}{24} k\ell,
\end{equation}
and the following does not hold with probability at most $\exp(-\Omega(\gamma k\ell))$:
\begin{equation}
\label{eq:hh-decoy-4}
|\cup_{i\in H_k(\vx)\cup S} \B^4_i\setminus\B_i^3| \le \frac{\gamma}{24} k\ell.
\end{equation}

\paragraph{Proof of (\ref{eq:hh-decoy-3}).} We will prove (\ref{eq:hh-decoy-3}) using an argument very similar to that in~\cite{ely-martin}. (The only difference is that the argument for light tail with small support is less involved because of the use of expanders.) We will first compute the sum of the $\ell_2^2$ contribution of the light tail elements to $\Gamma(H_k(\vx)\cup S)$. The main idea is to decompose this sum into (sub)-convex combination of flat tail contributions and contribution from a tail with small support.

First, we consider the contribution of the light tail elements from $S$ itself that contribute to the potential \badthree\ buckets. We take care of this contribution with an argument similar to the one we used to prove (\ref{eq:hh-decoy-1}). In particular, if bucket $j$ is \badthree\-bucket for $i\in H_k(\vx\cup S)$ due to a light tail element from $S$ then bucket $j$ is not a unique neighbor. Thus, the number of such \badthree\-buckets is upper bounded (due to Lemma~\ref{lem:exp-unique}) by at most $4\zeta^3\eta k\ell\le \frac{\gamma}{54}k\ell <\frac{\gamma}{48}k\ell$ if we choose $\zeta\le \gamma/6$.

Next, let $L\subseteq \light\setminus S$ be the set of light tail elements with non-zero value outside of $S$. We will first bound the sum of the $\ell_2^2$ contribution of light tail elements from $L$ to $\Gamma(H_k(\vx)\cup S)$-- we denote the sum by $\Sigma_L$. Define $\vw=(x_i^2)_{i\in L}$. (Note that $\|\vw\|_1\le \|\vz\|_2^2$.)

We first assume that $|L|\ge 2k$. Let $m$ be the smallest $w_i$ value among all $i\in L$ and consider the vector $\vw_m$ where $\supp(\vw_m)=L$ and each non-zero value is $m$. (Note that $\vw_m$ is a flat tail scaled by $m|L|$). By Lemma~\ref{lem:exp-intersect}, the contribution of $\vw_m$ to $\Sigma_L$ is upper bounded by $4\eps\ell m|L|=4\eps\ell \|\vw_m\|_1$.  Update $\vw\leftarrow \vw-\vw_m$ and $L\leftarrow L\setminus \{i\in L|\vw=m\}$. If $|L|\ge 2k$ we repeat the process above. We note that the total contribution while $|L|\ge 2k$ in the process above is at most $4\eps \ell\|\vw\|_1\le 4\eps \ell\|\vz\|_2^2$. Finally, when we are left with $|L|\le 2k$, since each element in the residual $\vw$ can contribute at most $\zeta^2 \eta/k\cdot \|\vz\|_2^2$ (repeated at most $\ell$ times) to $\Sigma_L$, we conclude that the final contribution to $\Sigma_L$ is at most $2\zeta^2\eta\ell\|\vz\|_2^2$. This implies that
\[\Sigma_L\le (4\eps+2\zeta^2\eta)\ell\|\vz\|_2^2 \le 6\zeta^2\eta \ell\|\vz\|_2^2.\]

In the worst-case the sum $\Sigma_L$ contributes to new \badthree\-buckets. In particular, by the Markov argument (and the bound on $\Sigma_L$ above along with the fact that there are at most $2k\ell$ buckets), the number of new \badthree\-buckets is upper bounded by
$6\zeta k\ell \le \frac{\gamma}{48}k\ell$, if we pick $\gamma\le \zeta/288$.

Adding the contributions of light tail elements from $S$ and outside of $S$ proves (\ref{eq:hh-decoy-3}).

\paragraph{Proof of (\ref{eq:hh-decoy-4}).} Consider any bucket $j$ that is not a \badthree-bucket. This implies that the $\ell_2$ contribution from the light tail elements is at most $\sqrt{\zeta\eta/k}\cdot\|\vz\|_2$. Then by Lemma~\ref{lem:ams}, we have that with probability at most $\zeta$, we have $\|\hat{\vx}_j\|_2>\sqrt{\eta/k}\cdot\|\vz\|_2$ (where $\hat{\vx}_j$ is as defined earlier). This implies that the expected number of \badfour\-buckets is at most $\zeta k\ell$. Further, since every non-zero $\M_{i,j}$ is an independent $\pm 1$ value, we can apply Chernoff bound (Lemma~\ref{lem:chernoff}) to bound the probability of more than $2\zeta k\ell$ \badfour\-buckets by $\exp(-\Omega(\zeta k\ell))$. Thus, except with probability $\exp(-\Omega(\zeta k\ell))$ the number of \badfour\-buckets is upper bounded by
\[2\zeta k\ell \le \frac{\gamma}{24},\]
if we pick $\zeta\le \gamma/48$.

\paragraph{Wrapping up.} Looking at all the conditions on $\zeta$, we note that the choice $\zeta=\frac{\gamma}{288}$ satisfies all the required conditions. Note that this also implies that the probability of not satisfying (\ref{eq:hh-decoy-4}) is upper bounded by $\exp(-\Omega(\gamma k\ell))$, as desired.

\subsection{Some observation on the use of randomness}
\label{sec:rand-use-exp}
We start off by observing that the only place that uses randomness in the proof of Theorem~\ref{thm:l2l2-exp} is in Lemma~\ref{lem:hh-decoy}. Further, the only place in the proof of Lemma~\ref{lem:hh-decoy} that uses randomness is in the proof of (\ref{eq:hh-decoy-4}).

First note that by the dependence of the probability bound in Lemma~\ref{lem:chernoff}, we can argue the same probability dependence as in (\ref{eq:hh-decoy-4}) even if the random $\pm 1$ values were $O(k\ell)$-wise independent. In other words,
\begin{remark}
\label{rem:lim-indep-exp}
Theorems~\ref{thm:l2l2-exp},~\ref{thm:l2l2-exp-identification} and~\ref{thm:l2l2-exp-weak}, hold even if the random $\pm 1$ entries are $O(k\ell)$-wise independent.
\end{remark}

Next we note that if we were proving the corresponding $\lolo$ result to Theorem~\ref{thm:l2l2-exp}, then we would not need {\em any} randomness. In particular, for the $\lolo$ case, if we define \badthree\ event to be such that the total $\ell_1$ mass of all light tail elements other than $i$ is $>\zeta\eta/k\|\vz\|_1$, then if a bucket is not \badthree\, then we do not need to consider the \badfour\ event. In other words,
\begin{remark}
\label{rem:l1l1-exp}
The versions of Theorems~\ref{thm:l2l2-exp},~\ref{thm:l2l2-exp-identification} and~\ref{thm:l2l2-exp-weak} for $\lolo$ sparse recovery holds even without multiplying the matrix $\M_G$ with random $\pm 1$ values. In other words, the results hold {\em deterministically}.
\end{remark}

\section{Proof of Lemma~\ref{lem:intermediate}}
\label{app:intermediate}

The algorithm is the same as in the proof of Theorem~\ref{thm:l2l2-exp-identification}. 

The main idea in proving the correctness of the algorithm is to argue that under the map $f$, the heavy hitters do not suffer much collisions or ``acquire" heavy $\ell_2^2$ noise. Then we apply Corollary~\ref{cor:l2l2-exp-weak-identification} (or more precisely the proof of Lemma~\ref{lem:estimate}) on the vector obtained by ``applying" $f$ to $\vx$. In what follows, we will be using notation from the previous section.

We first consider the effect of $f$ on $\vx$. Call a heavy hitter $h\in H_k(\vx)$ being \textit{corrupted} if either $f(h)=f(h')$ for some heavy hitter/heavy tail element $h'$ or the $\ell_2^2$ sum contribution of light tail elements to $f(h)$ is $>\frac{\zeta^2 \eta}{k}\cdot\|\vz\|_2^2$. We next argue that 

\begin{lemma}
\label{lem:corruption}
Except with probability $\left(\frac{M}{k}\right)^{\Omega(-\alpha \zeta k)}$, at most $5\zeta k$ heavy hitters are corrupted.
\end{lemma}
\begin{proof} The argument is very similar to that in~\cite{ely-martin}.

To being with we upper bound the number of corruptions due to a collision. Just for this proof, we will refer to an element $i\in [N]$ as a heavy tail if $|x_i|^2>\frac{\zeta^5\eta^2\|\vz\|_2^2}{k}$. 

Note that even conditioned on the the maps of the heavy tail items (of which there are at most $\zeta^{-5}\eta^{-2} k$) and the other heavy hitters (of which there are at most $k-1$), a given heavy hitter suffer a collision with probability at most $O(k(\zeta^{-5}\eta^{-2}+1)/M^{1-\alpha})$. Then using a similar argument used in~\cite{ely-martin}, we can applying Lemma~\ref{lem:chernoff} to bound the probability of having more than $\zeta k$ corruptions due to collisions by
\begin{equation}
\label{eq:corruption-prob1}
\left(\frac{ek^2(1+\zeta^{-5}\eta^{-2})}{M^{1-\alpha}\zeta k}\right)^{\Omega(\zeta k)} \le \left(\frac{M}{k}\right)^{\Omega(-\alpha\zeta k)},
\end{equation}
where the inequality follows from the fact that $\frac{M^{1-\alpha}}{k\zeta^{-6}\eta^{-2}}\ge (M/k)^{\alpha}$ (which in turn follows from the lower bound $M\ge \Omega(\Mlb)$).

Next we upper bound the number of corruptions due to large $\ell_2^2$ noise from light tail items. Define $\vw\in\R^N$ be to zero on the locations of heavy hitters and heavy tail items. For a light tail element $i$, define $w_i$ be $x_i^2$ rounded up to the next highest number in $\{\|\vz\|_2^2/2^i\}_{i\ge 0}$. Note that all $i$ such that $|w_i|\le \zeta^5\eta^2\|\vz\|_2^2/N$ can contribute at most $\zeta^5\eta^2\|\vz\|_2^2$. 

Thus, ignoring such elements, since the largest value of a light tail element is $\zeta^5\eta^2|\vz\|_2^2/k$, we are left with $n'=\log\left(\zeta^5\eta^2 N/k)=O(\log(N/k)\right)$ distinct values in $\vw$. In particular, we can decompose $\vw$ into $n'$ (scaled) flat tails: call these tails $\vw_0,\dots,\vw_{n'-1}$. (W.l.o.g. assume that $\|\vw_{i+1}\|_{\infty}\le \|\vw_i\|_{\infty}/2$.)  Call a tail $\vw_i$ \textit{small} if $|\supp(\vw_i)|\le k/(\zeta^2\eta)$. Otherwise, call a tail $\vw_i$ \textit{large}. Note that since $\|\vw_0\|_{\infty}\le \zeta^5\eta^2\|\vz\|_2^2/k$, the total contribution of all small tails is at most
\[\sum_{i=0}^{n'-1}\frac{k}{\zeta^2\eta}\cdot\|\vw_0\|_{\infty}\cdot 2^{-i} \le 2\frac{k}{\zeta^2\eta}\cdot\|\vw_0\|_{\infty} \le 2\zeta^3\eta\|\vz\|_2^2.\]
Now consider a large tail $\vw_i$. Each item in $\supp(\vw_i)$ collides with a heavy hitter (independently) with probability at most $O(k/M^{1-\alpha})$ (even conditioned on the map of a heavy hitter under $f$). Thus, by Lemma~\ref{lem:chernoff}, except with probability 
\begin{equation}
\label{eq:corruption-prob2}
\left(\frac{ek|\supp(\vw_i)|}{\zeta^3\eta M^{-1\alpha}|\supp(\vw_i)|}\right)^{\Omega(\zeta^3\eta|\supp(\vw_i)|)}\le \left(\frac{ek}{\zeta^3\eta M^{1-\alpha}}\right)^{\Omega(\zeta k)},
\end{equation}
$\vw_i$ contributes at most $\zeta^3\eta\|\vw_i\|_2^2$ $\ell_2^2$ noise to the heavy hitters under $f$. Thus the total noise contribution of all the large tails is at most $\zeta^3\eta\|\vz\|_2^2$ except with probability $n'\cdot \left(\frac{\zeta^3\eta M^{1-\alpha}}{k}\right)^{-\zeta k}\le (M/k)^{-\Omega(\alpha\zeta k)}$, where the last inequality follows from the lower bound on $M$.

Thus, the total $\ell_2^2$ contribution of all the light tail elements to the heavy hitters under the map $f$ is at most $4\zeta^3\eta\|\vz\|_2^2$. Thus, by the Markov inequality at most $4\zeta k$ heavy hitters get corrupted by a light tail element.

Thus, except with probability $\left(\frac{M}{k}\right)^{\Omega(-\alpha\zeta k)}$, at most $5\zeta k$ heavy hitter get corrupted, as desired. 
\end{proof}

The rest of the proof is very similar to that of Theorem~\ref{thm:l2l2-exp} so we will just sketch the differences below. The main reason we can do this is because other than the analysis of the \badfour\ buckets,  we essentially are working with the $\ell_2^2$ tail. (Also the analysis of the \badfour\ buckets on depends on whether the $\ell_2^2$ noise in a bucket is large or not.)

To be more precise, define a vector related to $\vx$ where all but the heavy hitters are squared: i.e. $\vy=(y_1,\dots,y_N)$, where $y_i=x_i$ if $i\in H_k(\vx)$ otherwise $y_i=x_i^2$. Now define $f(\vy)$ vector to be the natural ``partial sum" vector of $\vy$ under $f$: i.e, $\vw=(w_1,\dots,w_M)=f(\vy)$ such that $w_j=\sum_{i:f(i)=j} y_i$. For the proof, call an item $j\in [M]$ a heavy hitter if it contains an uncorrupted heavy hitter from $\vx$ mapped under $f$ to it. Note that by Lemma~\ref{lem:intermediate}, we can assume that there are at least $k(1-4\zeta)$ heavy hitters in $\vw$. As before, define an item $j\in [M]$ to be a heavy tail item if $|w_j|>\zeta^2\eta\|\vz\|_2^2/k$. Note that we can have at most $\zeta^{-2}\eta^{-1}k+4\zeta k$ (where the last term can arise from the corrupted heavy hitters from $\vx$). Now the rest of the argument remains unchanged (by adjusting the constants), except for the following simple change-- in a bucket that is not \badone\, \badtwo\ or \badthree, we still have to account for the $\ell_2^2$ noise that an uncorrupted heavy hitter obtains due to the mapping $f$. However, this only increases the $\ell_2^2$ noise by a constant factor than what was handled earlier. Again by adjusting constants one can arrange for the claim in Lemma~\ref{lem:intermediate}. (One also has to add up the failure probabilities from Lemma~\ref{lem:corruption} and the one gets from the expander part of $\M$ but both are of the same order.)

\subsection{Some observations on the use of randomness}
\label{sec:rand-use-inter}

In this section, we make some observation on the use of randomness in the proof of Lemma~\ref{lem:intermediate} which were not already covered in Section~\ref{sec:rand-use-exp}.

The new place where one uses randomness in the proof of Lemma~\ref{lem:intermediate} is in Lemma~\ref{lem:corruption}. In particular, we used the randomness of the map $f$ in proving bounds (\ref{eq:corruption-prob1}) and (\ref{eq:corruption-prob2}). We note that both of these bounds would hold even if the map $f$ were $O(k)$-wise independent. Along with Remark~\ref{rem:lim-indep-exp}, this implies that
\begin{remark}
\label{rem:lim-indep-intermediate}
Lemma~\ref{lem:intermediate} hold even if the random $\pm 1$ entries are $O(k\ell)$-wise independent and the map $f$ is only $O(k)$-wise independent.
\end{remark}

Next we note that if we were proving the corresponding $\lolo$ result to Lemma~\ref{lem:intermediate}, then by Remark~\ref{rem:l1l1-exp} we do not need the random $\pm 1$ entries. There is another source of randomness-- the map $f$.  We will not get rid of the randomness used in the map $f$ (at least not quite yet). However, for future use it would be beneficial to note if one were to try and get rid of the randomness in proof of Lemma~\ref{lem:corruption}, how many events would one have to take a union bound against.

Let us first consider (\ref{eq:corruption-prob1}). In this case we assume that the locations of the at most $k'=\zeta^{-5}\eta^{-2}k+k$ ``heavy hitters" are fixed. Thus if we can do union bound over all $\binom{N}{k'}$ choices of these items, then we would be able to prove that there are no more than $\zeta k$ collisions probabilistically. Now let us consider (\ref{eq:corruption-prob2}). We first note that we only used randomness in trying to bound the corruption due to noise from light tail elements that come from $O(\log(N/k))$ large tails. In particular, we note that the argument assumes that the following are fixed: (i) The locations of the at most $k$ heavy hitters and (ii) assignment of the light tail elements to one of the $O(\log(N/k))$ large tails (or if it is not part of any large tail). The key observation is that we only care about the locations of elements in (i) and (ii) but not the actual {\em value} at those locations. Note that then the number of choices are $\binom{N}{k}+N^{O(\log(N/k))}$.
\begin{remark}
\label{rem:l1l1-intermediate}
The version of Lemma~\ref{lem:intermediate} for $\lolo$ sparse recovery holds even without multiplying the matrix $\M_G$ with random $\pm 1$ values. Further, one can also get rid of the randomness in the map $f$ if we are willing to take union bound against $\binom{N}{k'}+N^{x}$ events, where $k'=O(\zeta^{-5}\eta^{-2} k)$ and $x=O(\log(N/k))$.
\end{remark}

\section{Inverting a Random function}
\label{app:gen-invert}

We pick $f:[N]\rightarrow [N^2]$ randomly by picking a random degree $N-1$ polynomial over $\F_{N^2}$ and define $f(i)$ to be the evaluation of this polynomial at $i$. It is well-known that such a random hash function is $N$-wise independent (see e.g.~\cite{coding-book}) or fully independent.

We have three issues to tackle.

First, we now have to first compute the matrix $\M'$ with $N^2$ columns (then we pick the columns of $\M'$ indexed by $f(i)\in [N^2]$ for $i\in [N]$) instead of $N$ columns as we have used so far. This however, is not an issue since we have a family of matrices $\M_{n}$. This change results in some constants in the parameters of $\M^*$ changing but it does not affect the asymptotics.

Second, we have to show that this new $f$ is fine with respect to applying Lemma~\ref{lem:intermediate}.
In particular, we have to show the following: consider any node $v$ in the recursion tree. Then we will show that for the (new) $\phi_v:[N]\rightarrow [M]$ is $(d,\alpha)$-random (for some small $\alpha$). We will in fact show that these functions are $(d,0)$-random. Indeed, since $f$ is $N>(d+1)$-wise independent and the codes $\{C_n\}$ are uniform, the new map $\phi_v$ is also $(d+1)$-wise independent. (This can be proven e.g. by induction on the level of $v$.) This in turn implies that $\phi_v$ is $(d,0)$-random, which is sufficient for Lemma~\ref{lem:intermediate}.

Finally, we come to the inversion. Given $j\in [N]$, we would like to compute ``$f^{-1}(j)$." (Recall we do this step only once at the root of the recursion tree after the identification algorithm above has terminated with the output $I_w\subseteq [N^2]$. Recall that $|I_w|\le O(k/\eta)$.) Unfortunately, we cannot still guarantee that the inverse is unique. So we do the following: If we have $j\in I_w$ such that $|f^{-1}(j)|>1$, we just drop this index; otherwise we output $f^{-1}(j)$. We first argue that this step does not drop too many indices and then consider how quickly we can solve this step.

Recall that we only care about whether we identify elements of $H_{k'}(\vx)$, where $k'=O(\zeta^{-2}\eta^{-1} k)$.
Note that since $f$ is completely independent, for every $i\in H_{k'}(\vx)$ even conditioned on the value of $f(j)$ for every other $j\in [N]\setminus \{i\}$, $f(i)$ is completely independent. Thus, the probability that $f(i)=f(j)$ for one of these $j$ is at most $N/N^2=1/N$. Thus, by the argument similar to the one in~\cite{ely-martin} (to care of the dependencies), we can use Lemma~\ref{lem:chernoff} to argue that the probability of more than $\zeta k$ $i$'s colliding with another $j$ is at most $(N/k)^{-\Omega(\zeta k)}$. Thus, the algorithm above loses an extra $\zeta k$ indices while maintaining the same failure probability. This extra $\zeta k$ additive factor only changes the constants and thus, can be absorbed into the analysis without much trouble.

We now come to the part about computing $f^{-1}(j)$. To do this step, we just store the pairs $(f(i),i)_{i\in [N]}$ in an array (sorted by the first entry)  of size $O(N\log{N})$ bits. Note that given this array (by binary search), we can in time $O(\log{N})$, given a $j\in [N^2]$, figure out $f^{-1}(j)$ (if it is unique). Since we have to do this inversion $O(k/\eta)$ times, we add an additional factor of $O(k/\eta\log{N})$ to the identification time. (This additive factor will never be asymptotically significant in our final results.)

We will call this scheme above \schemeo. Note that we have show that \schemeo\ works and
\begin{lemma}
\label{lem:scheme1}
\schemeo\ adds $O(k/\eta\log(N/k))$ to the decoding time in Theorem~\ref{thm:recursive-exp} and needs $O(N\log{N})$ bits of space.
\end{lemma}

We also need $O(\zeta^{-5}\eta^{-2}\cdot k\cdot \log^2{N})$ bits of space
to store the  randomness needed to define $g$ (which we need to store). However, this is subsumed by the $O(N\log{N})$ space to store the array above.

\end{document}

%% file: overview.pdf_t
\begin{picture}(0,0)%
\includegraphics{overview.pdf}%
\end{picture}%
\setlength{\unitlength}{3947sp}%
\begingroup\makeatletter\ifx\SetFigFont\undefined%
\gdef\SetFigFont#1#2#3#4#5{%
  \reset@font\fontsize{#1}{#2pt}%
  \fontfamily{#3}\fontseries{#4}\fontshape{#5}%
  \selectfont}%
\fi\endgroup%
\begin{picture}(12795,9849)(268,-9073)
\put(826,239){\makebox(0,0)[lb]{\smash{{\SetFigFont{12}{14.4}{\rmdefault}{\mddefault}{\updefault}{\color[rgb]{0,0,0}$x_0$}%
}}}}
\put(826,-211){\makebox(0,0)[lb]{\smash{{\SetFigFont{12}{14.4}{\rmdefault}{\mddefault}{\updefault}{\color[rgb]{0,0,0}$x_1$}%
}}}}
\put(826,-661){\makebox(0,0)[lb]{\smash{{\SetFigFont{12}{14.4}{\rmdefault}{\mddefault}{\updefault}{\color[rgb]{0,0,0}$x_2$}%
}}}}
\put(826,-1111){\makebox(0,0)[lb]{\smash{{\SetFigFont{12}{14.4}{\rmdefault}{\mddefault}{\updefault}{\color[rgb]{0,0,0}$x_3$}%
}}}}
\put(826,-1561){\makebox(0,0)[lb]{\smash{{\SetFigFont{12}{14.4}{\rmdefault}{\mddefault}{\updefault}{\color[rgb]{0,0,0}$x_4$}%
}}}}
\put(826,-2011){\makebox(0,0)[lb]{\smash{{\SetFigFont{12}{14.4}{\rmdefault}{\mddefault}{\updefault}{\color[rgb]{0,0,0}$x_5$}%
}}}}
\put(826,-2461){\makebox(0,0)[lb]{\smash{{\SetFigFont{12}{14.4}{\rmdefault}{\mddefault}{\updefault}{\color[rgb]{0,0,0}$x_6$}%
}}}}
\put(826,-2911){\makebox(0,0)[lb]{\smash{{\SetFigFont{12}{14.4}{\rmdefault}{\mddefault}{\updefault}{\color[rgb]{0,0,0}$x_7$}%
}}}}
\put(1951,-5611){\makebox(0,0)[lb]{\smash{{\SetFigFont{12}{14.4}{\rmdefault}{\mddefault}{\updefault}{\color[rgb]{0,0,0}Lemma~\ref{lem:intermediate}}%
}}}}
\put(10576,-5611){\makebox(0,0)[lb]{\smash{{\SetFigFont{12}{14.4}{\rmdefault}{\mddefault}{\updefault}{\color[rgb]{0,0,0}Lemma~\ref{lem:intermediate}}%
}}}}
\put(6151,-6736){\makebox(0,0)[lb]{\smash{{\SetFigFont{14}{16.8}{\rmdefault}{\mddefault}{\updefault}{\color[rgb]{0,0,0}$+$}%
}}}}
\put(6151,-7561){\makebox(0,0)[lb]{\smash{{\SetFigFont{14}{16.8}{\rmdefault}{\mddefault}{\updefault}{\color[rgb]{0,0,0}$+$}%
}}}}
\put(6151,-8236){\makebox(0,0)[lb]{\smash{{\SetFigFont{14}{16.8}{\rmdefault}{\mddefault}{\updefault}{\color[rgb]{0,0,0}$+$}%
}}}}
\put(6151,-8911){\makebox(0,0)[lb]{\smash{{\SetFigFont{14}{16.8}{\rmdefault}{\mddefault}{\updefault}{\color[rgb]{0,0,0}$+$}%
}}}}
\put(10651,-6736){\makebox(0,0)[lb]{\smash{{\SetFigFont{14}{16.8}{\rmdefault}{\mddefault}{\updefault}{\color[rgb]{0,0,0}$+$}%
}}}}
\put(10651,-7561){\makebox(0,0)[lb]{\smash{{\SetFigFont{14}{16.8}{\rmdefault}{\mddefault}{\updefault}{\color[rgb]{0,0,0}$+$}%
}}}}
\put(10651,-8236){\makebox(0,0)[lb]{\smash{{\SetFigFont{14}{16.8}{\rmdefault}{\mddefault}{\updefault}{\color[rgb]{0,0,0}$+$}%
}}}}
\put(10651,-8911){\makebox(0,0)[lb]{\smash{{\SetFigFont{14}{16.8}{\rmdefault}{\mddefault}{\updefault}{\color[rgb]{0,0,0}$+$}%
}}}}
\put(1951,-6736){\makebox(0,0)[lb]{\smash{{\SetFigFont{14}{16.8}{\rmdefault}{\mddefault}{\updefault}{\color[rgb]{0,0,0}$+$}%
}}}}
\put(1951,-7561){\makebox(0,0)[lb]{\smash{{\SetFigFont{14}{16.8}{\rmdefault}{\mddefault}{\updefault}{\color[rgb]{0,0,0}$+$}%
}}}}
\put(1951,-8236){\makebox(0,0)[lb]{\smash{{\SetFigFont{14}{16.8}{\rmdefault}{\mddefault}{\updefault}{\color[rgb]{0,0,0}$+$}%
}}}}
\put(1951,-8911){\makebox(0,0)[lb]{\smash{{\SetFigFont{14}{16.8}{\rmdefault}{\mddefault}{\updefault}{\color[rgb]{0,0,0}$+$}%
}}}}
\put(2776,239){\makebox(0,0)[lb]{\smash{{\SetFigFont{12}{14.4}{\rmdefault}{\mddefault}{\updefault}{\color[rgb]{0,0,0}$x_1$}%
}}}}
\put(2776,-211){\makebox(0,0)[lb]{\smash{{\SetFigFont{12}{14.4}{\rmdefault}{\mddefault}{\updefault}{\color[rgb]{0,0,0}$x_3$}%
}}}}
\put(2776,-661){\makebox(0,0)[lb]{\smash{{\SetFigFont{12}{14.4}{\rmdefault}{\mddefault}{\updefault}{\color[rgb]{0,0,0}$x_2$}%
}}}}
\put(2776,-1111){\makebox(0,0)[lb]{\smash{{\SetFigFont{12}{14.4}{\rmdefault}{\mddefault}{\updefault}{\color[rgb]{0,0,0}$x_2$}%
}}}}
\put(2776,-1561){\makebox(0,0)[lb]{\smash{{\SetFigFont{12}{14.4}{\rmdefault}{\mddefault}{\updefault}{\color[rgb]{0,0,0}$x_5$}%
}}}}
\put(2776,-2011){\makebox(0,0)[lb]{\smash{{\SetFigFont{12}{14.4}{\rmdefault}{\mddefault}{\updefault}{\color[rgb]{0,0,0}$x_4$}%
}}}}
\put(2776,-2461){\makebox(0,0)[lb]{\smash{{\SetFigFont{12}{14.4}{\rmdefault}{\mddefault}{\updefault}{\color[rgb]{0,0,0}$x_7$}%
}}}}
\put(2776,-2911){\makebox(0,0)[lb]{\smash{{\SetFigFont{12}{14.4}{\rmdefault}{\mddefault}{\updefault}{\color[rgb]{0,0,0}$x_6$}%
}}}}
\put(6826,239){\makebox(0,0)[lb]{\smash{{\SetFigFont{12}{14.4}{\rmdefault}{\mddefault}{\updefault}{\color[rgb]{0,0,0}$x_1$}%
}}}}
\put(6826,-211){\makebox(0,0)[lb]{\smash{{\SetFigFont{12}{14.4}{\rmdefault}{\mddefault}{\updefault}{\color[rgb]{0,0,0}$x_3$}%
}}}}
\put(6826,-661){\makebox(0,0)[lb]{\smash{{\SetFigFont{12}{14.4}{\rmdefault}{\mddefault}{\updefault}{\color[rgb]{0,0,0}$x_4$}%
}}}}
\put(6826,-1111){\makebox(0,0)[lb]{\smash{{\SetFigFont{12}{14.4}{\rmdefault}{\mddefault}{\updefault}{\color[rgb]{0,0,0}$x_2$}%
}}}}
\put(6826,-1561){\makebox(0,0)[lb]{\smash{{\SetFigFont{12}{14.4}{\rmdefault}{\mddefault}{\updefault}{\color[rgb]{0,0,0}$x_5$}%
}}}}
\put(6826,-2011){\makebox(0,0)[lb]{\smash{{\SetFigFont{12}{14.4}{\rmdefault}{\mddefault}{\updefault}{\color[rgb]{0,0,0}$x_4$}%
}}}}
\put(6826,-2461){\makebox(0,0)[lb]{\smash{{\SetFigFont{12}{14.4}{\rmdefault}{\mddefault}{\updefault}{\color[rgb]{0,0,0}$x_7$}%
}}}}
\put(6826,-2911){\makebox(0,0)[lb]{\smash{{\SetFigFont{12}{14.4}{\rmdefault}{\mddefault}{\updefault}{\color[rgb]{0,0,0}$x_6$}%
}}}}
\put(676,-5761){\makebox(0,0)[lb]{\smash{{\SetFigFont{12}{14.4}{\rmdefault}{\mddefault}{\updefault}{\color[rgb]{0,0,0}$x_1$}%
}}}}
\put(676,-6211){\makebox(0,0)[lb]{\smash{{\SetFigFont{12}{14.4}{\rmdefault}{\mddefault}{\updefault}{\color[rgb]{0,0,0}$x_3$}%
}}}}
\put(676,-6661){\makebox(0,0)[lb]{\smash{{\SetFigFont{12}{14.4}{\rmdefault}{\mddefault}{\updefault}{\color[rgb]{0,0,0}$x_4$}%
}}}}
\put(676,-7111){\makebox(0,0)[lb]{\smash{{\SetFigFont{12}{14.4}{\rmdefault}{\mddefault}{\updefault}{\color[rgb]{0,0,0}$x_2$}%
}}}}
\put(676,-7561){\makebox(0,0)[lb]{\smash{{\SetFigFont{12}{14.4}{\rmdefault}{\mddefault}{\updefault}{\color[rgb]{0,0,0}$x_5$}%
}}}}
\put(676,-8011){\makebox(0,0)[lb]{\smash{{\SetFigFont{12}{14.4}{\rmdefault}{\mddefault}{\updefault}{\color[rgb]{0,0,0}$x_4$}%
}}}}
\put(676,-8461){\makebox(0,0)[lb]{\smash{{\SetFigFont{12}{14.4}{\rmdefault}{\mddefault}{\updefault}{\color[rgb]{0,0,0}$x_7$}%
}}}}
\put(676,-8911){\makebox(0,0)[lb]{\smash{{\SetFigFont{12}{14.4}{\rmdefault}{\mddefault}{\updefault}{\color[rgb]{0,0,0}$x_6$}%
}}}}
\put(4876,-5761){\makebox(0,0)[lb]{\smash{{\SetFigFont{12}{14.4}{\rmdefault}{\mddefault}{\updefault}{\color[rgb]{0,0,0}$x_1$}%
}}}}
\put(4876,-6211){\makebox(0,0)[lb]{\smash{{\SetFigFont{12}{14.4}{\rmdefault}{\mddefault}{\updefault}{\color[rgb]{0,0,0}$x_3$}%
}}}}
\put(4876,-6661){\makebox(0,0)[lb]{\smash{{\SetFigFont{12}{14.4}{\rmdefault}{\mddefault}{\updefault}{\color[rgb]{0,0,0}$x_4$}%
}}}}
\put(4876,-7111){\makebox(0,0)[lb]{\smash{{\SetFigFont{12}{14.4}{\rmdefault}{\mddefault}{\updefault}{\color[rgb]{0,0,0}$x_2$}%
}}}}
\put(4876,-7561){\makebox(0,0)[lb]{\smash{{\SetFigFont{12}{14.4}{\rmdefault}{\mddefault}{\updefault}{\color[rgb]{0,0,0}$x_5$}%
}}}}
\put(4876,-8011){\makebox(0,0)[lb]{\smash{{\SetFigFont{12}{14.4}{\rmdefault}{\mddefault}{\updefault}{\color[rgb]{0,0,0}$x_4$}%
}}}}
\put(4876,-8461){\makebox(0,0)[lb]{\smash{{\SetFigFont{12}{14.4}{\rmdefault}{\mddefault}{\updefault}{\color[rgb]{0,0,0}$x_7$}%
}}}}
\put(4876,-8911){\makebox(0,0)[lb]{\smash{{\SetFigFont{12}{14.4}{\rmdefault}{\mddefault}{\updefault}{\color[rgb]{0,0,0}$x_6$}%
}}}}
\put(9376,-5686){\makebox(0,0)[lb]{\smash{{\SetFigFont{12}{14.4}{\rmdefault}{\mddefault}{\updefault}{\color[rgb]{0,0,0}$x_1$}%
}}}}
\put(9376,-6136){\makebox(0,0)[lb]{\smash{{\SetFigFont{12}{14.4}{\rmdefault}{\mddefault}{\updefault}{\color[rgb]{0,0,0}$x_3$}%
}}}}
\put(9376,-6586){\makebox(0,0)[lb]{\smash{{\SetFigFont{12}{14.4}{\rmdefault}{\mddefault}{\updefault}{\color[rgb]{0,0,0}$x_4$}%
}}}}
\put(9376,-7036){\makebox(0,0)[lb]{\smash{{\SetFigFont{12}{14.4}{\rmdefault}{\mddefault}{\updefault}{\color[rgb]{0,0,0}$x_2$}%
}}}}
\put(9376,-7486){\makebox(0,0)[lb]{\smash{{\SetFigFont{12}{14.4}{\rmdefault}{\mddefault}{\updefault}{\color[rgb]{0,0,0}$x_5$}%
}}}}
\put(9376,-7936){\makebox(0,0)[lb]{\smash{{\SetFigFont{12}{14.4}{\rmdefault}{\mddefault}{\updefault}{\color[rgb]{0,0,0}$x_4$}%
}}}}
\put(9376,-8386){\makebox(0,0)[lb]{\smash{{\SetFigFont{12}{14.4}{\rmdefault}{\mddefault}{\updefault}{\color[rgb]{0,0,0}$x_7$}%
}}}}
\put(9376,-8836){\makebox(0,0)[lb]{\smash{{\SetFigFont{12}{14.4}{\rmdefault}{\mddefault}{\updefault}{\color[rgb]{0,0,0}$x_6$}%
}}}}
\put(5926,-4186){\makebox(0,0)[lb]{\smash{{\SetFigFont{12}{14.4}{\rmdefault}{\mddefault}{\updefault}{\color[rgb]{0,0,0}Lemma~\ref{lem:lw}}%
}}}}
\put(6151,-5611){\makebox(0,0)[lb]{\smash{{\SetFigFont{12}{14.4}{\rmdefault}{\mddefault}{\updefault}{\color[rgb]{0,0,0}Lemma~\ref{lem:intermediate}}%
}}}}
\put(10276,-1336){\makebox(0,0)[lb]{\smash{{\SetFigFont{12}{14.4}{\rmdefault}{\mddefault}{\updefault}{\color[rgb]{0,0,0}Theorem~\ref{thm:l2l2-exp-identification}}%
}}}}
\put(1651,-1111){\makebox(0,0)[lb]{\smash{{\SetFigFont{12}{14.4}{\rmdefault}{\mddefault}{\updefault}{\color[rgb]{0,0,0}$f$}%
}}}}
\put(1126,-3661){\makebox(0,0)[lb]{\smash{{\SetFigFont{12}{14.4}{\rmdefault}{\mddefault}{\updefault}{\color[rgb]{0,0,0}$f^{-1}$}%
}}}}
\put(751,-4111){\makebox(0,0)[lb]{\smash{{\SetFigFont{12}{14.4}{\rmdefault}{\mddefault}{\updefault}{\color[rgb]{0,0,0}Section~\ref{sec:invert}}%
}}}}
\end{picture}%